\documentclass[11pt]{article}
\usepackage[margin=1in]{geometry}

\usepackage{hyperref}

\usepackage[utf8]{inputenc} 
\usepackage{booktabs}       
\usepackage{amsfonts}       
\usepackage{nicefrac}       
\usepackage{microtype}      
\usepackage{xcolor}         
\usepackage{parskip}        
\usepackage{amsmath}
\usepackage{amsthm}
\usepackage{mathtools}
\usepackage{dsfont}
\usepackage{txfonts} 
\usepackage{color,graphicx}
\usepackage{url}
\usepackage{enumerate}
\usepackage{algorithm}
\usepackage[noend]{algpseudocode}
\usepackage{fancybox}
\usepackage{xspace}
\usepackage{cuted,tcolorbox,lipsum}
\usepackage{multicol}
\usepackage{xcolor}
\usepackage[T1]{fontenc}
\usepackage{tikz}
\usepackage{tcolorbox}
\usetikzlibrary{calc}
\tcbuselibrary{skins}
\usepackage[colorinlistoftodos]{todonotes}
\usepackage{array,multirow}
\usepackage{roboto}
\usepackage[outline]{contour}
\usepackage{framed}
\usepackage{mdframed}

\newcommand{\REAL}{\ensuremath{\mathbb{R}}}

\usepackage{color-edits}

\newcommand{\algorithmicbreak}{\textbf{break}}
\newcommand{\Break}{\algorithmicbreak}

\newtheorem{theorem}{Theorem}
\newtheorem{corollary}[theorem]{Corollary}
\newtheorem{lemma}[theorem]{Lemma}

\newtheorem{fact}[theorem]{Fact}
\newtheorem{definition}{Definition}[section]
\newtheorem{claim}{Claim}[section]

\newcommand{\ceil}[1]{\left\lceil #1 \right\rceil}
\newcommand{\floor}[1]{\left\lfloor #1 \right\rfloor}

\colorlet{shadecolor}{gray!15}
\definecolor{lightgray}{gray}{0.25}

\newcommand{\mO}{O}
\newcommand{\mI}{{\mathcal{I}}}
\newcommand{\wtilde}{\widetilde}

\renewcommand{\Pr}[1]{\ensuremath{\mathbb{P}\left[#1\right]}}
\newcommand{\Ex}[1]{\ensuremath{\mathbb{E}\left[#1\right]}}
\newcommand{\Exu}[2]{\ensuremath{\mathbb{E}_{#1}\left[#2\right]}}
\newcommand{\ind}[1]{\ensuremath{\mathds{1}\left[#1\right]}}

\DeclareMathOperator{\poly}{poly}

\DeclareMathOperator{\argmin}{arg\,min}
\newcommand{\eps}{\varepsilon}
\newcommand{\err}{{\textsc{Fail}}}
\newcommand{\matroid}{\mathcal{M}(\ground,\mathcal{I})}

\newcommand{\constLevel}{\textsc{ConstructLevel}}
\newcommand{\update}{\textsc{Update}}
\newcommand{\suit}{\textsc{Promote}}
\newcommand{\boolsuit}{\textsc{BoolPromote}}
\newcommand{\init}{\textsc{Init}}
\newcommand{\dynamicmatroid}{\textsc{DynamicMatroid}}
\newcommand{\matroidleveling}{\textsc{MatroidLeveling}}
\newcommand{\matroidupdates}{\textsc{MatroidUpdates}}
\newcommand{\carupdates}{\textsc{CardinalityConstraintUpdates}}
\newcommand{\levelingconstraint}{\textsc{CardinalityConstraintLeveling}}
\newcommand{\MatroidConstLevel}{\textsc{MatroidConstructLevel}}

\newcommand{\replacementTester}{\textsc{Promote}}

\newcommand{\insertv}{{\textsc{Insert}}}
\newcommand{\deletev}{{\textsc{Delete}}}

\newcommand{\ground}{\ensuremath{\mathcal{V}}}
\newcommand{\bR}{\ensuremath{\mathbf{R}}}

\newcommand{\bE}{\ensuremath{\mathbf{e}}}
\newcommand{\bT}{\ensuremath{\mathbf{T}}}

\newcommand{\bM}{\ensuremath{\mathbf{M}}}

\newcommand{\bH}{\ensuremath{\mathbf{H}}}
\newcommand{\bY}{\ensuremath{\mathbf{Y}}}

\newcommand{\spn}{{\textsc{Span}}}
\newcommand{\rank}{rank}

\definecolor{ShadowColor}{RGB}{30,150,190}

\makeatletter
\newcommand\Cshadowbox{\VerbBox\@Cshadowbox}
\def\@Cshadowbox#1{
  \setbox\@fancybox\hbox{\fbox{#1}}
  \leavevmode\vbox{
    \offinterlineskip
    \dimen@=0.7\shadowsize
    \advance\dimen@ .5\fboxrule
    \hbox{\copy\@fancybox\kern.5\fboxrule\lower\shadowsize\hbox{
      \color{ShadowColor}\vrule \@height\ht\@fancybox \@depth\dp\@fancybox \@width\dimen@}}
    \vskip\dimexpr-\dimen@+0.5\fboxrule\relax
    \moveright\shadowsize\vbox{
      \color{ShadowColor}\hrule \@width\wd\@fancybox \@height\dimen@}}}
\makeatother

\newcommand{\marginalgain}[2]{f(#2 + #1) - f(#2)}

\title{
\Large
Dynamic Algorithms for  Matroid Submodular Maximization }

\author{ 
Kiarash Banihashem\thanks{Computer Science Department, University of Maryland, College Park, MD, USA. {\tt kiarash@umd.edu}.}
\and
Leyla Biabani\thanks{Department of Mathematics and Computer Science, TU Eindhoven, the Netherlands. {\tt l.biabani@tue.nl}.}
\and
Samira Goudarzi \thanks{Computer Science Department, University of Maryland, College Park, MD, USA. {\tt samirag@umd.edu}. }
\and
MohammadTaghi Hajiaghayi\thanks{Computer Science Department, University of Maryland, College Park, MD, USA. {\tt hajiagha@cs.umd.edu}.}
\and
Peyman Jabbarzade\thanks{Computer Science Department, University of Maryland, College Park, MD, USA. {\tt peymanj@umd.edu}.}
\and
Morteza Monemizadeh\thanks{Department of Mathematics and Computer Science, TU Eindhoven, the Netherlands. {\tt m.monemizadeh@tue.nl}.}
}

\date{\vspace{-5ex}}

\begin{document}
\maketitle

\begin{abstract}

Submodular maximization under matroid and cardinality  
constraints are classical problems with a wide range of applications in machine learning, 
auction theory, and combinatorial optimization. 
In this paper, we consider these problems in the dynamic setting where 
(1) we have oracle access to a monotone submodular function $f: 2^{\ground} \rightarrow \mathbb{R}^+$ and 
(2) we are given a sequence $\mathcal{S}$ of insertions and deletions of elements of an underlying ground set $\ground$.

We develop the first fully dynamic algorithm for the submodular maximization problem under the matroid constraint that maintains a $(4+\epsilon)$-approximation solution ($0 < \epsilon \le 1$) using an expected query complexity of $O(k\log(k)\log^3{(k/\epsilon)})$ 
, which is indeed parameterized by the rank $k$ of the matroid $\matroid$ as well.

Chen and Peng~\cite{DBLP:journals/corr/abs-2111-03198} at STOC'22 
studied the complexity of this problem in the insertion-only dynamic model (a restricted version of the fully dynamic model 
where deletion is not allowed), and they raised the following important open question: 
\emph{"for fully dynamic streams [sequences of 
insertions and deletions of elements], there is no known constant-factor approximation 
algorithm with poly(k) amortized queries for matroid constraints." }
Our dynamic algorithm answers this question as well as an open problem of Lattanzi et al.~\cite{DBLP:conf/nips/LattanziMNTZ20} (NeurIPS'20) affirmatively. 

As a byproduct, for the submodular maximization under the cardinality  constraint $k$, 
we propose a parameterized (by the cardinality constraint $k$) dynamic algorithm that maintains a $(2+\epsilon)$-approximate solution of the sequence $\mathcal{S}$ at any time $t$ using an expected query complexity of $\mO(k\epsilon^{-1}\log^2(k))$, which is an improvement upon the dynamic algorithm that Monemizadeh~\cite{DBLP:conf/nips/Monemizadeh20} (NeurIPS'20) 
developed for this problem using an expected query complexity $O(k^2\epsilon^{-3}\log^5(n))$. 
In particular, this dynamic algorithm is the first one for this problem whose query complexity is independent of the size of ground set $\ground$ (i.e., $n = |\ground|$). 

{We develop our dynamic algorithm for the submodular maximization problem under the matroid or cardinality constraint by designing a randomized leveled data structure that supports insertion and deletion operations, maintaining an approximate solution for the given problem. In addition, we develop a fast construction algorithm for our data structure that uses a one-pass over a random permutation of the elements and utilizes monotonicity property of our problems which has a subtle proof in the matroid case. We believe these techniques could also be useful for other optimization problems in the area of dynamic algorithms.}

\end{abstract}

\section{Introduction}

\emph{Submodularity} is a fundamental notion
that arises in many applications such as image segmentation, data summarization~\cite{DBLP:conf/aaai/KumariB21,DBLP:journals/jmlr/SchreiberBN20}, 
RNA and protein sequencing~\cite{DBLP:conf/bcb/YangBN20,DBLP:conf/bcb/LibbrechtBN18}
hypothesis identification~\cite{DBLP:journals/pami/BarinovaLK12,DBLP:conf/icml/ChenSMKWK14}, 
information gathering~\cite{DBLP:conf/aaai/RadanovicS0F18}, and
social networks~\cite{DBLP:conf/kdd/KempeKT03}. 
A function $f: 2^\ground \rightarrow \REAL^+$  is called \emph{submodular}  if for all $A \subseteq B \subseteq \ground$ and $e \notin B$, 
it satisfies $f(A \cup \{e\}) -f(A) \ge f(B \cup \{e\}) -f(B)$, and it is called \emph{monotone} if for every $A \subseteq B$, it satisfies $f(A) \leq f(B)$.

Given a monotone submodular function  $f: 2^{\ground} \rightarrow \REAL^+$
that is defined over 
a ground set $\ground$ and a parameter $k \in \mathbb{N}$, 
in the \emph{submodular maximization problem under the cardinality constraint $k$}, 
we would like to report a set $I^* \subseteq \ground$ of size at most $k$ whose 
submodular value is maximum among all subsets of $\ground$ of size at most $k$.

\emph{Matroid} \cite{DBLP:books/daglib/0070636} is a basic branch of mathematics 
that generalizes the notion of linear independence in vector spaces and 
has basic links to linear algebra~\cite{10.2307/24901346}, graphs~\cite{DBLP:journals/mp/Edmonds71}, lattices~\cite{Maeda1970}, codes~\cite{4557271}, transversals~\cite{Edmonds1965TransversalsAM}, and projective geometries~\cite{10.2307/2371070}. 
A matroid $\matroid$ consists of a \emph{ground set} $\ground$ and  
a nonempty downward-closed set system $\mathcal{I} \subseteq 2^{V}$  (known as the independent sets)
that satisfies the \emph{exchange axiom}: 
for every pair of independent sets $A,B \in \mathcal{I}$ such that $|A| < |B| $, 
there exists an element $x\in B\backslash A$ such that $A \cup \{x\} \in \mathcal{I}$.

A growing interest in machine learning~\cite{DBLP:journals/spm/TohidiACGLK20,DBLP:conf/nips/HanCCW20,DBLP:conf/nips/ElenbergDFK17,DBLP:conf/bmvc/Krause13,DBLP:conf/icml/BateniCEFMR19,DBLP:conf/acl/LinB11,DBLP:conf/cikm/SiposSSJ12,DBLP:conf/kdd/El-AriniG11,DBLP:conf/wsdm/AgrawalGHI09,DBLP:conf/icml/WeiIB15,DBLP:conf/iccv/DueckF07}, online auction theory~\cite{DBLP:journals/talg/BateniHZ13,DBLP:conf/wine/GuptaRST10,DBLP:journals/jacm/BabaioffIKK18,DBLP:conf/stoc/KleinbergW12,DBLP:journals/algorithmica/GharanV13,DBLP:journals/geb/KleinbergW19,DBLP:journals/jacm/BabaioffIKK18,DBLP:conf/stoc/KleinbergW12,DBLP:conf/soda/EhsaniHKS18}, and 
combinatorial optimization~\cite{DBLP:journals/mor/LeeSV10,DBLP:conf/soda/ChekuriVZ11,DBLP:conf/stoc/LeeSV10,DBLP:journals/toc/KempeKT15,DBLP:conf/uai/KvetonWAEE14,DBLP:books/ph/PapadimitriouS82,DBLP:journals/jcp/MagosMP06}
is to study the problem of maximizing a monotone submodular function $f: 2^{V} \rightarrow \REAL^+$ under a matroid $\matroid$ constraint. 
In particular, the goal in the \emph{submodular maximization problem under the matroid constraint} is to 
return an independent set $I^* \in \mathcal{I}$ 
of the maximum submodular value $f(I^*) $ among all independent sets in $\mathcal{I}$.

The seminal work of Fisher, Nemhauser and Wolsey~\cite{DBLP:journals/mp/NemhauserWF78} in the 1970s,
was the first that considered the submodular maximization problem under the matroid constraint problem in the offline model. 
Indeed, they developed a simple, elegant leveling algorithm for this problem that 
in time $O(nk)$ (where $k$ is the rank of the matroid $\matroid$), 
returns an independent set whose submodular value is a $2$-approximation of 
the optimal value $OPT = \max_{I^* \in \mathcal{I}} f(I^*)$.

However, despite the simplicity and optimality of this celebrated algorithm, 
there has been a surge of recent research efforts to
reexamine these problems under a variety of massive data models motivated by the unique challenges of
working with massive datasets. These include streaming 
algorithms~\cite{DBLP:conf/kdd/BadanidiyuruMKK14,DBLP:conf/stoc/FeldmanNSZ20,DBLP:conf/icalp/AlalufEFNS20,DBLP:conf/icml/0001MZLK19}, 
dynamic algorithms~\cite{DBLP:conf/icml/MirzasoleimanK017,DBLP:conf/icml/0001ZK18,DBLP:conf/nips/LattanziMNTZ20,DBLP:conf/nips/Monemizadeh20,DBLP:journals/corr/abs-2111-03198}, 
sublinear time algorithms~\cite{DBLP:conf/icml/StanZ0K17}, parallel
algorithms~\cite{DBLP:conf/nips/KupferQBS20,DBLP:conf/stoc/BalkanskiS18,DBLP:conf/icml/BalkanskiS18,DBLP:conf/icml/EneN20,DBLP:conf/soda/EneN19,DBLP:conf/stoc/EneNV19,DBLP:conf/stoc/ChekuriQ19}, online algorithms~\cite{DBLP:conf/nips/HarveyLS20}, 
private algorithms~\cite{DBLP:conf/aaai/ChaturvediNZ21}, learning algorithms~\cite{DBLP:conf/pkdd/BalcanH12,DBLP:conf/stoc/BalcanH11,DBLP:conf/stoc/BalkanskiRS17} and distributed algorithms~\cite{DBLP:conf/stoc/MirrokniZ15,DBLP:conf/nips/EneNV17,DBLP:conf/focs/BarbosaENW16,DBLP:conf/icml/BarbosaENW15}. 

Among these big data models, the \emph{(fully) dynamic model}~\cite{DBLP:conf/focs/Rauch92,DBLP:journals/jacm/HenzingerK99} has been of 
particular interest recently. 
In this model, we see a sequence $\mathcal{S}$ of updates (i.e., inserts and deletes) of elements of an underlying structure (such as a graph, matrix, and so on), and the goal is to maintain an approximate or exact solution of a problem that is defined for that structure using a fast update time. 
For example, the influential work of Onak and Rubinfeld~\cite{DBLP:conf/stoc/OnakR10}(STOC'10) studied dynamic versions of the matching and the vertex cover problems. Some other new advances in the dynamic model have been seen by developing dynamic 
algorithms for matching and vertex cover~\cite{DBLP:conf/stoc/OnakR10,DBLP:conf/soda/BhattacharyaHN17,DBLP:conf/icalp/BernsteinS15,DBLP:conf/focs/Solomon16,DBLP:journals/talg/NeimanS16,DBLP:conf/icalp/CharikarS18,DBLP:conf/stoc/BernsteinDL21,DBLP:journals/talg/BernsteinFH21,DBLP:conf/icalp/BhattacharyaK21}, graph connectivity~\cite{DBLP:conf/soda/KapronKM13,DBLP:conf/soda/AhnGM12}, 
graph sparsifiers~\cite{DBLP:conf/icalp/BernsteinBGNSS022,DBLP:conf/focs/AbrahamDKKP16,DBLP:conf/stoc/DurfeeGGP19,DBLP:conf/focs/ChuzhoyGLNPS20,DBLP:conf/focs/ChenGHPS20,DBLP:conf/focs/GaoLP21,DBLP:conf/stoc/BrandGJLLPS22}, set cover~\cite{DBLP:conf/soda/BhattacharyaHNW21,DBLP:conf/focs/GuptaL20,DBLP:conf/focs/BhattacharyaHN19,DBLP:conf/stoc/GuptaK0P17,DBLP:conf/focs/GuptaL20,DBLP:conf/stoc/AbboudA0PS19}, 
approximate shortest paths~\cite{DBLP:conf/focs/Bernstein09,DBLP:conf/focs/HenzingerKN13,DBLP:phd/us/Bernstein16,DBLP:reference/algo/Bernstein16a,DBLP:journals/siamcomp/HenzingerKN16,DBLP:conf/focs/BrandN19,DBLP:reference/algo/HenzingerKN16,DBLP:conf/soda/AbrahamCK17}, minimum spanning forests~\cite{DBLP:conf/focs/NanongkaiSW17,DBLP:conf/stoc/NanongkaiS17,DBLP:conf/soda/ChechikZ20}, densest subgraphs~\cite{DBLP:conf/stoc/BhattacharyaHNT15,DBLP:conf/stoc/SawlaniW20},  maximal independent sets~\cite{DBLP:conf/stoc/AssadiOSS18,DBLP:conf/focs/BehnezhadDHSS19,DBLP:conf/focs/ChechikZ19}, spanners~\cite{DBLP:journals/talg/BernsteinFH21,DBLP:conf/esa/Baswana06,DBLP:conf/esa/BodwinK16,DBLP:journals/talg/BaswanaKS12}, 
and graph coloring~\cite{DBLP:journals/talg/SolomonW20,DBLP:journals/talg/BhattacharyaGKL22}, to name a few\footnote{
Interestingly, the best paper awards at SODA'23 were awarded to 
two dynamic algorithms~\cite{DBLP:journals/corr/abs-2207-07607,DBLP:journals/corr/abs-2207-07438} for the matching size problem 
in the dynamic model.}. 

However, for the very basic problem of submodular maximization under the matroid constraint, there is no (fully) dynamic algorithm known. This problem was repeatedly  posed as an open problem at 
STOC'22~\cite{DBLP:journals/corr/abs-2111-03198} as well as NeurIPS'20~\cite{DBLP:conf/nips/LattanziMNTZ20}. 
Indeed, Chen and Peng~\cite{DBLP:journals/corr/abs-2111-03198}(STOC'22) raised the following open question:

\begin{tcolorbox}[width=\linewidth, colback=white!80!gray,boxrule=0pt,frame hidden, sharp corners]
\textbf{Open question~\cite{DBLP:journals/corr/abs-2111-03198,DBLP:conf/nips/LattanziMNTZ20}:} 
"For fully dynamic streams [sequences of 
insertions and deletions of elements], there is no known constant-factor approximation 
algorithm with poly(k) amortized queries for matroid constraints." 
\end{tcolorbox}

In this paper, we answer this question as well as the open problem of Lattanzi et al.~\cite{DBLP:conf/nips/LattanziMNTZ20} (NeurIPS'20) affirmatively. 
As a byproduct, we also develop a dynamic algorithm for the 
submodular maximization problem under the cardinality constraint. 
We next state our main result.

\begin{tcolorbox}[width=\linewidth, colback=white!80!gray,boxrule=0pt,frame hidden, sharp corners]
\begin{theorem} 
[Main Theorem]
\label{thm:main}
Suppose we are provided with oracle access to a monotone submodular function $f: 2^{\ground} \rightarrow \mathbb{R}^+$ 
that is defined over a ground set $\ground$. 
Let $\mathcal{S}$ be a  sequence of insertions and deletions of elements of the underlying ground set $\ground$. 
Let $0 < \epsilon \le 1$ be an error parameter. 
\begin{itemize}
\item We develop the first parameterized (by the rank $k$ of a matroid $\matroid$) 
dynamic $(4+\epsilon)$-approximation algorithm for the submodular maximization problem under the matroid constraint using a worst-case expected $O(k\log(k)\log^3{(k/\epsilon)})$ query complexity. 
\item  We also  present a parameterized (by the cardinality constraint $k$) dynamic algorithm for the submodular maximization under the cardinality  constraint $k$, that maintains a $(2+\epsilon)$-approximate solution of the sequence $\mathcal{S}$ 
at any time $t$ using a worst-case expected complexity $\mO(k\epsilon^{-1}\log^2(k))$. 
\end{itemize}                
\end{theorem}

\end{tcolorbox}

The seminal work of Fisher, Nemhauser and Wolsey~\cite{DBLP:journals/mp/NemhauserWF78}, which we mentioned above,
 developed a simple and elegant leveling algorithm for the submodular maximization problem under the cardinality constraint that achieves the optimal approximation ratio of $\frac{e}{e-1} \approx 1.58$ in time $O(nk)$ ~\cite{DBLP:journals/mp/NemhauserWF78,DBLP:journals/jacm/Feige98}.

The study of the submodular maximization in the dynamic model was initiated at NeurIPS 2020 based on two independent works
due to Lattanzi, Mitrovic, Norouzi-Fard, Tarnawski, and Zadimoghaddam \cite{DBLP:conf/nips/LattanziMNTZ20} 
and Monemizadeh~\cite{DBLP:conf/nips/Monemizadeh20}.
Both works present dynamic algorithms 
that maintain $(2+\epsilon)$-approximate solutions for the submodular maximization  under the cardinality constraint $k$ 
in the dynamic model.
The amortized expected query complexity of these two algorithms are 
$\mO(\epsilon^{-11}\log^6(k)\log^2(n))$ and $\mO(k^2\epsilon^{-3}\log^5(n))$, respectively.

Our dynamic algorithm for the cardinality constraint 
improves upon the dynamic algorithm that Monemizadeh~\cite{DBLP:conf/nips/Monemizadeh20} (NeurIPS'20) 
developed for this problem using an expected query complexity $\mO(k^2\epsilon^{-3}\log^5(n))$. 
In particular, our dynamic algorithm is the first one for this problem whose 
query complexity is independent of the size of the ground set $\ground$.

We develop our dynamic algorithm for the submodular maximization problem under the matroid or cardinality constraint by designing a randomized leveled data structure that supports insertion and deletion operations, maintaining an approximate solution for the given problem. In addition, we develop a fast construction algorithm for our data structure that uses a one-pass over a random permutation of the elements and utilizes a monotonicity property of our problems which has a subtle proof in the matroid case. We believe these techniques could also be useful for other optimization problems in the area of dynamic algorithms.

Interestingly, our results can be seen from the lens of parameterized complexity~\cite{marx2008parameterized,fomin2021fast,DBLP:conf/focs/Korhonen21,KT11,DBLP:conf/mfcs/FafianieK14,DBLP:conf/soda/ChitnisCEHMMV16,DBLP:conf/soda/ChitnisCHM15,doi:10.1137/15M1032077,10.5555/3174304.3175483}. 
In particular, the query complexity of our dynamic algorithms for the submodular maximization problems under the matroid and cardinality 
constraints  (1) are independent of the size of the ground set $\ground$ (i.e., $|\ground| = n$) 
, and 
(2) are parameterized by the rank $k$ of the matroid $\matroid$ and the cardinality $k$, respectively. 
We hope that our work sheds light on the connection between dynamic algorithms  
and the Fixed-Parameter Tractability (FPT)~\cite{10.5555/2464827, DBLP:series/txtcs/FlumG06,cygan2015parameterized} world. 
We should mention that streaming algorithms~\cite{DBLP:conf/mfcs/FafianieK14,DBLP:conf/soda/ChitnisCEHMMV16,DBLP:conf/soda/ChitnisCHM15} through the lens of the parameterized complexity have been considered before 
where vertex cover and matching parameterized by their size were designed in these works.

Finally, one may ask whether we can obtain a dynamic $c$-approximate algorithm for the cardinality constraint 
for $c < 2$ with a  query complexity that is polynomial in $k$. 
Let $g: \mathbb{N} \to \mathbb{R}^+$ be an arbitrary function.
Building on a hardness result recently obtained by Chen 
and Peng~\cite{DBLP:journals/corr/abs-2111-03198}, 
we show in Appendix~\ref{sec:ower:bound} 
that there is no randomized $(2-\epsilon)$-approximate algorithm for the dynamic submodular maximization under cardinality constraint $k$ 
with amortized expected query time of $g(k)$ (e.g., not even doubly exponentially in $k$), 
even if the optimal value is known after every insertion/deletion.

\paragraph{Concurrent work.}
In a concurrent work, 
Dütting, Fusco, Lattanzi, Norouzi-Fard, and Zadimoghaddam~\cite{dutting2023fully} also provide an algorithm for
dynamic submodular maximization under a matroid constraint. Their algorithm
obtains a
$4+\epsilon$ approximation with $\frac{k^2}{\epsilon}\log(k)\log^2(n)\log^3(\frac{k}{\epsilon})$ amortized expected query comlpexity.~\footnote{The two works appeared on arxiv at the same time;
We had
submitted 
an earlier version of our work to
SODA'23, in July 2022.
}

Our query complexity of $k\log(k)\log^{3}(\frac{k}{\epsilon})$ is strictly better as
\textbf{(a)} it does not depend on $n$ and \textbf{(b)} its dependence
on $k$ is nearly linear rather than nearly quadratic and the dependence on $\epsilon^{-1}$ is polylogarithmic. Additionally, our guarantees are worst-case expected, rather than amortized expected.

\subsection{Preliminaries}
\label{sec:prelim}
\paragraph{Notations.}
For two natural numbers $x < y$, we use $[x, y]$ to denote the set $\{x,x+1,\cdots,y\}$, and $[x]$ to denote the set $\{1,2,\cdots,x\}$. 
For a set $A$ and an element $e$, we denote by $A+e$, the set that is the union of two sets $A$ and $\{e\}$. 
Similarly, for a set $A$ and an element $e \in A$, we denote by $A - e$ or $A \backslash e$, the set $A$ from which the element $e$ is removed. 
For a level $L_{i}$, we represent by $L_{1 \le j \le i}$ the levels $L_{1}, L_2,\cdots, L_i$, and we simplify $L_{1 \le j \le i}$ and show it by $L_{\le i}$. The levels $L_{i \le j \le T}$ and its simplification $L_{ i \le }$ are defined similarly. 
For a function $x$ and a set $A$, we denote by $x[A]$ the function $x$ that is restricted to domain $A$.
For an event $E$, we use $\ind{E}$ as the \emph{indicator function} of $E$. That is, $\ind{E}$ is set to one if $E$ holds and is set to zero otherwise. 
For random variables and their values, we use bold and non-bold letters, respectively. For example, we denote a random variable by $\textbf{X}$ and its value by $X$. 
We will use the notations $\Pr{\textbf{X}}$ and $\Ex{\textbf{X}}$ to denote the probability and the expectation of a random variable $\textbf{X}$. 
For two events $A$ and $B$, we will use the notation $\Pr{A|B}$ 
to denote "the conditional probability of $A$ given $B$" or "the probability of $A$ under the condition $B$".
For an event $A$ with nonzero probability and a discrete random variable $\textbf{X}$, we denote by $\Ex{\textbf{X}|A}$ conditional expectation of $X$ given $A$, which is $\Ex{\textbf{X} | A} = \sum_{x} x\cdot \Pr{\textbf{X} = x | A}$. 
Similarly, if $\textbf{X}$ and $\textbf{Y}$ are discrete random variables, 
the conditional expectation of $\textbf{X}$ given $\textbf{Y}$ is denoted by $\Ex{\textbf{X} | \textbf{Y} = y}$.

\paragraph{Submodular function.}
Given a ground set $\ground$, a function $f: 2^{\ground} \rightarrow \REAL^+$  is called \emph{submodular}  if it satisfies 
$ f(A \cup \{e\}) -f(A) \ge f(B \cup \{e\}) -f(B)$, for all $A \subseteq B \subseteq \ground$ and $e \notin B$. 
In this paper, we assume that $f$ is \emph{normalized}, i.e., $f(\emptyset) = 0$. 
When $f$ satisfies the additional property that $f(A \cup \{e\}) -f(A) \ge 0$ for all $A$ and $e \notin A$, we say $f$ is \emph{monotone}. 
For a subset  $A \subseteq \ground$ and an element $e \in \ground \backslash A$, the function $ f(A \cup \{e\}) - f(A)$ is often called the \emph{marginal gain}~\cite{DBLP:conf/kdd/BadanidiyuruMKK14,DBLP:conf/icml/0001MZLK19} of 
adding  $e$ to $A$.

Let $f: 2^{\ground} \rightarrow \REAL^+$ be a monotone submodular function defined on 
the ground set $\ground$.  
The \emph{monotone submodular maximization problem under the cardinality constraint} $k$ 
is defined as finding $OPT = \max_{I \subseteq \ground: |I| \le k} f(I)$. 
We denote by $I^*$ an optimal subset of size at most $k$ that achieves the optimal value $OPT = f(I^*)$. 
Note that we can have more than one optimal set. 

The leveling algorithm of the seminal work of Nemhauser, Wolsey, and Fisher \cite{DBLP:journals/mp/NemhauserWF78} that can approximate $OPT$ to a factor of $(1-1/e)$, is as follows. 
In the beginning, we let $S=\emptyset$. We then take $k$ passes over the set $V$, and in each pass, we 
find an element $e \in V$ that maximizes the marginal gain $f(S \cup \{e\}) - f(S)$, add it to $S$ and delete it from $V$. 

\paragraph{Access Model.}
We assume the access to a monotone submodular function $f: 2^{\ground} \rightarrow \REAL^+$ is given by an \emph{oracle}. 
That is, we consider the oracle that allows \emph{set queries} such that 
for every subset $A \subseteq \ground$, one can query the value $f(A)$. 
The marginal gain $f(A \cup \{e\}) - f(A)$ for every subset  $A \subseteq \ground$ and an element $e \in \ground$ in this query access model can be computed using two queries $f(A \cup \{e\})$ and $f(A)$.

\paragraph{Matroid.}
A matroid $\matroid$ consists of a \emph{ground set} $\ground$ and  
a nonempty downward-closed set system $\mathcal{I} \subseteq 2^{\ground}$  (known as the independent sets) that satisfies the \emph{exchange axiom}: 
for every pair of independent sets $A,B \in \mathcal{I}$ such that $|A| < |B| $, there exists an element $x\in B \backslash A$ such that $A \cup \{x\} \in \mathcal{I}$. 
A subset of the ground set $\ground$ that is not independent is called \emph{dependent}. 
A maximal independent set—that is, an independent set that becomes dependent upon adding any other element—is called a \emph{basis} for the matroid $\matroid$. 
A \emph{circuit} in a matroid $\matroid$ is a minimal dependent subset of $\ground$—that is, a dependent set whose proper subsets are all independent. 
Let $A$ be a subset of $V$. 
The rank of $A$, denoted by $rank(A)$, is the maximum cardinality of an independent subset of $A$. 

Let $f: 2^{\ground} \rightarrow \REAL^{\ge 0}$ be a monotone submodular function defined on 
the ground set $\ground$ of a matroid   $\matroid$. 
We denote by $OPT = \max_{I \in \mathcal{I}} f(I)$ the maximum submodular value of an independent set in $\mathcal{I}$. 
We denote by $I^* \in \mathcal{I}$ an independent set that achieves the optimal value $OPT = f(I^*)$. 

Here, we bring two lemmas about matroids that will be used in our paper. 
\begin{lemma}[\cite{DBLP:books/daglib/0070636}]
\label{lem:elimination:axiom}
The family $\mathcal{C}$ of circuits of a matroid $\matroid$ has the following properties:
\begin{itemize}
    \item ($C1$) $\emptyset \notin \mathcal{C}$.
    \item ($C2$) if $C_1$, $C_2 \in \mathcal{C}$ and $C_1 \subseteq C_2$, then $C_1=C_2$.
    \item ($C3$)  if $C_1$, $C_2 \in \mathcal{C}$, $C_1 \neq C_2$ and $e\in C_1 \cap C_2$, then there exists $C_3 \in \mathcal{C}$ such that $C_3 \subseteq C_1 \cup C_2 \setminus \{e\}$.
\end{itemize}
\end{lemma}

\begin{lemma}
\label{lem:<=1circuit}
Let $e\in V$ be an element, and $I \in \mathcal{I}$ be an independent set. Then $I \cup \{e\}$ has at most one circuit.
\end{lemma}

\begin{proof}
For the sake of contradiction, suppose there are two circuits $C_1, C_2 \subseteq I \cup \{e\}$, where $C_1 \neq C_2$. 
As $I$ is an independent set, $C_1, C_2 \nsubseteq I$, which means $e \in C_1 \cap C_2$. 
Then using Lemma~\ref{lem:elimination:axiom}, there exists a circuit $C_3 \subseteq C_1 \cup C_2 \setminus \{e\}$. Since $C_1 \cup C_2 \setminus \{e\} \subseteq I$, 
we have $C_3 \subseteq I$, which is a contradiction to the fact that the set $I$ is an independent set in $\mathcal{I}$. 
\end{proof}

\paragraph{Dynamic Model.}
Let $\mathcal{S}$ be a sequence of insertions and deletions of elements of an underlying ground set $\ground$.
Let $\mathcal{S}_t$ be the sequence of the first $t$ updates (insertion or deletion) of the sequence $\mathcal{S}$.
By time $t$, we mean the time after the first $t$ updates of the sequence $\mathcal{S}$, or simply when the updates of $\mathcal{S}_t$ are done. 
We define $V_t$ as the set of elements that have been inserted until time $t$ but have not been deleted after their latest insertion. 

Given a monotone submodular function  $f: 2^{\ground} \rightarrow \REAL^+$ defined on the ground set $\ground$, the aim of \emph{dynamic monotone submodular maximization problem under the cardinality constraint} $k$ is to have (an approximation of) $OPT_t = \max_{I_t \subseteq V_t: |I| \le k} f(I_t)$ at any time $t$. Similarly, the aim of \emph{dynamic monotone submodular maximization problem under a matroid $\matroid$ constraint} for a monotone function $f$ defined over the ground set $\ground$ is to have (an approximation of) $OPT_t=\max_{I_t \subseteq V_t: I_t\in \mathcal{I}}  f(I_t)$ at any time $t$. We also define $MAX_t$ to be $\max_{e \in V_t} f(e)$.
For simplicity, during the analysis for a fixed time $t$, we use $V$, $OPT$, and $MAX$ instead of $V_t$, $OPT_t$, and $MAX_t$ respectively.

Our dynamic algorithm is in the oblivious adversarial model as is common for analysis of randomized data structures such as universal hashing \cite{DBLP:conf/stoc/CarterW77}. 
The model allows the adversary, who is aware of the submodular function $f$ and the algorithm that is going to be used, to determine all the arrivals and departures of the elements in the ground set $\ground$.
However, the adversary is unaware of the random bits used in the algorithm and so cannot choose updates adaptively in response to the randomly guided choices of the algorithm. Equivalently, we can suppose that the adversary prepares the full input (insertions and deletions) before the algorithm runs. \\
The \emph{query complexity} of an $\alpha$-approximate dynamic algorithm is the number of oracle queries that the algorithm must make to 
maintain an $\alpha$-approximate of the solution at time $t$, given all computations that have been done 
till time $t-1$.\\
We measure the \emph{time complexity} of our dynamic algorithm 
in terms of its \emph{query complexity}, taking into account queries made to either the submodular oracle for $f$ or the matroid independence oracle for $\mI$. 

The query complexities of the algorithms in our paper will be worst-case expected.
An algorithm is said to have worst-case expected update time (or query time) $\alpha$ 
if for every update $x$, the expected time to process $x$ is at most $\alpha$. 
We refer to Bernstein, Forster, and Henzinger~\cite{DBLP:journals/talg/BernsteinFH21} 
for a discussion about the worst-case expected bound for dynamic algorithms.

\subsection{Our contribution and overview of techniques}
\label{sec:contrib}
Our dynamic algorithms for the submodular maximization problems with cardinality and matroid constraints 
consist of the following three building blocks. 
\begin{itemize}
    \item \emph{Fast leveling algorithms:} We first develop linear-time leveling algorithms 
    for these problems based on random permutations of elements.  
    These algorithms are used in \emph{rare occasions} when we need to (partially or totally) reset a solution that we maintain.  
    \item \emph{Insertion and deletion subroutines:} We next design subroutines for inserting and deleting a new element. 
    Upon insertion or deletion of an element, these subroutines often perform \emph{light} computations, but in rare occasions, 
    they perform \emph{heavy} operations by invoking the leveling algorithms to (partially or totally) reset a solution that we maintain. 
    \item \emph{Relax $OPT$ or $MAX$ assumptions:} For the leveling algorithms, and the insertion and the deletion subroutines, 
    we assume we know either an approximation of $OPT$ (as for the cardinality constraint) or 
    an approximation of the maximum submodular value $MAX = \max_{e \in V} f(e)$ of an element (as for the matroid constraint). 
    In the final block of our dynamic algorithms, we show how to relax such an assumption. 
\end{itemize}

\subsubsection{Submodular maximization problem  under the cardinality constraint} 
Designing and analyzing the above building blocks for the cardinality constraint is simpler than 
designing and developing them for the matroid constraint. 
Therefore, we outline them first for the cardinality constraint. 
That gives the intuition and sheds light on how we develop these building blocks for the matroid constraint 
which are more involved. 
Since the main contribution of our paper is developing a dynamic algorithm for the matroid constraint, 
we explain the algorithms (in Section~\ref{sec:algorithm:matroid}) and the analysis (in Section~\ref{sec:matroid:analysis}) for the matroid first. 
The dynamic algorithm and its analysis for the cardinality constraint are given in Section~\ref{sec:klogk} and Appendix~\ref{sec:analysis:klogk}, respectively. 

Suppose for now, we know the optimal value $OPT = \max_{I^* \subseteq V: |I^*| \le k} f(I^*)$ of any subset of the set $V$ 
of size at most $k$. We consider the fixed threshold $\tau = \frac{OPT}{2k}$.

\paragraph{First building block: Fast leveling algorithm.}
Our leveling algorithm constructs a set of levels $L_0, L_1, \cdots, L_T$, where $T$ is a random variable guaranteed to be $T \le k$. 
Every level $L_{\ell}$ consists of two sets $R_{\ell}$ and $I_{\ell}$, 
and an element $e_{\ell}$ so that:  
\begin{enumerate}
    \item $R_0 = V$, $I_0=\emptyset$, and $R_1 = \{e \in R_0: f(e) \ge \tau\}$
    \item $R_0 \supseteq R_1 \supset \cdots \supset R_T \supset R_{T+1}=\emptyset$ 
    
    \item For $1 \le \ell \le T$, we have $I_{\ell} = I_{\ell-1} \cup \{e_{\ell}\}$ 
    \item We report the set $I_T$ as the solution
\end{enumerate}

The illustration of our construction is shown in Figure~\ref{fig:layering}. 

\begin{figure}[h]
\begin{center}
\includegraphics[scale=0.35]{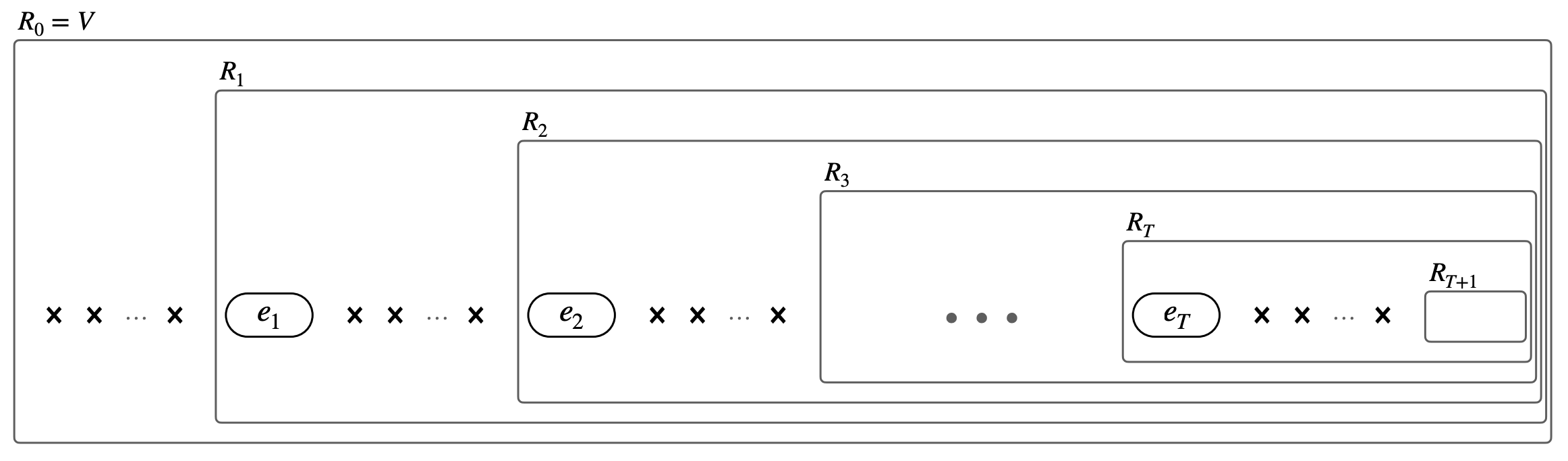}
\end{center}
\vspace{-0.3cm}
\caption{The illustration of our leveling algorithm. }
\label{fig:layering}
\end{figure}

The key concept in constructing the levels is the notion of \emph{promoting elements}. 

\begin{tcolorbox}[width=\linewidth, colback=white!80!gray,boxrule=0pt,frame hidden, sharp corners]

\begin{definition}[Promoting elements]
\label{def:cardinality:promote}
Let $L_{1 \le \ell \le T}$ be a level. 
We call an element $e \in R_{\ell}$, a promoting element for the set $I_{\ell}$ if 
$f(I_{\ell} \cup \{e\}) - f(I_{\ell}) \ge \tau$ and $|I_{\ell}| < k$.
\end{definition}

\end{tcolorbox}

The levels are constructed as follows:  
Let $\ell =  1$. 
We first randomly permute the elements of the set $R_1$ and denote by $P$ this random permutation. 
We next iterate through the elements of $P$ and for every element $e \in P$, 
we check if $e$ is a promoting element with respect to the set $I_{\ell - 1}$ or not. 

\begin{itemize}
    \item If $e$ is a promoting  element for the set $I_{\ell - 1}$, 
we then let $e_{\ell}$ be $e$ and let $I_{\ell}$ be $I_{\ell - 1} \cup \{e_{\ell}\}$. 
Observe that now we have the set $I_{\ell}$ and the element $e_{\ell}$, 
however, the set $R_{\ell}$ is not complete yet, 
as some of its elements may come after $e$ in the permutation $P$. 
We create the next level $L_{\ell+1}$ by setting $R_{\ell+1} = \emptyset$. 
We then proceed to the next element in $P$. 
Note that in this way, for all levels $L_{1 < j \le \ell}$, the sets $R_j$ are 
not complete, and they will be complete when we reach the end of the permutation $P$.

\item Next, we consider the case when $e$ is not a promoting element for the set $I_{\ell-1}$. 
    This essentially means that we need to find  
    the largest $z \in [1,\ell-1]$ so that $e$ is promoting  for the set $I_{z-1}$, 
    but it is not promoting for the set $I_{z}$. 
    A naive way of doing that is to perform a linear scan 
    for which we need one oracle query to compute $f(I_{x} \cup \{e\}) - f(I_{x}) $ for every $x \in [1,\ell-1]$. 
    However, we do a binary search in the interval $[i, \ell-1]$, which needs $O(\log k)$ oracle queries to find $z$.     
    Once we find such  $z$, we add $e$ to all sets $R_{r}$ for $2 \le r \le z$\footnote{Observe that $z$ is already in $R_1$}. 
    Observe that adding $e$ to all these sets 
    may need $O(k)$ time, but we do not need to do oracle queries in order to add 
    $e$ to these sets.
\end{itemize}

The permutation $P$ has at most $n$ elements, and we do the above operations for every such element.
Thus, the leveling algorithm may require a total of $O(n \log k)$ oracle queries.
Observe that the implicit property that we use to perform the binary search 
is the \emph{monotonicity property} which says if an element is a promoting for a set $I_{z-1}$, 
it is promoting for all sets $I_{\le z-1}$. For the cardinality constraint, the monotone property is trivial to see. 
However, it is \textit{\textbf{intricate}}  for the \textit{\textbf{matroid}}  constraint. 
We will develop a leveling algorithm for the matroid constraint, which satisfies a monotonicity property, allowing us to perform the binary search.

\paragraph{Second building block: Insertion and deletion of an element.}
Next, we explain the insertion and deletion subroutines. 
Let $\mathcal{S}$ be a sequence of insertions and deletions of elements of an underlying ground set $\ground$.
First, suppose we would like to delete an element $v$.
Observe that the set $R_0$ should contain all elements that have been inserted but not deleted so far. 
Thus, we remove $v$ from $R_0$. 
Now, two cases can occur:
\begin{itemize}
    \item \emph{Light computation:} 
        The first case is when $v \notin I_i$ for all $i \in [T]$. 
        Then, we do a light computation by   
        iterating through the levels $L_1,\cdots, L_T$, and 
        for each level $L_i$, we delete $v$ from $R_i$. 
        Handling this light computation takes zero query complexity as we do not make any oracle query. 
    \item \emph{Heavy computation:}
       However, if there exists a level $i \in [T]$ where $e_i = v$, we then  
       rebuild all levels $L_{i \le j \le T}$. To this end, we invoke the leveling algorithm 
       for the level $L_i$ to rebuild the levels $L_{i},\cdots, L_T$. 
       This computation is heavy, for which we need to make $O(|R_i|  \log k)$ oracle queries. 
\end{itemize}

When we invoke the leveling algorithm for a level $L_i$, 
it randomly permutes the elements $R_i$ and iterates through this random permutation 
to compute $I_{\ell}$, $R_{\ell}$, and $e_{\ell}$ for $i \le \ell \le T$. 
We prove that this means for every level $L_{\ell \ge i}$, the promoting element $e_{\ell}$ that we picked is sampled uniformly at random from the set $R_{\ell}$. 
This ensures that the probability that the sampled element $e_{\ell}$ being deleted is $\frac{1}{|R_{\ell}|}$.
Therefore, the expected number of oracle queries to reconstruct the levels $[i,T]$ is at most $ O( \frac{1}{|R_i|}  |R_i|  \log k) = O(\log k)$.
Since there are at most $T \le k$ levels, by the linearity of expectation, 
the expected oracle queries that a deletion can incur is $O(k\log k)$.

Next, suppose we would like to insert an element $v$. 
First of all, the set $R_0$ should contain all elements that have been inserted but not deleted so far. 
Thus, we add $v$ to $R_0$. 
Later, we iterate through levels $L_1,\cdots,L_{T+1}$, and 
for each level $L_i$, we check if $v$ is a promoting element for the previous level or not. 
If it is not, we break the loop and exit the insertion subroutine. 
However, if $v$ is a promoting element for the level $L_{i-1}$, 
we then add it to the set $R_i$ and with probability $\frac{1}{|R_i|}$, 
we let $e_i$  be $v$  and invoke the leveling algorithm 
for the level $L_{i+1}$ to rebuild the levels $L_{i+1},\cdots, L_T$. 
The proof of correctness for insertion uses similar techniques to the proof for deletion.

\paragraph{Third building block: Relax the assumption of having $OPT$.}
Our dynamic algorithm assumes the optimal value $OPT = \max_{I^* \subseteq V: |I^*| \le k} f(I^*)$ is given as a parameter.
However, in reality, the optimal value is not known in advance and it may change after every update. 
To remove this assumption, we use the well-known technique that has been also used in~\cite{DBLP:conf/nips/LattanziMNTZ20}.
Indeed, we run parallel instances of 
our dynamic algorithm for different guesses of the optimal value 
$OPT_t$ for the set $V_t$ of elements that have been inserted till time $t$, but not deleted, such that for any time $t$,
$\max_{I^* \subseteq V: |I^*| \le k} f(I^*) \in (OPT_t/(1+\epsilon), OPT_t]$ in one of the runs. 
These guesses are $(1+\epsilon)^i$ where $i\in \mathbb{Z}$. 
We apply each update on only $\mO(\log (k)/\epsilon)$ instances of our algorithm. 
See Section~\ref{sec:klogk} for the details.

\subsubsection{Submodular maximization problem  under the matroid constraint} 
The dynamic algorithm that we develop for the matroid constraint has similar building blocks as 
the cardinality constraint, but it is more intricate. 
We outline these building blocks for the matroid constraint next.

\paragraph{First building block: Leveling algorithm.}
Let $\matroid$ be a matroid whose rank is $k = rank(\mathcal{M})$. 
We first assume that we have the maximum submodular value $MAX = \max_{e \in V} f(e)$. 
We relax this assumption later. 
Our leveling algorithm builds a set of levels $L_0, L_1,\cdots,L_T$, where $T$ is 
a random variable guaranteed to satisfy $T = O(k\log(k/\epsilon))$. 
Every level $L_i$ consists of three sets $R_i$, $I_i$, and $I'_i$, 
and an element $e_i$. 
For these sets, we have the following properties: 
\begin{enumerate}
    \item $V=R_0 \supseteq R_1 \supset \cdots \supset R_T \supset R_{T + 1} = \emptyset$
    \item The sets $I_i$ are independent sets, i.e., $I_i \in \mathcal{I}$
    \item Each $I'_i$ is the union of all $I_{j}$ for $j \le i$, i.e., $I'_i = \bigcup_{j \le i} I_{j}$
    \item The sets $I'_i$ \emph{are not necessarily} independent
    \item We report the set $I_T$ as the solution
\end{enumerate}

The illustration of our construction is similar to the one for the cardinality constraint and 
is shown in Figure~\ref{fig:layering}. 
A key concept in our algorithm is again the notion of \emph{promoting elements}. 
However, the definition of promoting elements for the matroid constraint is more 
complex than that of the cardinality constraint,
and is inspired by the streaming algorithm of Chakrabarti and Kale~\cite{chakrabarti2015submodular}. 
The complexity comes from the fact that adding an element $e$ to an independent set, say $I$ may 
preserve the independency of $I$ or it may violate it by 
creating a circuit\footnote{A \emph{circuit} in a matroid $\mathcal{M}$ is a minimal dependent subset of $V$—that is, 
a dependent set whose proper subsets are all independent. }. 
In Lemma~\ref{lem:<=1circuit}, we prove that adding $e$ to an independent set can create at most one circuit.
Thus, we need to handle both cases when we define the notion of promoting elements. 

\begin{tcolorbox}[width=\linewidth, colback=white!80!gray,boxrule=0pt,frame hidden, sharp corners]

\begin{definition}[Promoting elements]
Let $L_{1 \le \ell \le T}$ be a level.
We call an element $e$, a promoting element for the level $L_{\ell}$  if 
\begin{itemize}
    \item \textbf{Property 1:} $f(I'_{\ell} + e) - f(I'_{\ell}) \ge \frac{\epsilon}{10k}\cdot MAX$, and 
    \item One of the following properties hold: 
        \begin{itemize}
            \item \textbf{Property 2: } $I_{\ell} + e$ is independent set (i.e., $I_{\ell}+e \in \mathcal{I}$) \emph{or}
            \item \textbf{Property 3: } $I_{\ell} +e$ is not independent and
            the minimum weight element $\hat{e} = \arg\min_{e' \in C} w(e')$ of 
            the set $C = \{ e' \in I_{\ell}: I_{\ell} + e - e' \in \mI\}$ satisfies 
            $2w(\hat{e}) \le {f(I'_{\ell} + e) - f(I'_{\ell})}$. 
        \end{itemize}    
\end{itemize}
\end{definition}

\end{tcolorbox}

We next explain the leveling algorithm. 
We first initialize $I_{0}$ and $I'_0$ as empty sets and let $R_0$ be the set of existing elements $V$.
We then let $R_1$ be all elements of the set $R_0$ that are promoting with respect to $L_0$. 
Observe that since $I_0 $ and $I'_0$ are empty sets, an element is  
filtered out from level $L_0$ because of Property 1.

The leveling algorithm can be called for any level $L_i$ and starting at that level, 
it builds the rest of levels $L_{i \le j \le T}$. 
Suppose our leveling algorithm is called for a level $L_i$ for $i \ge 1$. 
Let $\ell=i$.
We randomly permute the set $R_i$ and let $P$ be this random permutation. 
We iterate through the elements of $P$ and upon seeing a new element $e$, 
we check if $e$ is a promoting element for the level $L_{\ell - 1}$. 
\begin{itemize}
    \item The first case occurs if $e$ is a promoting  element for the level $L_{\ell - 1}$. 
    Note that $e$ is promoting if satisfies Property 1 and one of Properties 2 and 3.
        \begin{itemize}
            \item If the element $e$ satisfies Properties 1 and 2,  we set $I_{\ell} = I_{\ell - 1} + e$.
            \item If $e$ satisfies Properties 1 and 3, we set $I_{\ell} = I_{\ell - 1} + e - \hat{e}$.
        \end{itemize}        
        In both cases, the resulting $I_{\ell}$ is an independent set in $\mathcal{I}$. 
        We then fix the weight of $e$ to be $w(e) = \marginalgain{e}{I'_{\ell-1}}$. 
        Later, we let $I'_{\ell} := I'_{\ell - 1} + e$, and $e_{\ell} = e$. 
        Similar to the leveling algorithm that we develop for the cardinality case, 
        we now have the sets $I_{\ell}$ and $I'_{\ell}$, and the element $e_{\ell}$.
        However, the set $R_{\ell}$ is not complete yet, as some of its elements may come after $e$ in the permutation $P$. 
        We create the next level $L_{\ell+1}$ by setting $R_{\ell+1} = \emptyset$. 
        We then proceed to the next element in $P$. 
        Note that in this way, for all levels $L_{i < j \le \ell}$, the sets $R_j$ are 
        not complete and they will be complete when we reach the end of the permutation $P$. 
    \item The second case is when $e$ is not a promoting element for $L_{\ell - 1}$. 
    Here, similar to the cardinality constraint, our goal is to perform the binary search 
    to find the smallest $z \in [i,\ell-1]$ so that 
    $e$ is promoting  for the level $L_{z-1}$, but it is not promoting  for the level $L_{z}$. 
    \textbf{\textit{Interestingly, we prove the monotonicity property holds for the matroid constraint}}. 
    (The proof of this subtle property is given in Section~\ref{sec:binary:search}.)
    That is, we prove if $e$ is promoting for a level $L_{x-1}$, 
    it is promoting for all levels $L_{r \le x-1}$ and if $e$ is not promoting for a level $L_{x}$, 
    it is not promoting for all levels $L_{r \ge x}$. 
    Thus, we can do the binary search to find the smallest $z \in [i,\ell-1]$ so that 
    $e$ is promoting  for the level $L_{z-1}$, but it is not promoting  for the level $L_{z}$, which needs $\mO(\log(T))=\mO(\log(k\log(k/\epsilon))$ steps of binary search.     
    Once we find such  $z$, we add $e$ to all sets $R_{r}$ for $i < r \le z$. 
    Unlike the previous case, however,  we stay in 
    the current level $L_{\ell}$ and proceed to the next element of $P$.
    Observe that adding $e$ to all these sets 
    may need $ T = \mO(k\log(k/\epsilon))$ time, but we do not need to do oracle queries in order to add 
    $e$ to these sets.
\end{itemize}

\paragraph{Overview of the analysis:}
In order to prove the correctness of our leveling algorithm and compute its query complexity, 
we define two invariants; \textbf{\textit{level}} and \textbf{\textit{uniform}} invariants.  
The level invariant itself is a set of $5$ invariants \textbf{\textit{starter}}, \textbf{\textit{survivor}}, 
\textbf{\textit{independent}}, \textbf{\textit{weight}}, and \textbf{\textit{terminator}}. 
We show that these invariants are fulfilled by the end of the leveling algorithm (in Section~\ref{sec:level:proofs:matroid})
and after every insertion and deletion of an element (in Section~\ref{sec:update:matroid}). 

The level invariants assert that all elements that are added to $R_i$ at a level $L_i$ 
are promoting elements for the previous level. 
In other words, those elements of the set $R_{i-1}\backslash e_{i-1}$ 
that are not promoting will be filtered out and not be seen in $R_i$. 
Intuitively, this invariant provides us the approximation guarantee.
The uniform invariant asserts that for every level $L_{i \in [T]}$, 
conditioned on the random sets $R_{j \le i}$ and random elements $e_{j < i}$, 
the element $e_i$ is chosen uniformly at random from the set $R_i$. 
That is, $\Pr{e_i = e | R_{j \le i} \wedge e_{j < i}} = \frac{1}{|R_i|} \cdot \ind{e \in R_i}$\footnote{For an event $A$, we define $\ind{A}$ as the \emph{indicator function} of $A$.
That is, $\ind{A}$ is set to one if $A$ holds and is set to zero otherwise. }.
Intuitively, this invariant provides us with the randomness that we need  
to fool the adversary in the (fully) dynamic model 
which in turn helps us to develop a dynamic algorithm 
for the submodular maximization problem under a matroid constraint. 

The proof that the level and uniform invariants hold after every insertion and deletion 
is \textbf{\textit{novel}} and \textbf{\textit{subtle}}. 
This proof is given in Sections~\ref{sec:level:proofs:matroid} and~\ref{sec:update:matroid}. 
The technical part is to show that all promoting elements that are added to $R_i$ at a level $L_i$ (from the previous level) 
will be promoting after every update (i.e., insert or delete) and also, the sets $I_i$ will remain independent after updates. 
In addition, we need to show that uniformly chosen elements $e_i$ from survivor set $R_i$ 
will be uniform after every update. 

Now, we overview how we analyze the query complexity of our leveling algorithm.
Checking if an element $e$ is promoting for a level $L_{i}$ 
can be done using $O(\log(k))$ oracle queries using a binary search argument.
The proof is given in Section~\ref{sec:binary:search}.
The binary search that we perform in order to
place an element $e$ in the correct level requires $O(\log T)$ such promoting checks.
Thus, if we initiate the leveling algorithm with a set $R_i$, 
our algorithm needs $O(|R_i| \log(k) \log(T))$ oracle queries to build the levels $L_i,\cdots, L_T$ 
for $T=O(k\log(k/\epsilon))$.

\paragraph{Second building block: Insertion and deletion of an element.}
Now, we explain how to maintain the independent set $I_T$ upon insertions and deletions of elements. 
First, suppose we would like to delete an element $e$.
We iterate through levels $L_1,\cdots,L_T$ and 
for each level $L_i$ we delete $e$ from $R_i$ and we later check if $e$ is the element $e_i$ that we have picked for the level $L_i$. 
If this is the case, we then 
invoke the leveling algorithm for the set $R_i$ to \underline{reset} 
the levels $L_i,\cdots, L_T$. 
Since, the invocation of the leveling algorithm for the level $L_i$
may initiate $O(|R_i| \log(k) \log(T))$ oracle queries 
(to build the levels $L_i,\cdots, L_T$) and 
since the element $e_i$ is chosen uniformly at random from the set $R_i$ and 
we iterate through levels $L_1,\cdots,L_T$, thus, 
the worst-case expected query complexity of deletion is 
$\sum_{i=1}^T \frac{1}{|R_i|} \cdot \mO( |R_i|\cdot \log(k) \cdot \log(T)) = \mO(k\log(k)\log^2(k/\epsilon))$.

Next, suppose we would like to insert an element $e$. 
First of all, the set $R_0$ should contain all elements that have been inserted but not deleted so far. 
Thus, we add $e$ to $R_0$. 
Later, we iterate through levels $L_1,\cdots,L_T$ and 
for each level $L_i$, we check if $e$ is a promoting element for that level or not. 
If it is not, we break the loop and exit the insertion subroutine. 
However, if $e$ is indeed, a promoting element for the level $L_i$, 
we then add it to the set $R_i$ and with probability $1/|R_i|$, 
we set $e_i=e$ and invoke the leveling algorithm (with the input index $i+1$) to \underline{reset} 
the subsequent levels $L_{i+1},\cdots, L_T$. 
The query complexity of an insertion is proved similar to what we showed for a deletion.

The third block of our dynamic algorithm for the matroid constraint is to relax the assumption of knowing $MAX$. 
Relaxing this assumption is similar to what we did for the cardinality constraint. 
See Section~\ref{sec:algorithm:matroid} for the details.

\subsection{Related Work}
In this section, we state some known results for the submodular maximization problem under the matroid and cardinality constraints or some other related problems in the streaming, distributed, and dynamic models. 
In Table~\ref{table:summary}, we summarize the results in streaming and dynamic models for the submodular maximization problem under the matroid or cardinality constraint.

\begin{table*}[h]
\begin{center}
\begin{tabular}{|c|c|c|c|c|c|}
\hline
model & result & problem & approx. & query complexity  & ref. \\
\hline
 
\multirow{3}{*}{\parbox{2.5cm}{dynamic streaming model} } 
        &  \multirow{2}{*}{algorithm} & cardinality& $2+\epsilon$ & $O(\epsilon^{-1}dk\log(k))$  & \cite{DBLP:conf/icml/MirzasoleimanK017} \\
        
        &   & cardinality & $2+\epsilon$ & $O(dk\log^2(k)+d\log^3(k))$  & \cite{DBLP:conf/icml/0001ZK18} \\
        &   & matroid & $5.582+O(\epsilon)$ & $O(k+\epsilon^{-2}d\log(k))$  & \cite{DBLP:conf/icml/DuettingFLNZ22} \\
\hline

\multirow{2}{*}{\parbox{2.5cm}{insertion-only dynamic model }  }
        &  \multirow{2}{*}{algorithm} & \multirow{2}{*}{matroid}& $2+\epsilon$ & $k^{\tilde{O}(1/\epsilon)}$  & \cite{DBLP:journals/corr/abs-2111-03198} \\
        
        &   && $\frac{e}{e-1}+\epsilon$ & $k^{\tilde{O}(1/\epsilon^2)}\cdot \log(n)$  & \cite{DBLP:journals/corr/abs-2111-03198} \\
\hline

\multirow{8}{*}{\parbox{2.5cm}{fully dynamic model} }     

        &   \multirow{6}{*}{algorithm}  & \multirow{2}{*}{matroid} & \multirow{2}{*}{$4+\epsilon$}  &  $O(k\log(k)\log^3(k/\epsilon))$ &   this paper \\

        &  &&   &  $O(\frac{k^2}{\epsilon}\log(k)\log^2(n)\log^3(\frac{k}{\epsilon}))$ & \cite{dutting2023fully} \\
 
        \cline{3-6}
 
        &  & \multirow{4}{*}{cardinality}& \multirow{4}{*}{$2 + \epsilon$}  &  $O(\epsilon^{-3}k^2\log^4(n))$  & \cite{DBLP:conf/nips/Monemizadeh20}  \\

        &  &&   &  $O(\epsilon^{-4}\log^4(k)\log^2(n))$ & \cite{LattanziMNTZ20update} \\

        &  &&   &  $O(\poly(\log(n), \log(k), \frac{1}{\epsilon}))$ & \cite{pmlr-v202-banihashem23a} \\
 
        &  &&   &  $\mO(k\epsilon^{-1}\log^2(k))$  & this paper  \\
 
        \cline{2-6}
  
        &  \multirow{2}{*}{lower bound} & \multirow{2}{*}{cardinality}& $2-\epsilon$ & $n^{\tilde{\Omega}(\epsilon)}/k^3$  & \cite{DBLP:journals/corr/abs-2111-03198} \\
  
        &   && $1.712$ & $\Omega(n/k^3)$   & \cite{DBLP:journals/corr/abs-2111-03198} \\
        \hline
\end{tabular}
\end{center}
\caption{Results for the submodular maximization subject to cardinality  and matroid 
constraints. The lower bounds presented in~\cite{DBLP:journals/corr/abs-2111-03198} assume that we 
know the optimal submodular maximization value of the sub-sequence $\mathcal{S}_t$,  
where $\mathcal{S}_t$ is the set of elements that are inserted but not deleted from the beginning of the sequence $\mathcal{S}$ 
till any time $t$. 
}
\vspace{-0.4cm} 
\label{table:summary}
\end{table*}

\medskip

\paragraph{Known dynamic algorithms.} 
The study of the submodular maximization in the dynamic model is initiated 
at NeurIPS 2020 based on two independent works. 
The first work is due to Lattanzi, Mitrovic, Norouzi-Fard, Tarnawski, and Zadimoghaddam \cite{DBLP:conf/nips/LattanziMNTZ20} 
who present a randomized dynamic algorithm 
that maintains an expected $(2+\epsilon)$-approximate solution of the maximum submodular (under the cardinality constraint $k$) 
of a dynamic sequence $\mathcal{S}$. The amortized expected query complexity of their algorithm 
is $O(\epsilon^{-11}\log^6(k)\log^2(n))$. 
The second work is due to Monemizadeh~\cite{DBLP:conf/nips/Monemizadeh20} who 
presents a randomized dynamic  
algorithm with approximation guarantee $(2+\epsilon)$. 
The amortized expected query complexity of his algorithm is 
$O(\epsilon^{-3}k^2\log^5(n))$.
The original version of the algorithm in Lattanzi et al.~\cite{DBLP:conf/nips/LattanziMNTZ20} has some correctness issues, as pointed out by Banihashem, Biabani, Goudarzi, Hajiaghayi, Jabbarzade, and Monemizadeh ~\cite{pmlr-v202-banihashem23a} at ICML 2023, who also provide an alternative algorithm for solving this problem with polylogarithmic update time.
Those issues were also subsequently fixed by Lattanzi et al.~\cite{LattanziMNTZ20update} by modifying their algorithm. The query complexity of their new algorithm is $O(\epsilon^{-4}\log^4(k)\log^2(n))$ per update.
Peng's work at NeurIPS 2021 \cite{DBLP:conf/nips/Peng21} focuses on the dynamic influence maximization problem, which is a white box dynamic submodular maximization problem.
Work of Banihashem, Biabani, Goudarzi, Hajiaghayi, Jabbarzade, and Monemizadeh \cite{nonbanihashem2023dynamic} at NeurIPS 2023 solves dynamic non-monotone submodular maximization under cardinality constraint $k$.

At STOC 2022, Chen and Peng~\cite{DBLP:journals/corr/abs-2111-03198} show two lower bounds for the submodular maximization in the dynamic model. Both of these lower bounds hold even if we know the optimal submodular maximization value of the sequence $\mathcal{S}$ at any time $t$. Their first lower bound shows that any randomized algorithm that achieves an approximation ratio of $2-\epsilon$ for dynamic submodular maximization under cardinality constraint $k$ 
requires amortized query complexity  $n^{\tilde{\Omega}(\epsilon)}/k^3$. 
They also prove a stronger result by showing that any randomized algorithm for dynamic submodular maximization under cardinality constraint $k$ that obtains an approximation guarantee of $1.712$ must have amortized query complexity at least $\Omega(n/k^3)$. 

Chen and Peng~\cite{DBLP:journals/corr/abs-2111-03198} also
studied the complexity of the submodular maximization under matroid constraint in the insertion-only dynamic model (a restricted version of the fully dynamic model where deletions are not allowed) and they developed two algorithms for this problem.
The first algorithm maintains a $(2+\epsilon)$-approximate independent set of a matroid $\matroid$ such that the expected number of oracle queries per insertion is $k^{\tilde{O}(1/\epsilon)}$. Their second algorithm is a $(\frac{e}{e-1}+\epsilon)$-approximation algorithm 
using an amortized query complexity of $k^{\tilde{O}(1/\epsilon^2)}\cdot \log(n)$, where $k$ is the rank of $\matroid$ and $n = |\ground|$. 
However, these results do not work for the classical (fully) dynamic model, and they posed developing a dynamic algorithm for the submodular maximization problem under the matroid constraint in the (fully) dynamic model as an open problem. 

And as discussed previously, the concurrent work of 
Dütting et al.~\cite{dutting2023fully} at ICML 2023 provides an algorithm for dynamic submodular optimization under matroid constraint. Their algorithm has a 
$4+\epsilon$ approximation guarantee and $O(\frac{k^2}{\epsilon}\log(k)\log^2(n)\log^3(\frac{k}{\epsilon}))$ amortized expected query complexity.

\paragraph{Known (insertion-only) streaming algorithms.}
The first streaming algorithm for the submodular maximization under the cardinality constraint 
was developed by Badanidiyuru, Mirzasoleiman, Karbasi, and Krause~\cite{DBLP:conf/kdd/BadanidiyuruMKK14}. 
In this seminal work, the authors developed a $(2+\epsilon)$-approximation algorithm 
for this problem using $O(k\epsilon^{-1}\log k )$ space. 
Later, Kazemi, Mitrovic, Zadimoghaddam, Lattanzi and Karbasi \cite{DBLP:conf/icml/0001MZLK19} 
proposed a space streaming algorithm for this problem that improves the space complexity down to 
$O(k\epsilon^{-1})$.

In a groundbreaking work, Chakrabarti and Kale~\cite{chakrabarti2015submodular} at IPCO'14 
designed a streaming framework for submodular maximization problems 
under the matroid and matching constraints, as well as other constraints where independent sets are given 
either by a hypermatching constraint in $p$-hypergraphs or by the intersection of $p$ matroids.  
In particular, their streaming framework gives a $4$-approximation streaming algorithm for the 
submodular maximization under the matroid constraint using $O(k)$  space, where $k$ is 
the rank of the underlying matroid $\matroid$. 
The approximation ratio was recently improved to 3.15
by Feldman, Liu, Norouzi-Fard, Svensson, and Zenklusen~\cite{feldman2021streaming}.

Later, Chekuri, Gupta, and Quanrud~\cite{DBLP:conf/icalp/ChekuriGQ15} 
developed one-pass streaming algorithms for 
(non-monotone)  submodular maximization problems under $p$-matchoid\footnote{A set system $(N,\mathcal{I})$ 
is $p$-matchoid if there exists $m$ matroids $(N_1,\mathcal{I}_1),\cdots,(N_m,\mathcal{I}_m)$ 
such that every element of $N$ appears in the ground set of at
most $p$ of these matroids and $\mathcal{I} = \{S \subseteq 2^{\mathcal{N}}: \forall_{1 \le i \le m} S \cap N_i \in \mathcal{I}_i\}$.} constraint as well as simpler 
streaming algorithms for the monotone case that have the same bounds as 
those of Chakrabarti and Kale~\cite{chakrabarti2015submodular}. 
(These two works~\cite{chakrabarti2015submodular,DBLP:conf/icalp/ChekuriGQ15} were inspiring works for us as well). 

\paragraph{Known streaming algorithms for related submodular problems.} 
For non-monotone submodular objectives, the first streaming result was obtained by
Buchbinder, Feldman, and Schwartz~\cite{DBLP:conf/soda/BuchbinderFS15a}, 
who designed a randomized streaming algorithm achieving an $11.197$-approximation 
for the problem of maximizing a non-monotone submodular function subject
to a single cardinality constraint.

Chekuri, Gupta, and Quanrud~\cite{DBLP:conf/icalp/ChekuriGQ15} further extended the work of Chakrabarti and Kale by developing  
$(5p+2+1/p)/(1-\epsilon)$-approximation algorithm for 
the non-monotone submodular maximization problems under $p$-matchoid constraints in the insertion-only streaming model. 
They also devised a deterministic streaming algorithm achieving $(9p+O(\sqrt{p}))/(1-\epsilon)$-approximation for the same problem. 
Later, Mirzasoleiman, Jegelka, and Krause~\cite{DBLP:conf/aaai/MirzasoleimanJ018} designed a different deterministic algorithm for 
the same problem achieving an approximation ratio of $4p+4\sqrt{p}+1$.

At NeurIPS'18, Feldman, Karbasi and Kazemi~\cite{DBLP:conf/nips/FeldmanK018} improved these results 
for monotone and non-monotone submodular maximization under the $p$-matchoid constraint 
with respect to the space usage and approximation factor. 
As an example, their streaming algorithm for non-monotone submodular under $p$-matchoid 
achieves $4p+2-o(1)$-approximation that improves upon the randomized 
streaming algorithm proposed in~\cite{DBLP:conf/icalp/ChekuriGQ15}. 

\paragraph{Known dynamic streaming algorithms.}
Mirzasoleiman, Karbasi and Krause \cite{DBLP:conf/icml/MirzasoleimanK017} and
Kazemi, Zadimoghaddam and Karbasi \cite{DBLP:conf/icml/0001ZK18} proposed dynamic streaming algorithms for 
the cardinality constraint. 
In particular,  the authors in \cite{DBLP:conf/icml/MirzasoleimanK017} developed a dynamic streaming algorithm that given
a stream of inserts and deletes of elements of an underlying ground set $\ground$, $(2+\epsilon)$-approximates
the submodular maximization under cardinality constraint using $O((dk\epsilon^{-1}\log k)^2)$ space and
$O(dk\epsilon^{-1}\log k)$ average update time, where $d$ is an upper-bound for the number of deletes that are allowed.

The follow-up paper \cite{DBLP:conf/icml/0001ZK18} studies approximating submodular maximization under cardinality constraint
in three models, (1) centralized model, (2) dynamic streaming where we are allowed to insert and delete (up to $d$) elements
of an underlying ground set $\ground$, and (3) distributed (MapReduce) model.
In order to design a generic framework for all three models, they compute a coreset for submodular maximization
under cardinality constraint. Their coreset has a size of $O(k\log k+d\log^2 k)$. Out of this coreset, we can extract
a set $S$ of size at most $k$ whose $f(S)$ in expectation is at least $2$-approximation of the optimal solution.
The time to extract such a set $S$ from the coreset is $O(dk\log^2 k+d\log^3 k)$. 

The algorithms presented in \cite{DBLP:conf/icml/MirzasoleimanK017} and \cite{DBLP:conf/icml/0001ZK18} 
are dynamic streaming algorithms (not fully dynamic algorithms) whose time complexities depend on 
the number of deletions (Theorem 1 of the second reference). 
Therefore, their query complexities will be high
if we recompute a solution after each insertion or deletion.  
Indeed, if the number of deletions is linear in terms of the maximum size of the ground set $\ground$, 
it is in fact better to re-run the known leveling algorithms (say, \cite{DBLP:journals/mp/NemhauserWF78}) 
after every insertion and deletion.
A similar result was recently obtained for the submodular maximization 
under the matroid constraint. 
At ICML 2022, Duetting,  Fusco, Lattanzi, Norouzi{-}Fard, Zadimoghaddam~\cite{DBLP:conf/icml/DuettingFLNZ22}
presented a streaming $(5.582+O(\epsilon))$-approximation algorithm for 
    the deletion robust version of this problem, where the number of deletions is known to the algorithm, and they are revealed at the end of the stream. 
The space usage of their algorithm is $O(k+\epsilon^{-2}d\log(k))$, 
which is again linear if the number of deletions ($d$) is linear in terms of the maximum size of the ground set $\ground$. 
This was subsequently improved by
Zhang, Tatti, and Gionis~\cite{zhang2022coresets}.

\paragraph{Known MapReduce algorithms.}
The first distributed algorithm for the cardinality constrained submodular maximization was due to Mirrokni and Zadimoghaddam \cite{DBLP:conf/stoc/MirrokniZ15} who gave a $3.70$-approximation in $2$
rounds without duplication and a $1.834$-approximation with significant duplication of the ground set (each element being sent
to $\Theta(\frac{1}{\epsilon}\log(\frac{1}{\epsilon}))$ machines).
Later, Barbosa, Ene, Nguyen and    Ward \cite{DBLP:conf/focs/BarbosaENW16}
achieves a $(2+\epsilon)$-approximation in $2$ rounds and
was the first to achieve a $(\frac{e}{e-1}+\epsilon)$ approximation  in $O(\frac{1}{\epsilon})$ rounds.
Both algorithms require $\Omega(\frac{1}{\epsilon})$ duplication.
\cite{DBLP:conf/focs/BarbosaENW16}
mentions that without duplication, the two algorithms could be implemented in
$O(\frac{1}{\epsilon}\log(\frac{1}{\epsilon}))$ and $O(\frac{1}{\epsilon^2})$ rounds, respectively.

In a subsequent work, Liu and Vondrak \cite{DBLP:conf/soda/LiuV19} develop a simple thresholding
algorithm that with one random partitioning of the dataset (no duplication) achieves the following:
In $2$ rounds of MapReduce, they obtain a $(2+\epsilon)$-approximation 
and in $2/\epsilon$ rounds, they achieve $(\frac{e}{e-1}-\epsilon)$-approximation.
Their algorithm is inspired by the streaming algorithms that are presented in \cite{DBLP:journals/topc/KumarMVV15} and 
\cite{DBLP:journals/mst/McGregorV19}. It is also similar to the algorithm of Assadi and Khanna \cite{DBLP:conf/soda/AssadiK18}
who study the communication complexity of the maximum coverage problem.

\section{Dynamic algorithm for submodular matroid maximization}
\label{sec:algorithm:matroid}
In this section, we present our dynamic algorithm for the submodular maximization problem under the matroid constraint.
The pseudocode of our algorithm is provided in Algorithms \ref{alg:matroid} and \ref{alg:matroid-updates}. 
The overview of our dynamic algorithm is given in Section~\ref{sec:contrib} "Our contribution". 

\paragraph{Promoting Elements}
As we explained in Section~\ref{sec:contrib} "Our contribution", 
a key concept in our algorithm is the notion of \emph{promoting elements}.

\begin{tcolorbox}[width=\linewidth, colback=white!80!gray,boxrule=0pt,frame hidden, sharp corners]

\begin{definition}[Promoting elements]
\label{def:promote}
Let $L_{1 \le \ell \le T}$ be a level.
We call an element $e$, a promoting element for the level $L_j$  if 
\begin{itemize}
    \item \textbf{Property 1:} $f(I'_{\ell} + e) - f(I'_{\ell}) \ge \frac{\epsilon}{10k}\cdot MAX$, and 
    \item One of the following properties hold: 
        \begin{itemize}
            \item \textbf{Property 2: } $I_{\ell} + e$ is independent set (i.e., $I_{\ell}+e \in \mathcal{I}$) \emph{or}
            \item \textbf{Property 3: } $I_{\ell}+e$ is not independent and
            the minimum weight element $\hat{e} = \arg\min_{e' \in C} w(e')$ of 
            the set $C = \{ e' \in I_{\ell}: I_{\ell} + e - e' \in \mI\}$ satisfies 
            $2w(\hat{e}) \le {f(I'_{\ell} + e) - f(I'_{\ell})}$. 
        \end{itemize}    
\end{itemize}
\end{definition}

\end{tcolorbox}

\begin{algorithm}[ht]
  \caption{\matroidleveling$(\matroid, MAX)$ }
  \label{alg:matroid}
  \begin{algorithmic}[1]
    \Function{\init}{$V$}
        \State $I_{0} \gets \emptyset,\quad I'_{0} \gets \emptyset, \quad R_{0} \gets V$
        \State $R_{1} \gets \{ e \in R_0 : \replacementTester{}(I_0, I'_0, e, w[I_0]) \ne \err\}$
        \State Invoke \MatroidConstLevel$(i = 1)$
    \EndFunction
    
 \rule{15cm}{0.4pt} 
 
    \Function{\MatroidConstLevel}{$i$}
        \State Let $P$ be a random permutation of elements of $R_{i}$ and $\ell \gets i$ 
        \For{$e$ in $P$}\label{line:iterate_P}
            \If{ \replacementTester$(I_{\ell-1}, I'_{\ell-1}, e, w[I_{\ell-1}]) \ne \err$}
                \State $y \gets$ \replacementTester$(I_{\ell-1}, I'_{\ell-1}, e, w[I_{\ell-1}]))$ and $z \gets \ell$
                \label{line:const_level_matroid:def_y}
                \State Fix the weight $w(e) \gets f(I'_{\ell - 1} + e) - f(I'_{\ell -1})$, and set the element $e_{\ell} \gets e$
                \label{line:const_level_matroid:w}
                \State Let  $I_{\ell} \gets (I_{\ell-1} + e) \backslash y$,  \ \ $I'_{\ell} \gets I'_{\ell-1} + e$,  \ \ $R_{\ell+1} \gets  \emptyset$, \ \  and then $\ell \gets \ell + 1$\label{line:constlevelmatroid:increase_ell} 

            \Else
                \State Run binary search to find the lowest $z \in [i, \ell-1]$ such that \replacementTester$(I_z, I'_z, e, w[I_z]) = \err$ \label{line:const_level_matroid:binary_search}
            \EndIf
                \For{$r \gets i+1$ \textbf{to} $z$}  
                    \State $R_r \gets R_r + e$\label{line:constlevelmatroid:addR_bs}
                \EndFor
          \EndFor
        \State \Return $T \gets \ell-1$ which is the final $\ell$ that the for-loop above returns subtracted by one
    \EndFunction
    
 \rule{15cm}{0.4pt} 
 
    \Function{\replacementTester}{$I, I', e, w[I]$}
        \label{alg:replacement}
          \If{$f(I' \cup \{e\}) - f(I') \notin [\frac{\epsilon}{10k}\cdot MAX, MAX]$}\label{line:remove_w_not_in_range}
              \State \Return \err
          \EndIf
          \If{$I + e \in \mI$}
             \State \Return $\emptyset$
          \EndIf
          \State $C \gets \{ e' \in I: I + e - e' \in \mI\}$ and let $\hat{e} \gets \arg\min_{e' \in C} w(e')$ \label{line:matroid:promote:find_min}
          \If{$2\cdot w(\hat{e}) \le  f(I' + e) - f(I')$} 
             \State \Return $\{\hat{e}\}$
          \Else \ \ 
             \State \Return \err
          \EndIf
    \EndFunction
  \end{algorithmic}
\end{algorithm}

We define the function 
$\replacementTester(I_{\ell}, I'_{\ell}, e, w[I_{\ell}])$ for an element $e \in V$ 
with respect to the level $L_{\ell}$ which 
\begin{itemize}
    \item returns $\emptyset$ if properties 1 and 2 hold; 
    \item returns $\hat{e}$ if properties 1 and 3 hold;
    \item returns $\err$ otherwise.
\end{itemize}

Subroutine \replacementTester{} in Algorithm \ref{alg:matroid} implements this function. 
This subroutine checks if an element $e \in V$ is a promoting element for a level $L_{\ell}$ or not.  
In case that $e$ is a promoting element for $L_{\ell}$, the subroutine \replacementTester{} finds 
an element $e'$ (if it exists) that satisfies Property 3 of definition~\ref{def:promote} and replaces it by $e$. 

Our leveling algorithm consists of three subroutines, \init{}, \MatroidConstLevel{}, and \replacementTester{}. 
We explained in above Subroutine \replacementTester{}. 
In Subroutine \init{}, we first initialize $I_{0}$ and $I'_0$ as empty set and set $R_0$ to the ground set $V$.
We then let $R_1$ be all elements of the set $R_0$ that are promoting  with respect to $L_0$. 
Observe that since $I_0 $ and $I'_0$ are empty sets, 
if an element $e$ 
filtered out from the level $L_0$, 
i.e., $e\in L_0$ but $e\notin L_1$, then
$e$ was filtered
because of Property 1.
Finally, we invoke \MatroidConstLevel{} for the level $L_1$, to build all the remaining levels.
Subroutine \MatroidConstLevel{} implements our leveling algorithm 
that we gave an overview of it in Section~\ref{sec:contrib} "Our contribution".

\begin{algorithm}[ht]
\caption{\matroidupdates$(\matroid, MAX)$ }
\label{alg:matroid-updates}
    \begin{algorithmic}[1]
    \Function{\deletev}{$v$}
        \State $R_0 \gets R_0 - v$
        \For{ $i \gets 1$ \textbf{to} $T$} 
            \If{$v \notin R_i$}  \State \Break
            \EndIf
            \State $R_i \gets R_i - v$ 
            \If{$e_i = v$}
            \State Invoke $\MatroidConstLevel(i)$
            \label{line:reset:delete:matroid}
            \State \Break
            \EndIf
        \EndFor
    \EndFunction
    
\rule{15cm}{0.4pt} 

    \Function{\insertv}{$v$}
        \State $R_0 \gets R_0 + v$. 
        \For{$i \gets 1$ \textbf{to} $T+1$}
            \If{\replacementTester$(I_{i-1}, I'_{i-1}, v, w[I_{i-1}])$ = \err}
            \label{line:mat:insert:break}
             \State \Break
            \EndIf
            \State $R_{i} \gets R_{i} + v$.
            \State Let $p_i=1$ with probability $\frac{1}{|R_i|}$, and otherwise $p_i=0$ \label{line:p:insert:klogk}
            \If{$p_i=1$}   \label{line:mat:insert:if}
                \State {$e_i \gets v$, \quad $w(e_i) \gets f(I'_{i - 1} + v) - f(I'_{i -1})$, \quad 
                $y \gets \replacementTester{}(I_{i-1}, I'_{i-1}, v, w[I_{i-1}])$ } \label{line:mat:insert:setw}
                \State {
                $I_i \gets I_{i-1} + v - y, \quad I'_i \gets I'_i + v$} \label{line:mat:insert:setI}
                \State {$R_{i+1} = \{e' \in R_i: \replacementTester{}(I_i, I'_i, e', w[I_i]) \ne \err \}$} \label{line:mat:insert:setRi+1}
                \State {$\MatroidConstLevel{}(i+1)$
                }
                \State \Break
            \EndIf
        \EndFor
    \EndFunction
    
  \end{algorithmic}
\end{algorithm}

\paragraph{Relaxing $MAX$ assumption.}

Our dynamic algorithm assumes the maximum value $\max_{e \in V} f(e)$ is given as a parameter.
However, in reality, the maximum value is not known in advance and it may change after every insertion or deletion. 
To remove this assumption, we run parallel instances of 
our dynamic algorithm for different guesses of the maximum value $MAX_t$ at any time $t$ of the sequence $\mathcal{S}_t$, such that 
$\max_{e \in V_t} f(e) \in (MAX_t/2, MAX_t]$ in one of the runs. 
Recall that $V_t$ is the set of elements that have been inserted but not deleted from the beginning of the sequence till time $t$. 
These guesses that we take are $2^i$ where $i\in \mathbb{Z}$. 
If $\rho$ is the ratio between the maximum and minimum non-zero possible value 
of an element in $V$, 
then the number of parallel instances of our algorithm will be 
$\mO(\log\rho)$. 
This incurs an extra $\mO(\log\rho)$-factor in the query complexity of our dynamic algorithm.

\begin{algorithm}[h] 
\caption{Unknown $MAX$} 
\begin{algorithmic}[1]
    \State Let $\mathcal{A}_i$ be the instance of our dynamic algorithm, for which $MAX=2^i$
    
     \rule{15cm}{0.4pt} 
    
    \Function{UpdateWithoutKnowingMAX}{$e$}
        \For{\textbf{each} $i \in \left[\ceil{\log{f(e)}},\floor{\log{\left(\frac{10k}{\epsilon}\cdot f(e)\right)}}\right]$} \Comment{$\frac{\epsilon}{10k}\cdot{2^i} \leq f(e) \leq 2^i$}
            \State Invoke $\update(e)$ for instance $\mathcal{A}_i$
        \EndFor
    \EndFunction
\end{algorithmic}
\end{algorithm}

Next, we show how to replace this extra factor with 
an extra factor of $\mO(\log{(k/\epsilon)})$ which is independent of $\rho$. 
We use the well-known technique that has been also used in~\cite{DBLP:conf/nips/LattanziMNTZ20}. 
In particular, for every element $e$, we add it to those instances $i$ 
for which we have $\frac{\epsilon}{10k}\cdot{2^i} \leq f(e) \leq 2^i$.
The reason is if the maximum value of $V_t$ is within the range $(2^{i-1},2^i]$ and 
$f(e) > 2^i$, then $f(e)$ is greater than the maximum value and can safely be ignored for the instance $i$ 
that corresponds to the guess $2^i$. 
On the other hand, we can safely ignore all elements $e$ whose $f(e) < \frac{\epsilon}{10k}\cdot{2^i}$, 
since these elements will never be a promoting element in the run with $MAX = 2^i$.
This essentially means that every element $e$ is added to at most 
$\mO(\log{(k/\epsilon)})$ parallel instances. 
Thus, after every insertion or deletion, 
we need to update only $\mO(\log{(k/\epsilon)})$ instances of our dynamic algorithm.

\section{Analysis of dynamic algorithm for submodular matroid}
\label{sec:matroid:analysis}

In this section, we prove the correctness of our \MatroidConstLevel{}, \insertv{}, and \deletev{} algorithms. We will also compute the query complexity of each one of them. To analyze our randomized algorithm, for any variable $x$ in our pseudo-code, we use $\bold{x}$ to denote it as a random variable and use $x$ itself to denote its value in an execution.
The most frequently used random variables in our analysis are as follows:

\begin{tcolorbox}[width=\linewidth, colback=white!80!gray,boxrule=0pt,frame hidden, sharp corners]

\begin{itemize}
    \item  We denote by $\bE_i$ the random variable corresponding to the element $e_i$ picked at level $L_i$. 
     \item  We denote by $\bR_i$ the random variable that corresponds to the set $R_i$.
    \item  The random variable $\bT$ corresponds to $T$, which is the index of the last non-empty level created. 
    Indeed, for a level $L_i$ to be existent and non-empty, $\bT \ge i$ should hold. 
    
    \item We define
    $H_i = (e_1, \dots, e_{i-1}, R_0, \dots, R_{i})$
    as \emph{the partial configuration up to the level $L_i$}.
    Note that $R_{i}$ is included in this definition, while $e_{i}$ is not.
    $\bH_i := (\bE_1, \dots, \bE_{i-1}, \bR_0,  \bR_{1}, \dots, \bR_{i})$ is the random variable corresponding to the partial configuration $H_i$.
\end{itemize}

\end{tcolorbox}

We break the analysis of our algorithm into a few steps. 

\paragraph{Step 1: Analysis of binary search.} 
In the first step, we prove that the binary search that 
we use to speed up the process of finding the right levels for non-promoting elements  
works. Indeed, we prove that if $e \in V$ is a promoting element for a level $L_{z-1}$, 
it is promoting  for all levels $L_{r \le z-1}$ and 
if $e$ is not promoting  for the level $L_{z}$, 
it is not promoting  for all levels $L_{r \ge z}$. 
Therefore, because of this monotonicity property, we can do a binary search to find the smallest $z \in [i,\ell-1]$ so that 
$e$ is promoting  for the level $L_{z-1}$, but it is not promoting  for the level $L_{z}$. 
Additionally, we show that checking
whether $e$ is promoting for a level $L_{z}$ can be done with $O(\log(k))$ queries using a binary search argument.

\paragraph{Step 2: Maintaining invariants.}
We define six invariants, and we show that these invariants \emph{hold} when \init{} is run, and our whole data structure gets built, \emph{and are preserved} after every insertion and deletion of an element. 

\begin{tcolorbox}[width=\linewidth, colback=white!80!gray,boxrule=0pt,frame hidden, sharp corners]
\textbf{Invariants:} 
\begin{enumerate}

\item \textbf{Level invariants.} 
\begin{enumerate}
    \item \textbf{Starter.} $R_0=V$ and $I_0 = I'_0 = \emptyset$
    \item \textbf{Survivor.}  For $1 \leq i \leq T + 1$, $R_i = \{ e \in R_{i-1} - e_{i-1}: \replacementTester{}(I_{i-1}, I'_{i-1}, e, w[I_{i-1}]) \ne \err\}$ 
    \item  \textbf{Independent.} For $1 \leq i \leq T$, $I_i = I_{i-1} + e_i - \replacementTester{}(I_{i-1}, I'_{i-1}, e_i, w[I_{i-1}])$, and $I'_i = \cup_{j \le i} I_j$
    \item  \textbf{Weight.} For $1 \leq i \leq T$, $e_i \in R_i$ and  $w(e_i) = f(I'_{i-1} + e_i) - f(I'_{i-1})$
    \item \textbf{Terminator.} $R_{T+1}=\emptyset$
\end{enumerate}
     \item  \textbf{Uniform invariant.} For all $i \ge 1$, conditioned on the random variables $\bT$ and $\bH_i$, the element $e_i$ is chosen uniformly at random 
          from the set $R_i$. That is, 
            $ \Pr{\bE_i = e |  \bT \geq i \text{ and } \bH_i = H_i  }= \frac{1}{|R_i|}\cdot \ind{e\in R_i}  $.

\end{enumerate}
\end{tcolorbox}

The survivor invariant says that all elements that are added to $R_i$ at a level $L_i$ 
are promoting elements for that level. 
In other words, those elements of the set $R_{i-1} - e_{i-1}$ 
that are not promoting will be filtered out and not be seen in $R_i$. 
The terminator invariant shows that the recursive construction of levels stops when the survivor set becomes empty. 
The independent invariant shows that the sets $I_i$ are independent sets of the matroid $\matroid$, and $I'_i$ is equal to the union of $I_1, \dots, I_i$.
The weight invariant explains that the weight of every element $e_i$  added to the independent set $I_i$ is defined with respect to the marginal gain it adds to the set $I'_{i-1}$, and it is fixed later on. 
Intuitively, the level invariants provide the approximation guarantee.

The uniform invariant asserts that, conditioned on $\bT \geq i$ which means that $L_{i}$ is a non-empty level and $\bH_i = H_i$, which implies that $e_1, \cdots, e_{i - 1}$ are chosen and $R_i$ is well-defined, the element $\bE_i$ is uniform random variable over the set $R_i$. That is, $\Pr{\bE_i = e | \bT \geq i \text{ and } \bH_i = H_i } = \frac{1}{|R_i|} \cdot \ind{e \in R_i}$.
Intuitively, this invariant provides us with the randomness that we need  
to fool the adversary in the (fully) dynamic model 
which in turn helps us to develop a dynamic algorithm 
for the submodular matroid maximization.

\paragraph{Step 3: Query complexity.}
In the third part of the proof, 
we show that if the \textbf{uniform} invariant holds, we can bound 
the worst-case expected query complexity of the leveling algorithm, and later,  
the worst-case expected query complexity of the insertion and deletion operations. 
\paragraph{Step 4: Approximation guarantee.}
Finally, in the last step of the proof, we show that if the survivor, terminator, independent and 
weights invariants hold, we can report an independent set $I_T \in \mathcal{I}$ 
whose submodular value is an $(4+\epsilon)$-approximation of the optimal value.

\subsection{Monotone property and binary search argument}
\label{sec:binary:search}
Recall that we defined the function 
$\suit(I_j, I'_j, e, w[I_j])$ for an element $e \in V$ 
with respect to the level $L_j$ which 
\begin{itemize}
    \item returns $\emptyset$ if properties 1 and 2 hold;
    \item returns $\hat{e}$ if properties 1 and 3 hold;
    \item returns $\err$ otherwise.
\end{itemize}
Here properties 1, 2, and 3 are the ones that we defined in Definition~\ref{def:promote}.
Recall that if the first two cases occur, we say that $e$ 
is a promoting element with respect to the level $L_j$. 
In this section, we consider a boolean version of the function $\suit(I_j, I'_j, e, w[I_j])$. 
We denote this boolean function by $\boolsuit(e,L_j)$ which is 
$True$ if either of the first two cases happen. That is, 
when $\suit(I_j, I'_j, e, w[I_j])$ returns either $\emptyset$ or $\hat{e}$; 
otherwise, $\boolsuit(e,L_j)$ returns $False$.

\begin{lemma}
\label{lm:binary_search_argument}
Let $L_j$ be an arbitrary level of the Algorithm~\dynamicmatroid{}, where $1 \le j \le T$. 
Let $e$ be an arbitrary element of the ground set. 
If $\boolsuit(e,L_{j-1})$ returns $False$, then 
$\boolsuit(e,L_{j})$ returns $False$. 
\end{lemma}

Suppose for the moment that this lemma is correct. 
Then by applying a simple induction, 
we can show the function $\boolsuit(e,L_j)$ 
is monotone which means that the function $\suit(I_j, I'_j, e, w[I_j])$ is monotone. 
Thus, for every arbitrary element $e$, 
it is possible to perform a binary search on the interval $[i, \ell-1]$ 
to find the smallest $z\in [i, \ell-1]$ such that $\boolsuit(e,L_{z-1}) = True$ and $\boolsuit(e,L_z) = False$.

Now we prove the lemma. 

\begin{proof}
Suppose that $\boolsuit(e,L_{j-1})$ returns $False$. 
It means that either property 1 or both properties 2 and 3 do not hold.
If property 1 does not hold, then $f(I'_{j-1}+e)-f(I'_{j-1}) < \frac{\eps}{10k}\cdot MAX$. Since $I'_{j-1} \subseteq I'_j$ and $f$ is submodular, we have $f(I'_{j}+e)-f(I'_{j}) \leq f(I'_{j-1}+e)-f(I'_{j-1}) < \frac{\eps}{10k}\cdot MAX$, which means that $\boolsuit(e,L_j) = False$.

For the remainder of the proof, we assume that both properties 2 and 3 do not hold.
This means that $I_{j-1}+e$ is not independent, and for
the minimum weight element $\hat{e} := \arg\min_{e' \in C} w(e')$ of 
the set $C := \{ e' \in I_{j-1}: I_{j-1} + e - e' \in \mI\}$, we have  
$ {f(I'_{j-1} + e) - f(I'_{j-1})} < 2w(\hat{e})$. 
Now, let us consider level $L_{j}$. 
There are two cases to consider: $|I_{j}|=|I_{j-1}|+1$ and $|I_{j}| = |I_{j-1}|$.

For the first case, we have $I_{j} = I_{j-1} + {e_j}$.
Thus, we have $I_{j-1} \subseteq I_{j}$.
Now, let us consider the element $e$. 
For the set $C := \{ e' \in I_{j-1}: I_{j-1} + e - e' \in \mI\}$, 
we have $C \subseteq I_{j-1} \subseteq  I_{j}$ 
which means that the circuit (dependent set) $C + e \subseteq  I_{j} + {e}$. 
Note that since $I_j$ is an independent set, we also know that
$C+e$ is the only circuit of $I_j + {e}$ according to Lemma~\ref{lem:<=1circuit}. 
Recall that 
$ {f(I'_{j-1} + e) - f(I'_{j-1})} < 2w(\hat{e})$ where 
$\hat{e}$ is the minimum weight element $\hat{e} := \arg\min_{e' \in C} w(e')$.
Since $I'_{j-1} \subseteq I'_{j}$, then by the submodularity of the function $f$, 
we have 
$$
f(I'_{j} + e) - f(I'_{j}) \leq
f(I'_{j-1} + e) - f(I'_{j-1}) < 2w(e') \enspace .
$$
Hence, $\boolsuit(e,L_{j})$ returns $False$ in this case.

For the second case, we have $I_{j} = I_{j-1} - {\hat{e_j}} + e_j $. 
This means that $I_{j-1} + e_j$ is not an independent set. 
Thus, the set $C' := \{ e' \in I_{j-1}: I_{j-1} + e_j - e' \in \mI\}$ has 
a minimum weight element $\hat{e_j}$  
that is replaced by $e_j$ to obtain the independent set $I_{j}$. 

Now, we consider two subcases.
Case (I) is $\hat{e_j} \in C$ and Case (II) is $\hat{e_j} \notin C$.

\begin{figure}[h]
\begin{center}
\includegraphics{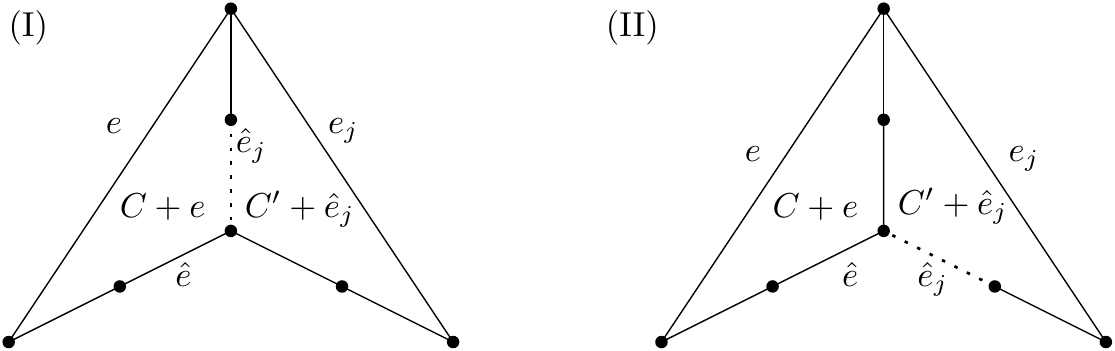}
\end{center}
\caption{Illustration of $I_j+e$ for the subcases (I) and (II) in Lemma~\ref{lm:binary_search_argument}.
$C+e$ and $C'+\hat{e}_j$ are circuits.
Case (I) is $\hat{e_j} \in C$. Then there is a circuit $C'' \subseteq (C+e)\cup (C'+e_j) -{\hat{e}_j}$.
Case (II) is $\hat{e_j} \notin C$. Then $C+e \subseteq I_j+e$.}
\label{fig:binaryS:lemma}
\end{figure}

First, we consider Case (I) which is $\hat{e_j} \in C$. 
Thus, $\hat{e_j} \in C\cap C'$. 
Note that $C \subseteq I_{j-1}$ 
and $e_j \notin I_{j-1}$, so $e_j \notin C$, which implies that $(C + e)$ and $(C' + e_j)$ are two different circuits.
Since $\hat{e_j} \in (C+e)\cap (C'+e_j)$, 
there is a circuit $C'' \subseteq (C+e) \cup (C'+e_j) - {\hat{e_j}}$ 
according to Lemma ~\ref{lem:elimination:axiom}. 
In addition, 
$(C+e) \cup (C'+e_j) \subseteq I_{j-1} + e + e_j = I_j +e + \hat{e_j}$. 
    Since $\hat{e_j} \notin C''$, we then have $C'' \subseteq I_j + e$.
    Recall that $\hat{e}$ and $\hat{e_j}$ are the minimum weight element in $C$ and $C'$, 
    respectively. Since $\hat{e_j} \in C$, then $w(\hat{e})\leq w(\hat{e_j})$. 

    Let $e''$ be the minimum weight element in $C''-e$.
    Since $C'' \subseteq (C+e) \cup (C'+e_j)$  and $w(e_j) > w(\hat{e_j})$, 
    we have $w(e'') \geq \min(w(\hat{e}), w(\hat{e_j})) = w(\hat{e})$.
    Since $I'_{j-1} \subseteq I'_{j}$ and $f$ is a submodular function, 
    we obtain the following: 

   $$
    f(I'_{j} + e) - f(I'_{j}) \leq
    f(I'_{j-1} + e) - f(I'_{j-1}) <
    2\cdot w(\hat{e}) \leq
    2\cdot w(e'') \enspace.
    $$

This essentially means that $I_{j} + e$ is not independent as $C'' \subseteq I_{j} + e$, and
     $f(I'_j + e) - f(I'_j) < 2\cdot w(e'')$, 
    where $e''$ is the minimum weight element in  $C'' - e$. 
    Thus, $\boolsuit(e,L_{j})$ returns $False$. 

Finally, we consider  Case (II) which is $\hat{e_j} \notin C$. 
In this case, $C+e \subseteq I_{j-1}-{\hat{e_j}} + e \subseteq I_j + e$. 
Note that $C+e$ is the only circuit of $ I_j + e$ by Lemma~\ref{lem:<=1circuit}. 
Recall $f(I'_{j-1} + e) - f(I'_{j-1}) < 2\cdot w(\hat{e})$ and $I'_{j-1} \subseteq I'_{j}$. 
Hence, by the submodularity of $f$ we have
$f(I'_{j} + e) - f(I'_{j}) \leq
f(I'_{j-1} + e) - f(I'_{j-1}) <
2\cdot w(\hat{e})$. 
Thus, $\boolsuit(e,L_{j})$ returns $False$ proving the lemma.
\end{proof}

\begin{lemma}\label{lm:find_min} Let $I \in \mI$ be an independent set and $e$ be an element such that $I \cup \{e\} \notin \mI$.
  Define $C := \{e': I + e - e' \in \mI\}$.
  Let $w : I \cup \{e\} \to \mathbb{R}^{\ge 0}$ be an arbitrary weight function and
  define $\hat{e} := \argmin_{e'\in C} w(e')$. The element $\hat{e}$ can be found
  using at most $O(\log(|I|))$ oracle queries.
\end{lemma}
\begin{proof}
  Let $e_1, \dots, e_{|I| + 1}$ denote an ordering of $I \cup \{e\}$
  such that $w(e_1) \ge w(e_2) \dots \ge w(e_{|I| + 1})$. 
  Let $i$ denote the smallest index such that
  $\{e_1, \dots, e_{i}\} \notin \mI$. Such an index exists
  because $\{e_1, \dots, e_{|I| + 1}\} = I \cup \{e\} \notin \mI$. 
  We claim that $\hat{e} = e_i$.
  We note that the element $e_i$ can be found
  using a binary search over $[|I| + 1]$ because
  for any $j$,
  if $\{e_1, \dots, e_j\} \notin \mI$, then
  $\{e_1, \dots, e_{j+1}\} \notin \mI$ as well.

  To prove this, we first claim that $e_i \in C$. To
  see why this holds, we first observe that
  since $\{e_1, \dots, e_{i-1}\}$ is independent but
  $\{e_1, \dots, e_i\}$ is not, we have
  $e_i \in \spn(\{e_1, \dots, e_{i-1}\}) \subseteq \spn(I + e - e_i)$.
  Therefore, since $e_j \in \spn(I + e - e_i)$ for all $j\ne i$, we have
  $I + e \subseteq \spn(I + e - e_i)$, which implies
  \begin{align*}
    \rank(I + e - e_i) \ge \rank(I + e) \ge \rank(I) = |I|
    = |I+ e - e_i|,
  \end{align*}
  which implies $I + e - e_i \in \mI$ as claimed.

  We need to show that for any $e' \in C$, we have $w(e_i) \le w(e')$. 
  Assume for contradiction that $w(e') < w(e_i)$. It follows that
  $e' = e_j$ for some $j > i$. By definition of $C$,
  we must have
  $I + e - e_j \in \mI$, which implies
  $\{e_1, \dots, e_{j-1}\} \in\mI$. Since $i < j$, this further implies
  $\{e_1, \dots, e_i\} \in  \mI$, which is not possible by definition of $i$.
\end{proof}

\subsection{Correctness of invariants after \MatroidConstLevel{} is called}
\label{sec:level:proofs:matroid}

In this section, we focus on the previously defined invariants at the end of the execution of the algorithm $\MatroidConstLevel{}(j)$. 
We first provide a definition explaining what we mean by stating that level invariants partially hold.

\begin{tcolorbox}[width=\linewidth, colback=white!80!gray,boxrule=0pt,frame hidden, sharp corners]

\begin{definition}
    For $j \ge 1$, we say that 
    \emph{the level invariants partially hold for the first $j$ levels} if the followings hold.
    \begin{enumerate}
        \item \textbf{Starter.} $R_0 = V$ and $I_0 = I'_0 = \emptyset$
        \item \textbf{Survivor.}  For $1 \leq i \le j$, $R_i = \{ e \in R_{i-1} - e_{i-1}: \replacementTester{}(I_{i-1}, I'_{i-1}, e, w[I_{i-1}]) \ne \err\}$ 
        \item  \textbf{Independent.} For $1 \leq i \le j-1$, $I_i = I_{i-1} + e_i - \replacementTester{}(I_{i-1}, I'_{i-1}, e_i, w[I_{i-1}])$, and $I'_i = \cup_{j \le i} I_j$
        \item  \textbf{Weight.} For $1 \leq i \le j - 1$, $e_i \in R_i$ and  $w(e_i) = f(I'_{i-1} + e_i) - f(I'_{i-1})$
    \end{enumerate}
\end{definition}

\end{tcolorbox}

Next, we have the following theorem, in which we ensure that all level invariants hold after the execution of $\MatroidConstLevel{}(j)$ given the assumption that level invariants partially hold for the first $j$ levels when $\MatroidConstLevel{}(j)$ is invoked. 
This theorem will be of use in the following sections in showing that level invariants hold after each update. It can also independently prove that level invariants hold after \init{} is run.

\begin{theorem}
    \label{thm:invariants:leveling}
If before calling $\MatroidConstLevel{}(j)$, the level invariants partially hold for the first $j$ levels,
then after the execution of $\MatroidConstLevel{}(j)$, level invariants fully hold.
\end{theorem}

\begin{proof}
    Considering that the starter invariant holds by the assumption of the theorem and needs no further proof, we have broken the proof of this theorem into four lemmas, each considering one of the survivor, independent, weight, and terminator invariants separately. 
    These mentioned lemmas and their proofs can be found in detail in Appendix \ref{sec:invar_appendix} as Lemmas~\ref{lm:survivor:leveling},~\ref{lm:independent:leveling},~\ref{lm:weight:leveling}, and~\ref{lm:terminator:leveling} in Section \ref{subs:thm:invariants:levelin}.
\end{proof}

Finally, we prove a lemma that says knowing that the level invariants are going to hold after the execution of $\MatroidConstLevel{}(j)$, 
a modified version of uniform invariant will also hold after this execution. We use this lemma in the next sections to prove that the uniform invariant holds after each update. It also shows that uniform invariant holds after \init{} is run since the previous theorem had proved that level invariants would hold. 

\begin{lemma}[Uniform invariant]
\label{lm:uniform:leveling}
If $\MatroidConstLevel{}(j)$ is invoked and 
the level invariants are going to hold after its execution, then for any $i \ge j$ we have $ \Pr{\bE_i = e |  \bT \geq i \text{ and } \bH_i = H_i  }= \frac{1}{|R_i|}\cdot \ind{e\in R_i}  $.
\end{lemma}

\begin{proof}
At the beginning of $\MatroidConstLevel{}(j)$, we take a random permutation of elements in $R_j$. Making a random permutation is equivalent to sampling all elements without replacement. In other words, instead of fixing a random permutation $P$ of $R_j$ and iterating through $P$ in Line \ref{line:iterate_P}, we can repeatedly sample a random element $e$ from the unseen elements of $R_j$ until we have seen all of the elements. 
Hence, in the following proof, we assume our algorithm uses sampling without replacement.

Given this view, we make the following claims.
    
    \textbf{Observation 1.}
    $e_i$ is the first element of $R_i$ seen in the permutation.~\\
            This is because before $e_i$ is seen,
            the value of $\ell$ is at most $i$. It is also clear from the algorithm that when an element $e$ is considered, it can only be added to sets $R_{x}$ for $x\le \ell$, both when $y=\err$ and when $y\ne \err$.
            Furthermore, $e$ can only be added to $R_{\ell}$ if $e=e_{\ell}$. Therefore, no element can be added to $R_i$ before $e_i$ is seen.
            
    \textbf{Observation 2.}
    Once $e_1, \dots, e_{i-1}$ have been seen, the set $R_i$ is uniquely determined.\\
    Note that $R_i$ is uniquely determined \emph{even though the algorithm has not observed its elements yet}. 
    This is because regardless of the randomness of $\MatroidConstLevel{}(j)$, the level invariants will hold after its execution. 
    This implies that the content of the set $R_i$ only depends on the value of $(\bE_1, \dots, \bE_{i - 1})$, which is not going to change after it is set to be equal to $(e_1, \dots, e_{i - 1})$.

    Let the random variable $\bM_i$ denote the sequence of elements that our algorithm observes until setting $\bE_{i-1}$ to be $e_{i - 1}$, including $e_{i-1}$ itself.
    In other words, if $e_{i-1}$ is the $x$-th element of the permutation $P$, $M_i$ is the first $x$ elements of $P$.
    ~\\
    Based on the above facts, conditioned on $\bM_i=M_i$, 
    \textbf{(a) }the value of $\bR_i$, or in other words $R_i$ is uniquely determined.
    \textbf{(b)} $e_i$ is going to be the first element of $R_i$ that the algorithm observes. Therefore,
    since we assumed that the algorithm uses sampling without replacement, $\bE_i$ is going to have a uniform distribution over $R_i$, i.e.,
    \begin{align*}
        \Pr{\bE_i = e | 
        \bT \ge i,
        \bM_i = M_i
        }
        = \frac{1}{|R_i|} \ind{e \in R_i} \enspace .
    \end{align*}

    By the law of total probability, we have 
    \begin{align*}
        \Pr{\bE_i = e_i | 
        \bT \ge i,
        \bH_i = H_i
        }
        &=
        \Exu{M_i}{
        \Pr{\bE_i = e_i | 
        \bT \ge i,
        \bH_i = H_i,
        \bM_i=M_i
        }
        } \enspace, 
    \end{align*}
    where the expectation is taken over all $M_i$ with positive probability.

    Also, note that knowing that $\bM_i = M_i$ uniquely determines the value of $\bH_i$ as well. This is because $M_i$ includes $(e_1, \dots, e_{i-1})$ and, with similar reasoning to what we used for Observation 2, we can say that $R_1, \dots, R_{i}$ are uniquely determined by $(e_1, \dots, e_{i-1})$.
    
Since we are only considering $M_i$ with positive probability, and $\bH_i$ is a function of $\bM_i$ given the discussion above, all the forms of $M_i$ that we consider in our expectation are the ones that imply $\bH_i = H_i$. Therefore, we can drop the condition $\bH_i=H_i$ from the condition $\bH_i=H_i, \bM_i = M_i$, which implies
\[
        \Pr{\bE_i = e_i | \bT \ge i, \bH_i = H_i }
        = \Exu{M_i}{ \Pr{\bE_i = e_i | \bT \ge i, \bM_i=M_i} }
        = \Exu{M_i}{ \frac{1}{|R_i|} \ind{e_i \in R_i} }
        =\frac{1}{|R_i|} \ind{e_i \in R_i} \enspace ,
\]
    as claimed.
\end{proof}

\subsection{Correctness of invariants after an update}
\label{sec:update:matroid}
In our dynamic model, 
we consider a sequence $\mathcal{S}$ of updates to the underlying ground set $V$ 
where at time $t$ of the sequence $\mathcal{S}$, we observe an update which can be the deletion of an element $e \in V$ or insertion of an element $e \in V$. 
We assume that an element $e$ can be deleted at time $t$, if it is in $V$ meaning that it was not deleted after the last time it was inserted.

We use several random variables for our analysis, including $\bE_i$, $\bR_i$, $\bT$, and $\bH_i$. Upon observing an update at time $t$, 
we should distinguish between each of these random variables and their corresponding values before and after the update.
To do so, we use the notations $\bY^-$ and $Y^{-}$ to denote a random variable and its value before time $t$ when $e$ is either deleted or inserted, and we keep using $\bY$ and $Y$ to denote them at the current time after the execution of update.
As an example,
    $\bH_i^- := (\bE_1^-, \dots, \bE_{i-1}^-, \bR_0^-,  \bR_{1}^-, \dots, \bR_{i}^-)$ 
    is the random variable that corresponds to the partial configuration $H_i^- = (e_1^-, \dots, e_{i-1}^-, R_0^-, \dots, R_{i}^-)$.

\subsubsection{Correctness of invariants after every insertion}

We first consider the case when the update at time $t$ of the sequence $\mathcal{S}$ is an insertion of an element $v$. 
In this section, we prove the following theorem.

\begin{theorem}
\label{mat_insert:invariants}
If before the insertion of an element $v$, the level invariants and uniform invariant hold, then they also hold after the execution of \insertv$(v)$. 
\end{theorem}

We break the proof of this theorem into Lemmas ~\ref{mat_insert_level} and ~\ref{mat_insert_uni}. 
Note that we use Lemma ~\ref{mat_insert_level} in the proof of Lemma ~\ref{mat_insert_uni}. However, Lemma ~\ref{mat_insert_uni} would not be used in the proof of ~\ref{mat_insert_level}, so no loop would form when combined to prove the theorem. Proof of Lemma \ref{mat_insert_level} is written in detail in Appendix \ref{sec:invar_appendix} Section \ref{subs:mat_insert_level}.

\begin{lemma} [Level invariants]
\label{mat_insert_level}
If before the insertion of an element $v$ the level invariants (i.e., starter, survivor, independent, weight, and terminator) hold, 
then they also hold after the execution of \insertv$(v)$. 
\end{lemma}

\begin{lemma} [Uniform invariant]
\label{mat_insert_uni}
If before the insertion of an element $v$ the level and uniform invariants hold, then the uniform invariant also holds after the execution of \insertv$(v)$. 
\end{lemma}

\begin{proof} 
By the assumption that the uniform invariant holds before the insertion of the element $v$, we mean that for any arbitrary $i$ and any arbitrary element $e$, the following holds:
\[
    \Pr{\bE_i^- = e | \bT^- \geq i, \bH_i^- = H_i^-} = 
    \frac{1}{|R_i^{-}|} \cdot \ind{e \in R_i^{-}} \enspace.
\]
We aim to prove that given 
our assumptions,
after the execution of \insertv$(v)$, for each arbitrary $i$ and each arbitrary element $e$, we have 
\[
    \Pr{\bE_i = e | \bT \geq i, \bH_i = H_i} = 
    \frac{1}{|R_i|} \cdot \ind{e \in R_i} \enspace .
\]

Note that $\Pr{\bE_i = e | \bT \geq i, \bH_i = H_i}$, is only defined when 
$\Pr{\bT \geq i, \bH_i = H_i} > 0$, which means that given the input and considering the behavior of our algorithm including its random choices, it is possible to reach a state where $\bT \geq i$ and $\bH_i = H_i$.
In this proof, we use $\bold{p}_i$ to denote to the variable $p_i$ used in the \insertv{} as a random variable.

Fix any arbitrary $i$ and any arbitrary element $e$. Since $\bH_i^- = (\bE_1^-, \dots, \bE_{i - 1}^-, \bR_0^-, \bR_1^-, \dots, \bR_i^-)$ refers to our data structure levels before the insertion of the element $v$, it is clear that the following facts hold about $\bH_i^-$. 

\begin{fact}
\label{ebeforeinsert}
For any $j < i$, $\bold{e}_j^- \neq v$.
\end{fact}

\begin{fact}
\label{rbeforeinsert}
For any $j \leq i$, $v \notin \bold{R}_j^-$.
\end{fact}

We consider the following cases based on which of the following holds for $H_i = (e_1, \dots, e_{i - 1}, R_0, R_1, \dots, R_i)$:
\begin{itemize}
    \item Case 1: If the $e_j = v$ for some $j < i$.
    \item Case 2: If $v \notin \{e_1, \dots, e_{i - 1}\}$.
\end{itemize}

We handle these two cases separately (Lemma~\ref{insertcase1} for the first case and Lemma~\ref{insertcase2} for the second case). We show that, no matter the case, $\Pr{\bE_i = e | \bT \geq i, \bH_i = H_i}$ is equal to $\frac{1}{|R_i|} \cdot \ind{e \in R_i}$, which completes the proof of the Lemma.

\begin{claim}
\label{insertcase1}
If $H_i$ is such that there is a $1 \le j < i$ that $e_j = v$, then  
$\Pr{\bE_i = e | \bT \geq i, \bH_i = H_i} = \frac{1}{|R_i|} \cdot \ind{e \in R_i}$.
\end{claim}

\begin{proof}
We know that, $\bold{p_j}$ 
must have been equal to 1, as otherwise, instead of having $\bold{e_j} = e_j = v$, we would have had $\bE_j = \bE_j^-$, which would not have been equal to $v$ as stated in Fact \ref{ebeforeinsert}. According to our algorithm, since $\bold{p_j}$ has been equal to $1$, we have invoked \MatroidConstLevel($j + 1$). 
By Lemma~\ref{mat_insert_level}, we know that the level invariants hold at the end of the execution of \insertv{}, which is also the end of the execution of \MatroidConstLevel{$(j + 1)$}. 
Thus, Lemma~\ref{lm:uniform:leveling}, proves that $\Pr{\bE_i = e | \bT \geq i, \bH_i = H_i} = \frac{1}{|R_i|}\cdot \ind{e \in R_i}$.
\end{proof}

\begin{claim}
\label{insertclaim}
Assume that $H_i$ is such that $e_j \neq v$ for any $1 \le j < i$ and define $H_i^-$ based on $H_i$ as $ H_i^- := 
    (R_0 \backslash \{v\} , \dots, R_i \backslash \{v\}, e_1, \dots, e_{i-1})$.
The events $[\bT \geq i, \bH_i = H_i]$ and $[\bT^- \geq i, \bH_i^- = H_i^-, \bold{p_1} = 0, \dots, \bold{p_{i - 1}} = 0]$ are equivalent 
and imply each other, thusly they are interchangeable.   
\end{claim}
\begin{proof}
First, we show that if $\bT \geq i, \bH_i = H_i$, then $\bT^- \geq i, \bH_i^- = H_i^-, \bold{p_1} = 0, \dots, \bold{p_{i - 1}} = 0$. 
Considering that case 2 holds for $H_i$, $\bH_i = H_i$, means that for any $j < i$, $\bE_i \neq v$, which means there is no $j < i$ with $\bold{p_j} = 1$. Note that if $\bold{p_j} = 1$, then we would have set $\bold{e_j}$ to be equal to $v$, and we would have invoked \MatroidConstLevel($j + 1$). Thus, in addition to knowing that for any $j < i$, $\bold{p_j} = 0$, we also know that, we have not invoked \MatroidConstLevel($j + 1$) for any $j < i$. As for any $j < i$, $\bold{p_j} = 0$ and \MatroidConstLevel($j + 1$) was not invoked, we have the following results: 
\begin{enumerate}
    \item 
    Level $i$ also existed before the insertion of $v$, i.e. $\bT^- \geq i$.
    \item 
    We have made no change in the values of $(\bold{e_1}, \dots, \bold{e_{i - 1}})$, and they still have the values they had before the insertion of $v$, i.e. for any  $j < i$, $\bold{e_j} = \bold{e_j^-}$, and so $\bold{e_j^-} = e_j$. 
    \item
    All the change we might have made in our data structure is limited to adding the element $v$ to a subset of $\{\bold{R_0^-}, \dots, \bold{R_i^-}\}$. Hence, for any $j \leq i$, whether $\bold{R_j}$ is equal to $\bold{R_j^-}$ or $\bold{R_j^-} \cup \{v\}$, $\bold{R_j^-} = \bold{R_j} \backslash \{v\} = R_j \backslash \{v\}$. 
\end{enumerate}
So far, we have proved that throughout our algorithm, we reach the state, where $\bT \geq i, \bH_i = H_i$, only if $\bT^- \geq i, \bH_i^- = H_i^-, \bold{p_1} = 0, \dots, \bold{p_{i - 1}} = 0$.

We know that in our insertion algorithm, there is not any randomness other than setting the value of $\bold{p_j}$ as long as we have not invoked \MatroidConstLevel, which only happens when for a $j$, $\bold{p_j}$ is set to be 1. It means that the value of $\bH_i$ can be determined uniquely if we know the value of $\bH_i^-$, and we know that $\bold{p_1}, \dots, \bold{p_{i - 1}}$ are all equal to $0$. 
Since we have assumed that $\bT \geq i, \bH_i = H_i$ is a valid and reachable state in our algorithm, $\bT^- \geq i, \bH_i^- = H_i^-$ must have been a reachable state as well. Plus, $\bT^- \geq i, \bH_i^- = H_i^-, \bold{p_1} = 0, \dots, \bold{p_{i - 1}} = 0$, should imply that $\bT \geq i$ and $ \bH_i = H_i$. Otherwise, $\bT \geq i, \bH_i = H_i$ could not be a reachable state, which is in contradiction with our assumption. 
\end{proof}

\begin{claim}
\label{insertcase2}
If $H_i$ is such that $e_j \neq v$ for any $1 \le j < i$, then  
$\Pr{\bE_i = e | \bT \geq i, \bH_i = H_i} = \frac{1}{|R_i|} \cdot \ind{e \in R_i}$.
\end{claim}

\begin{proof}
Define $H_i^-$ based on $H_i$ as $ H_i^- := 
    (R_0 \backslash \{v\} , \dots, R_i \backslash \{v\}, e_1, \dots, e_{i-1})$.
    ~\\
We calculate $\Pr{\bE_i = e|\bT \geq i, \bH_i = H_i}$. 
As stated above, considering that Case 2 holds for $H_i$, we know that 
$\bT \geq i, \bH_i = H_i$ implies that \MatroidConstLevel{} has not been invoked for any $j < i$. Thus, the value of $\bE_i$ will be determined based on the random variable $\bold{p_i}$. And we have: 
$$\Pr{\bE_i = e|\bT \geq i, \bH_i = H_i} = 
\sum_{p_i \in \{0, 1\}}
(\Pr{\bold{p_i} = p_i|\bT \geq i, \bH_i = H_i} \cdot
\Pr{\bE_i = e|\bT \geq i, \bH_i = H_i, \bold{p_i} = p_i}) \enspace .$$
According to the algorithm, if $v \in H_i$, then $\Pr{\bold{p_i} = 1|\bT \geq i, \bH_i = H_i}$ is equal to $\frac{1}{|R_i|}$. Otherwise, if $v \notin H_i$, then $\bold{p_i}$ would be zero by default, and $\Pr{\bold{p_i} = 1|\bT \geq i, \bH_i = H_i} = 0$. Hence, we can say that: 
$$\Pr{\bold{p_i} = 1|\bT \geq i, \bH_i = H_i} = 
\frac{1}{|R_i|} \cdot \ind{v \in R_i} \enspace.$$
Additionally, Having $\bT \geq i, \bH_i = H_i$, 
 if $\bold{p_i} = 1$, then $\bE_i$ would be $v$. 
Otherwise, if $\bold{p_i} = 0$, then $\bE_i^-$ would remain unchanged, i.e. $\bE_i = \bE_i^-$.  
Hence, $\Pr{\bE_i = e|\bT \geq i, \bH_i = H_i}$ is equal to 
$$
\frac{1}{|R_i|} \cdot \ind{v \in R_i} \cdot
\Pr{\bE_i = e|\bT \geq i, \bH_i = H_i, \bold{p_i} = 1} 
+ (1 - \frac{1}{|R_i|} \cdot \ind{v \in R_i})\cdot
\Pr{\bE_i^- = e|\bT \geq i, \bH_i = H_i, \bold{p_i} = 0}
.$$
We consider the following cases based on the value of $e$: 
\begin{itemize}
\item Case (i): $e = v$ 

In this case 
$\Pr{\bE_i = e|\bT \geq i, \bH_i = H_i, \bold{p_i} = 1} = 1$, and $\Pr{\bE_i^- = e|\bT \geq i, \bH_i = H_i, \bold{p_i} = 0} = 0$. Thus, we have: 
$$\Pr{\bE_i = e|\bT \geq i, \bH_i = H_i} = 
\frac{1}{|R_i|} \cdot \ind{v \in R_i} \cdot 1
+ (1 - \frac{1}{|R_i|} \cdot \ind{v \in R_i}) \cdot 0 = \frac{1}{|R_i|} \cdot \ind{v \in R_i} = \frac{1}{|R_i|} \cdot \ind{e \in R_i}.
 $$
\item Case (ii): $e \neq v$
In this case, $\Pr{\bE_i = e|\bT \geq i, \bH_i = H_i, \bold{p_i} = 1} = 0$. So we have: 
$$\Pr{\bE_i = e|\bT \geq i, \bH_i = H_i} =
\frac{1}{|R_i|} \cdot \ind{v \in R_i} \cdot 0
+ (1 - \frac{1}{|R_i|} \cdot \ind{v \in R_i}) \cdot 
\Pr{\bE_i^- = e|\bT \geq i, \bH_i = H_i, \bold{p_i} = 0}. $$
According to the claim that we proved beforehand, $\bT \geq i, \bH_i = H_i$ and $\bT^- \geq i, \bH_i^- = H_i^-, \bold{p_1} = 0, \dots, \bold{p_{i - 1}} = 0$ are interchangeable. So we have: 
$$\Pr{\bE_i = e|\bT \geq i, \bH_i = H_i} =
(1 - \frac{1}{|R_i|} \cdot \ind{v \in R_i}) \cdot 
\Pr{\bE_i^- = e| \bT^- \geq i, \bH_i^- = H_i^-, \bold{p_1} = 0, \dots, \bold{p_i} = 0}.
$$
Since for any $j \leq i$, $\bE_i^-$ and $\bold{p_i}$ are independent random variables, we have: 
\begin{align*}
\Pr{\bE_i = e|\bT \geq i, \bH_i = H_i} =
(1 - \frac{1}{|R_i|} \cdot \ind{v \in R_i}) \cdot 
\Pr{\bE_i^- = e| \bT^- \geq i, \bH_i^- = H_i^-} \\
= (1 - \frac{1}{|R_i|} \cdot \ind{v \in R_i}) \cdot 
\left(\frac{1}{|R_i^-|} \cdot \ind{e \in R_i^-}\right) \enspace, 
\end{align*}
where the last equality holds because of the assumption stated in Lemma. From the definition of $H_i^-$, we have $R_i^- = R_i \backslash \{v\}$. Therefore, 
$$\Pr{\bE_i = e|\bT \geq i, \bH_i = H_i} = 
\frac{|R_i| - \ind{v \in R_i}}{|R_i|} \cdot 
\left(\frac{1}{|R_i| - \ind{v \in R_i} } \cdot \ind{e \in R_i \backslash \{v\}}\right) \enspace .
$$
And since, $e \neq v$, we have: 
$$\Pr{\bE_i = e|\bT \geq i, \bH_i = H_i} =
\frac{1}{|R_i|} \cdot \ind{e \in R_i} \enspace .$$
\end{itemize}
\end{proof}
As stated before, proof of these claims completes the Lemma's proof. 
\end{proof}

\subsubsection{Correctness of invariants after every deletion}

Now, we consider the case when the update at time $t$ of the sequence $\mathcal{S}$, is a deletion of an element $v$, and prove the following theorem. 

\begin{theorem}
\label{mat_delete:invariants}
If before the deletion of an element $v$, the level invariants and the uniform invariant hold, then they also hold after the execution of \deletev$(v)$. 
\end{theorem}

Similar to Theorem~\ref{mat_insert:invariants}, we break the proof of this theorem into 
Lemmas ~\ref{mat_delete_level} and ~\ref{mat_delete_uni}. Proofs of Lemmas \ref{mat_delete_level} is given in Appendix \ref{sec:invar_appendix} Sections \ref{subs:mat_delete_level}.

\begin{lemma}[Level invariants]
\label{mat_delete_level}
If before the deletion of an element $v$ the level invariants (i.e., starter, survivor, independent, weight, and terminator) hold, 
then they also hold after the execution of \deletev$(v)$. 
\end{lemma}

\begin{lemma} [Uniform invariant]
\label{mat_delete_uni}
If before the deletion of an element $v$, the level and uniform invariants hold, then the uniform invariant also holds after the execution of \deletev$(v)$. 
\end{lemma}

\begin{proof}
In other words, we want to prove that if for any $i$ and any element $e$ 
\[
    \Pr{\bE_i^- = e | \bT^- \geq i, \bH_i^- = H_i^-} = 
    \frac{1}{|R_i^{-}|} \cdot \ind{e \in R_i^{-}} \enspace ,
\]
then, after execution \deletev$(v)$, for each $i$ and each element $e$, we have 
\[
    \Pr{\bE_i = e | \bT \geq i, \bH_i = H_i} = 
    \frac{1}{|R_i|} \cdot \ind{e \in R_i} \enspace .
\]

Fix any arbitrary $i$ and $e$. 
We define a random variable $\bold{X_i}$ attaining values from the set $ \{0, 1, 2\}$, as follows:
    \begin{enumerate}
    \item If the execution of \deletev$(v)$ has terminated after invoking \MatroidConstLevel$(j)$, then we set $\bold{X_i}$ to $2$.
    \item If the execution of \deletev$(v)$ has terminated in a level $L_{j \leq i}$ because $v \notin R^-_j$, then we set $\bold{X_i}$ to $1$.
    \item Otherwise, we set $\bold{X_i}$ to $0$. 
    That is, this case occurs if $v \in R_i^-$ and \deletev$(v)$ terminates because in a level $L_{j > i}$, 
    either $e_j=v$ or $v \notin R_j$.
\end{enumerate}

In Claims \ref{propos:mat:delete1}, \ref{propos:mat:delete2}, and \ref{sevomi},
we show that for each value $X_i \in \{0, 1, 2\}$, 
$\Pr{\bE_i = e | \bT \geq i, \bH_i = H_i, \bold{X_i} = X_i} = \frac{1}{|R_i|} \cdot \ind{e \in R_i}$. 
This would imply the statement of our Lemma and completes the proof since 
\begin{equation*}
    \Pr{\bE_i = e | \bT \geq i, \bH_i = H_i}=
\Exu{X_i\sim \bold{X_i} }{\Pr{\bE_i = e | \bT \geq i, \bH_i = H_i, \bold{X_i}=X_i}}
\end{equation*}
by the law of total probability.

\begin{claim}
\label{propos:mat:delete1}
$\Pr{\bE_i = e | \bT \geq i, \bH_i = H_i, \bold{X_i} = 0} = \frac{1}{|R_i|} \cdot \ind{e \in R_i}$.    
\end{claim}

\begin{proof}
First, we prove the following claim. 

\begin{claim}
\label{prop:not:in:x0}
If $\bold{X_i} = 0$,  then for every $j < i$, $e_j \ne v$ and $v \notin R_i$. 
\end{claim}

\begin{proof}

Since $\bold{X_i}=0$, then $\MatroidConstLevel{}(j)$ has not been invoked for any $j \le i$. 
Thus, $\bold{e_j^-} = \bold{e_j} = e_j$ for any $j < i$. 
However, if $e_j = v$ for a level index $j < i$, then $\bold{e_j^-} = v$ would have held for that $j < i$, 
which means that $\MatroidConstLevel{}(j)$ would have been executed for that $j$ . 
This contradicts the assumption that $\bold{X_i}=0$. 
Therefore, for all $j < i$, we must have  $e_j \ne v$ proving the first part of this claim.  

Next, we prove the second part. 
Since $\bold{X_i} = 0$, the algorithm \deletev$(v)$ neither has called \MatroidConstLevel{} nor it terminates its execution until level $L_i$. 
Thus, $\bold{R_i} = \bold{R_i^-} - v$, which implies that $v \notin \bold{R_i}$. 
However, if we had $v \in R_i$, then the event $[\bH_i = H_i, \bold{X_i = 0}]$ would have been impossible.
\end{proof}

Using Claim~\ref{prop:not:in:x0}, 
we know that $e_j \ne v$ for $j < i$ and $v \notin R_i$. 
However, we also know that $v \in R_j^-$ for $j \le i$. 
Thus, we can define $H_i^- = (e_1^-, \dots, e_{i-1}^-, R_0^-  , \dots, R_i^- )$ 
based on $H_i = (e_1, \dots, e_{i-1}, R_0  , \dots, R_i)$ as follows:
\begin{align*}
    H_i^- =     (e_1, \dots, e_{i-1}, R_0 \cup \{v\} , \dots, R_i \cup \{v\}) \enspace. 
\end{align*}

\begin{claim}
    \label{pro:events:eqn}
   Two events $[\bT \geq i, \bH_i = H_i, \bold{X_i} = 0]$ and $[\bT^- \geq i, \bH_i^- = H_i^-, \bE_i^- \neq v]$ are equivalent (i.e., they imply each other).
\end{claim}

\begin{proof}
We first prove that the event $[\bT \geq i, \bH_i = H_i, \bold{X_i} = 0]$ 
implies the event $[\bT^- \geq i, \bH_i^- = H_i^-,\bE_i^- \neq v]$. 
Indeed, since $\bold{X_i}=0 \ne 2$ we know that  
the algorithm $\MatroidConstLevel{}(j)$ was not invoked for any $j \le i$ and 
the element $v$ was contained in $\bold{R_j^-}$ for all $j \le i$. 
In this case, according to the algorithm \deletev$(v)$, 
we conclude that for any $j \le i$, we have $\bold{e_j^-} \ne v $ and  $\bold{e_j^-} = \bold{e_j}$, and $\bold{R_j} = \bold{R_j^-} - v$.
This means that $\bold{R_j^-} = \bold{R_j} \cup \{v\}$. 
Therefore, since $\bH_i = H_i$, 
we must have  $\bH_i^- = H_i^-$, $\bE_i^- \neq v$, and  $\bE_i^-=\bE_i$.

Next, we prove the other way around. 
That is, the event $[\bT^- \geq i, \bH_i^- = H_i^-,\bE_i^- \neq v]$ 
implies the event $[\bT \geq i, \bH_i = H_i, \bold{X_i} = 0]$. 
Indeed, since $\bH_i^- = H_i^- = (e_1, \dots, e_{i-1}, R_0 \cup \{v\} , \dots, R_i \cup \{v\})$, 
then, for any $j \leq i$, $v \in \bold{R_j^-}$ and for any $j < i$, $\bold{e_j^-} = e_j$. 

Recall from Claim~\ref{prop:not:in:x0} that 
for all $j < i$, $e_j \ne v$ and $v \notin R_i$. 
Thus, for any $j < i$, we know that $\bold{e_j^-} \ne v$. 
However, we also know that $\bE_i^- \neq v$. 
Thus, $\bold{e_j^-} \ne v$ for any $j \leq i$.
This essentially means that the algorithm \deletev$(v)$ neither invokes \MatroidConstLevel{} 
nor terminates its execution till the level $L_i$. 
This implies that $\bold{X_i} = 0$. 
On the other hand, the algorithm \deletev$(v)$ only removes $v$ from $R_i^-$ and 
does not make any change in $\bold{e_1^-, \dots, e_i^-}$. 
Thus, $\bold{R_i} = R_i^- - \{v\}= R_i \cup \{v\} - v = R_i$ and 
$\bE_i = \bE_i^-$.  Therefore,  we have $\bH_i = H_i$. 
\end{proof}

Therefore, we have the following corollary.
\begin{corollary}
\label{cor:eqn:events}
$\Pr{\bE_i = e | \bT \geq i, \bH_i = H_i, \bold{X_i} = 0} =  \Pr{\bE_i^- = e | \bT^- \geq i, \bH_i^- = H_i^-, \bE_i^- \neq v} $. 
\end{corollary}

Thus, in order to prove 
$\Pr{\bE_i = e | \bT \geq i, \bH_i = H_i, \bold{X_i} = 0} = \frac{1}{|R_i|} \cdot \ind{e \in R_i}$, 
we can prove 
\[
\Pr{\bE_i^- = e | \bT^- \geq i, \bH_i^- = H_i^-, \bE_i^- \neq v} = \frac{1}{|R_i|} \cdot \ind{e \in R_i} \enspace .
\]

Recall that the assumption of this lemma is  
$ \Pr{\bE_i^- = e | \bT^- \geq i, \bH_i^- = H_i^-} = 
  \frac{1}{|R_i^{-}|} \cdot \ind{e \in R_i^{-}} $. 
That is, conditioned on the event $[\bT^- \geq i, \bH_i^- = H_i^-]$, 
the random variable $\bE_i^- \sim U(R_i^-)$ is a uniform random variable over the set $R_i^-$. 
(i.e., the value $e_i$ of the random variable $\bE_i^-$ takes ones of the elements of the set $R_i^-$ uniformly at random.)  
However, since $X_i =0$ and using Claim~\ref{pro:events:eqn}, we have $\bE_i^- \ne v$. 
Thus, conditioned on the event $ [\bT^- \geq i, \bH_i^- = H_i^-, \bE_i^- \ne v]$, 
we have that the random variable $\bE_i^- \sim U(R_i^- \backslash \{v\} ) = U(R_i)$ 
should be a uniform random variable over the set $R_i^- \backslash \{v\} = R_i$. 
Indeed, we have

\begin{align*}
\Pr{\bE_i^- = e | \bT^- \geq i, \bH_i^- = H_i^-, \bE_i^- \neq v} 
&= \dfrac {\Pr{\bE_i^- = e, \bE_i^- \neq v| \bT^- \geq i, \bH_i^- = H_i^-}}
    {\Pr{\bE_i^- \neq v| \bT^- \geq i, \bH_i^- = H_i^-}} 
= \dfrac{\frac{1}{|R_i^{-}|} \cdot \ind{e \in R_i^{-} \backslash \{v\}}}{1 - \frac{1}{|R_i^{-}|}} \\
&= \dfrac{1}{|R_i^{-}| - 1} \cdot \ind{e \in R_i^{-} \backslash \{v\}} 
= \dfrac{1}{|R_i|} \cdot \ind{e \in R_i} \enspace ,
\end{align*}
where the second equality holds because of our assumption that the uniform invariant holds before the deletion, and the fourth invariant holds because $R_i^- = R_i \cup \{v\}$ and $v \notin R_i$ proving the case $X=0$.

\end{proof}

\begin{claim}
\label{propos:mat:delete2}
$\Pr{\bE_i = e | \bT \geq i, \bH_i = H_i, \bold{X_i} = 1} = \frac{1}{|R_i|} \cdot \ind{e \in R_i}$.
\end{claim}

\begin{proof}
We will be conditioning on possible values of $\bH_i^-$.  
\begin{align*}
    \Pr{\bE_i = e | \bT \geq i, \bH_i = H_i, X_i = 1} = \Exu{H_i^-}{\Pr{\bE_i = e | \bT \geq i, \bH_i = H_i, \bold{X_i} = 1, \bH_i^- = H_i^-}} \enspace ,
\end{align*}
where the expectation is taken over all $H_i$ for which $\Pr{\bT \geq i, \bH_i = H_i, \bold{X_i} = 1, \bH_i^- = H_i^-} > 0$.
For all such $H_i^-$, we claim that
this can be further rewritten as
$\Pr{\bT \geq i, \bH_i^- = H_i^-}$.
This is because \deletev$(v)$ is executed deterministically if it does not invoke 
the algorithm \MatroidConstLevel{}. Furthermore, the value of $\bold{X_i}$ is
deterministically determined by $\bH_i^-$.
Therefore, for any value of $H_i^-$, either
$\bH_i^-=H_i^-$ implies $\bold{X_i}\ne 1$, in which case $\Pr{\bT \geq i, \bH_i^- = H_i^-, \bold{X_i} = 1} = 0$, which is in contradiction with our assumption, or 
$\bH_i^- = H_i^-$ imply
$\bold{X_i}=1$. Therefore, for all such $H_i^-$ implies $\bold{X_i}=1$, which also means that \MatroidConstLevel{} never gets invoked, in which case $\bH_i$ is uniquely determined. Hence $\bH_i^- = H_i^-$ should also imply that $\bH_i = H_i$, as otherwise $\Pr{\bT \geq i, \bH_i^- = H_i^-, \bH_i = H_i} = 0$.
We therefore obtain: 
\begin{align*}
    \Pr{\bE_i = e | \bT \geq i, \bH_i = H_i, \bold{X_i} = 1, \bH_i^- = H_i^-} = \Pr{\bE_i = e | \bT \geq i, \bH_i^- = H_i^-} 
\end{align*}
as claimed.

Also, we know that $\bH_i = H_i, \bold{X_i} = 1$, implies that: 
\begin{align*}
    \bT^- = \bT, \enspace
    \bold{R_i^-} = \bold{R_i}, \enspace
    \bE_i^- = \bE_i, \enspace
\end{align*}
since it means that the execution of \deletev$(v)$ has terminated before level $i$, thus no change has been made for that level.
Therefore, for a $H_i^-$ used in our expectation, we know that $\bT \geq i, \bH_i^- = H_i^-$ also implies 
\begin{align*}
    \bT^- \geq i,  \enspace
    R_i^- = R_i, \enspace
    \bE_i^- = \bE_i,
\end{align*}
we have:
\begin{align*}
    \Pr{\bE_i = e | \bT \geq i, \bH_i^- = H_i^-} = \Pr{\bE_i^- = e | \bT^- \geq i, \bH_i^- = H_i^-} = \frac{1}{|R_i^{-}|} \cdot \ind{e \in R_i^{-}} = \frac{1}{|R_i|} \cdot \ind{e \in R_i} \enspace, 
\end{align*}
where the third equality holds because of our assumption that the uniform invariant holds before the deletion of element $v$.
Therefore, $\Pr{\bE_i = e | \bT \geq i, \bH_i = H_i, \bold{X_i} = 1} = \frac{1}{|R_i|} \cdot \ind{e \in R_i}$.
    
\end{proof}

\begin{claim}
\label{sevomi}
    $\Pr{\bE_i = e | \bT \geq i, \bH_i = H_i, \bold{X_i} = 2}  = \frac{1}{|R_i|}\cdot \ind{e \in R_i}$.
\end{claim}

\begin{proof}
By Lemma~\ref{mat_delete_level}, we know that the level invariants hold at the end of the execution of \deletev{}, which is also the end of the execution of \MatroidConstLevel{$(j)$}. 
Using Lemma~\ref{lm:uniform:leveling}, we know that 
since the level invariants are going to hold after the execution of $\MatroidConstLevel{}(j)$, 
for $i$ which is greater than $j$, we have: 
$$\Pr{\bE_i = e | \bT \geq i, \bH_i = H_i, \bold{X_i} = 2}  = \frac{1}{|R_i|}\cdot \ind{e \in R_i} \enspace, $$

which proves this claim. 
\end{proof}

\end{proof}

\subsection{Application of Uniform Invariant: Query complexity}
As for the query complexity of this algorithm, 
observe that checking if an element $e$ is promoting  for a level $L_{z}$ 
needs $O(\log(k))$ oracle queries 
because of Lemma~\ref{lm:find_min} and the fact that the size of the independent set $I_z$ is at most $k$.
The binary search that we perform needs $O(\log T)$ number of such suitability checks for the element $e$. 
Thus, if we initiate the leveling algorithm with a set $R_i$, 
our algorithm needs $O(|R_i|\cdot \log(k) \cdot \log(T))$ oracle queries to build the levels $L_i,\cdots, L_T$. 

\begin{lemma}
\label{lm:number_of_levels}
  The number of levels $T$ is at most $k\log(\frac{k}{\epsilon})$.
\end{lemma}
\begin{proof}
  Consider a directed graph $G$ with elements $I'_T = \{e_1,\cdots,e_T\}$ as vertices of this graph.
  For each element $e_i \in I'_T$, we know that $e_i$ is a promoting element for $L_{i-1}$, i.e. $\replacementTester(I_{i-1}, I'_{i-1}, e_i, w[I_{i-1}]) \neq \err$. Therefore, we define $parent(e_i) = \replacementTester(I_{i-1}, I'_{i-1}, e_i, w[I_{i-1}])$. This value is $\emptyset$ if $I_i = I_{i-1} + e_i$. Otherwise, if $I_i = I_{i-1} - e' + e_i$, this value would be $e'$.
  For each $e_i \in I'_T$, if $parent(e_i) \neq \emptyset$, we add an edge $e_i \to parent(e_i)$ to the graph.

  Since an element can only be replaced once, we have  $|\{e'|e' \in I'_T, parent(e') = e\}| = 1$, i.e. the in-degree of each $e \in I_T$ is at most 1.
  Furthermore, the out-degree of each vertex is $1$, because for each element $e\in I'_T$, $|parent(e)| \le 1$.
  Therefore, it follows that the graph is a union of disjoint paths and each $e_i \in I'_T$ is
  in exactly one path.

   An element $e$ is a starting element in a path (its in-degree is $0$), if and only if it has not been replaced by another element. That means, $e$ remains in $I_T$ at the end of the algorithm. Given that $|I_T| \le k$, there are at most $k$ paths in $G$.
  Furthermore, for two successive elements $(u, v)$ in the path where $parent(u) = v$, $w(u) \ge 2w(v)$.
  As the weights of all elements in $I'_T$ satisfy $w(e) \in [\frac{\epsilon}{10k}MAX, MAX]$,
  the length of each path is bounded by $\log(k/\epsilon) + 4$. 
  Consequently, it follows that the total number of vertices in the graph is at most $\mO(k\log(\frac{k}{\epsilon}))$.
\end{proof}

Next, we analyze the query complexity
of \MatroidConstLevel{}. 
\begin{lemma}
\label{lm:level_query_complexity}
  The total cost of calling $\MatroidConstLevel{}(i)$ is at most
  $\mO\left(|R_i|\log(k)\log\left(\frac{k}{\epsilon}\right)\right)$.
\end{lemma}
\begin{proof}
  Checking if an element $e$ is promoting needs $\mO(\log(k))$ query calls, because
  of Lemma~\ref{lm:find_min} and the fact that $|I|\le k$ for any $I \in \mI$. 
  The algorithm $\MatroidConstLevel{}(i)$ iterates over all elements in $R_i$.
  For each element $e$, it first calls the \replacementTester{} function, and
  select $e$ if it is a promoting element, i.e. $\replacementTester(I_{\ell-1},
  I'_{\ell-1}, e, w[I_{\ell-1}]) \ne \err$. In this case, we only need $\mO(\log(k))$
  query calls. However, if $e$ is not a promoting element, it reaches
  Line~\ref{line:const_level_matroid:binary_search} and runs the binary search
  on the interval $[i,\ell-1]$. Based on Lemma~\ref{lm:number_of_levels}, the
  length of this interval is
  $\mO\left(k\log\left(\frac{k}{\epsilon}\right)\right)$. Therefore, the number
  of steps in binary search is at most
  $\mO\left(\log\left(k\log\left(\frac{k}{\epsilon}\right)\right)\right) =
  \mO\left(\log\left(\frac{k}{\epsilon}\right)\right)$. In each step of the
  binary search, the algorithm calls $\replacementTester$ one time. Thus, for
  each element we need $\mO\left(\log(k)\log\left(\frac{k}{\epsilon}\right)\right)$,
  and for all elements, we need
  $\mO\left(|R_i|\log(k)\log\left(\frac{k}{\epsilon}\right)\right)$ query calls.
\end{proof}
\begin{lemma}
  For a specified value of $MAX$, each update operation in Algorithm \ref{alg:matroid-updates} has query complexity at most
  $\mO\left(k\log(k)\log^2\left(\frac{k}{\epsilon}\right)\right)$.
\end{lemma}
\begin{proof}
  We divide the queries made by the algorithm into two categories:
  the queries made directly by the update operations \insertv{} and \deletev{}, and the queries made indirectly, if the update triggers a call to \MatroidConstLevel{}. For the first category, the number of queries for each update is always $O(T)$ for insertion which can be bounded by $O\left(
  k\log\left(\frac{k}{\epsilon}\right)\right)$, and there are no queries made for deletion. We therefore focus on the second category.
  
  Based on uniform invariant, when we insert/delete an element, for each
  natural number $i \le T$, we call $\MatroidConstLevel{}(i)$ with probability
  $\dfrac{1}{|R_i|} \cdot \ind{e \in R_i}$ which is at most $\dfrac{1}{|R_i|}$.
  Using Lemma~\ref{lm:level_query_complexity}, the query complexity for calling
  $\MatroidConstLevel{}(i)$ is
  $\mO\left(|R_i|\log(k)\log\left(\frac{k}{\epsilon}\right)\right)$. Therefore,
  the expected number of queries caused by level $i$ is bounded by
  $\dfrac{1}{|R_i|}
  \cdot\mO\left(|R_i|\log(k)\log\left(\frac{k}{\epsilon}\right)\right) =
  \mO\left(\log(k)\log\left(\frac{k}{\epsilon}\right)\right)$. As the
  Lemma~\ref{lm:number_of_levels} bounded the number of levels by $T = O\left(
  k\log\left(\frac{k}{\epsilon}\right)\right)$, we calculate the expected
  number of query calls for each update by summing the expected number of query
  calls at each level:
  \begin{align*}
      \sum_{i=1}^T \mO\left(\log(k)\log\left(\frac{k}{\epsilon}\right)\right) \le \mO\left(k\log(k)\log^2\left(\frac{k}{\epsilon}\right)\right) \enspace.
  \end{align*}
\end{proof}

In order to obtain an algorithm that works regardless of the value of $MAX$, we guess $MAX$ up to a factor of $2$ using parallel runs.
Each element is inserted only to $\log(k/\epsilon)$ copies of the algorithm. Therefore, we obtain the total query complexity claimed in Theorem \ref{thm:matroid:query_complexity}.

\begin{theorem}
\label{thm:matroid:query_complexity}
  The expected query complexity of each insert/delete for all runs is $\mO\left(k\log(k)\log^3\left(\frac{k}{\epsilon}\right)\right)$.
\end{theorem}

\subsection{Application of Level Invariants: Approximation guarantee}
Recall that we run parallel instances of $\dynamicmatroid{}$ for different guesses of the maximum value $MAX$ such that after each update, there is a run with $\max_{e \in V_t} f(e) \in (MAX/2, MAX]$, where $V_t$ is the set of elements that have been inserted but not deleted yet. 
In this section, we only talk about the run with $\max_{e \in V_t} f(e) \in (MAX/2, MAX]$.
We prove that if the level invariants hold, then after each update the submodular value of the set $I_T$ in this run is a $(4+\epsilon)$-approximation 
of the optimal value $OPT$. Formally, we state this claim as follows: 

\begin{theorem}
\label{thm:survivor:gives:approximation}
Suppose that the level invariants hold in every run of $\dynamicmatroid{}$. 
Let $I_T$ be the independent set of the final level $L_T$ in the run with $\max_{e \in V} f(e) \in (MAX/2, MAX]$. 
Then, the set $I_T$ satisfies $(4 + \epsilon) \cdot f(I_T) \ge  OPT$, 
where $OPT = \max_{I^* \in \mathcal{I}} f(I^*)$.
\end{theorem}

To this end, we first define a few notations.
\begin{definition}
    For an element $e \in \ground$, we let $z(e)$ denote the largest $i$ such that $e \in R_i$.
    In Algorithms \ref{alg:matroid} and \ref{alg:matroid-updates}, $w(e)$ is defined for all elements $e\in I'_T$, but we need to define it for other elements as well. Therefore, if $e_{z(e)} = e$, we set $w(e)  = \marginalgain{e}{I'_{z(e) - 1}}$, to match the value defined in the Algorithm. Otherwise, we set $w(e) = \marginalgain{e}{I'_{z(e)}}$.
    For a set $E \subseteq \ground$, we define $w(E)=\sum_{e\in E} w(e)$.
\end{definition}

We split the proof of Theorem~\ref{thm:survivor:gives:approximation} into four steps.
We first (in Lemma~\ref{lm:wK_le_wI}) prove that $w(I'_T) \le 2w(I_T)$. 
Later, in Lemma~\ref{lm:wI_le_fI}  we show that the sum of the weight of the elements in $I_T$ 
is upper-bounded by the submodular function of $I_T$.  That is, $w(I_T) \le f(I_T)$. 
Recall that $OPT = \max_{I \in \mathcal{I}} f(I)$ and we used the notation 
$I^* = \arg\max_{I \in \mathcal{I}} f(I) $ 
for an independent set in $\mathcal{I}$ whose submodular value is maximum. 
In the third step of the proof of Theorem~\ref{thm:survivor:gives:approximation}, 
we show that $f(I^*) \le 2w(I_T) + w(I^*)$. We prove this in Lemma~\ref{lm:fIstar_2fI_wI}.
Our proofs for these lemmas are inspired by the analysis
in Chakrabarti and Kale~\cite{chakrabarti2015submodular} who study
the streaming version of the problem.
Finally, we show that $w(I^*) \le  2w(I_T) + \frac{\epsilon}{5} \cdot f(I^*)$.
This is proven in Lemma~\ref{lm:modular} using an argument
inspired by the analysis of Ashwinkumar~\cite{badanidiyuru2011buyback}.

Having all these tools in hand, we can then finish the proof of 
Theorem~\ref{thm:survivor:gives:approximation}. 
Indeed, we have 
    \begin{align}
        f(I^*) \overset{(a)}{\le} 2w(I_T) + w(I^*) \overset{(b)}{\le} 4w(I_T) + \frac{\epsilon}{5} \cdot f(I^*) \overset{(c)}{\le} 4f(I_T) + \frac{\epsilon}{5} \cdot f(I^*) \enspace,
    \end{align}
    where (a), (b), and (c) follow from Lemmas \ref{lm:fIstar_2fI_wI}, 
    \ref{lm:modular} and \ref{lm:wI_le_fI} respectively.
    
This essentially means that     
$f(I^*) \le \frac{4}{1-\frac{\epsilon}{5}} \cdot f(I_T)$. 
Now observe that $\frac{4}{1-\frac{\epsilon}{5}} \le 4+\epsilon$. 
Indeed, if we want to have this claim correct, we must have $20 - 4\epsilon + 5\epsilon - \epsilon^2 \ge 20$ 
which means we must have $ \epsilon(\epsilon-1) \le 0$. 
However, this is correct since $0 < \epsilon \le 1$, which finishes the proof of Theorem~\ref{thm:survivor:gives:approximation}.

Next, we prove the four steps that we explained above.

\begin{definition}[\spn]
\label{def:matroid:span}
    Let $E \subseteq V$ be a set of elements. We define $\spn(E) = \{e \in V : rank(E + e) = rank(E) \}$.
\end{definition}

\begin{lemma}
    \label{lm:wK_le_wI}
    $w(I'_T) \le 2w(I_T)$.
\end{lemma}

\begin{proof}
    We prove by induction on $i$ that
    $w(I'_i) \le 2w(I_i)$ for all $i$. Setting $i=T$ will finish the proof.~\\
    The claim holds for $i=0$ as $w(I_i) = w(I'_i) = w(\emptyset) = 0$.
    Assume that the claim holds for $i-1$, we prove it holds for $i$ as well.
    Given independent invariant, 
    $I'_{i} = I'_{i-1} + e_i$ 
    and either $I_{i} = I_{i-1} + e_i$ or
    $I_{i} = I_{i-1} + e_i - \hat{e}$ for some $\hat{e}$ satisfying
    $w(\hat{e}) \le \frac{w(e_i)}{2}$. In either case,
    \begin{align*}
        w(I_{i}) &\ge w(I_{i-1}) + w(e_i) - \frac{w(e_i)}{2}
        \ge w(I_{i-1}) + \frac{w(e_i)}{2}
        \\&\overset{(a)}{\ge} \frac{1}{2}w(I'_{i-1}) + \frac{w(e_i)}{2}
        = \frac{1}{2} w(I'_{i}) \enspace,
    \end{align*}
    where for $(a)$, we have used the induction assumption for $i-1$.
\end{proof}

\begin{lemma}
\label{lm:wI_le_fI}
    The sum of the weight of the elements in $I_T$ 
    is upper-bounded by the submodular function of $I_T$.  
    That is, $w(I_T) \le f(I_T)$. 
\end{lemma}

\begin{proof}
  For each $i \in [T]$, define
  $\wtilde{I}_i$ as
  $I_i \cap I_T$. We prove by induction on $i$ that
  $w(\wtilde{I}_i) \le f(\wtilde{I}_i)$. 
  Setting $i=T$ proves the claim.
  
  The case of $i=0$ holds trivially as
  $w(\wtilde{I}_0) = f(\wtilde{I}_0) = 0$.
  Assume that $w(\wtilde{I}_{i-1}) \le f(\wtilde{I}_{i-1})$,
  we will prove that $w(\wtilde{I}_i) \le f(\wtilde{I}_i)$.
  If $e_i \notin I_T$, then the claim holds trivially as
  $\wtilde{I}_{i} = \wtilde{I}_{i-1}$. Note that in this case, if an element has appeared in $I_{i-1}$, but it is removed from $I_i$, then it is not included in $I_T$ and hence $\wtilde{I}_{i-1}$. We therefore assume that $e_i \in I_T$. In this case,
  we note that
  \begin{align}
    \label{eq:jul5_1102}
      I'_{i-1} = \bigcup_{j \le i-1} I_{j} 
      \supseteq I_{i-1} \supseteq \wtilde{I}_{i-1} \enspace .
  \end{align}
  Therefore, 
  \begin{align*}
      w(\wtilde{I}_{i}) - w(\wtilde{I}_{i-1})
      =
      w(e_i) \overset{(a)}{=}
      \marginalgain{e}{I'_{i-1}}
      \overset{(b)}{\le}
      \marginalgain{e}{\wtilde{I}_{i-1}} = f(\wtilde{I}_i) - f(\wtilde{I}_{i-1}) \enspace,
  \end{align*}
  where for $(a)$ we have used weight invariant, and for $(b)$ we have used the definition of submodularity together with \eqref{eq:jul5_1102}.
  Summing the above inequality with the induction hypothesis
  $w(\wtilde{I}_{i-1}) \le f(\wtilde{I}_{i-1})$ proves the claim.    
\end{proof}

\begin{lemma}
\label{lm:fIstar_2fI_wI}
 Recall that $OPT = \max_{I \in \mathcal{I}} f(I)$ and we 
 used the notation $I^* = \arg\max_{I \in \mathcal{I}} f(I) $ 
 for an independent set in $\mathcal{I}$ whose submodular value is maximum. 
Then, $f(I^*) \le 2w(I_T) + w(I^*)$.
\end{lemma}

\begin{proof}
We first note that
\begin{align}
    f(I'_T) = 
    \sum_{i=1}^{T}
    f(I'_{i}) - f(I'_{i-1})
    =
    \sum_{i=1}^{T}
    f(I'_{i-1}+e_i) - f(I'_{i-1})
    \overset{(a)}{=}
    \sum_{i=1}^{T}
    w(e_i)
    {=}
    w(I'_T)
    \overset{(b)}{\le} 2w(I_T) \enspace ,
    \label{inequality:fIprime_2wI}
  \end{align}
  where $(a)$ follows from weight invariant, and $(b)$ follows from
  Lemma~\ref{lm:wK_le_wI}.

  We now bound $f(I^*)$.
  Enumerate $I^* \backslash I'_T$ as $\{e^*_1, \dots, e^*_{|I^* \backslash I'_T|}\}$ in an arbitrary order.
  Define
  $D_0 = I'_T$ and $D_i = I'_T \cup \{e^*_{1}, \dots e^*_{i}\}$.
  It is clear that $D_{i-1} \supseteq I'_T \supseteq I'_{z(e_i^*)}$.
  Therefore,
  \begin{align*}
    f(D_{i}) - f(D_{i-1}) 
    &=
    f(D_{i-1} +e^*_i) - f(D_{i-1})
    \overset{(a)}{\le}
    f(I'_{z(e^*_i)} +e^*_i) - f(I'_{z(e^*_i)})
    \overset{(b)}{=} w(e^*_i) \enspace ,
  \end{align*}
  where for $(a)$ we have used the definition of submodularity, and $(b)$ holds because $e^*_i \notin I'_T$.
  Summing over all $i$, we obtain 
  \begin{align*}
    \sum_{i=1}^{|I^* \backslash I'_T|} f(D_i) - f(D_{i-1}) &\le \sum_{i=1}^{|I^* \backslash I'_T|} w(e^*_i)\\
    f(D_{|I^* \backslash I'_T|}) - f(D_0) &\le w(I^* \backslash I'_T)\\
    &\le w(I^*) \enspace.
  \end{align*}
  Given that $D_0=I'_T$ and $D_{|I^* \backslash I'_T|}=I^* \cup I'_T$, we have
  \begin{align*}
      f(I^*) \le
      f(I^* \cup I'_T)
      \le
      f(I'_T) + w(I^*) \le 2w(I_T) + w(I^*) \enspace,
  \end{align*}
  where the last inequality follows from (\ref{inequality:fIprime_2wI}).
    
\end{proof}

\begin{lemma}
\label{lm:modular}
    $w(I^*) \le  2w(I_T) + \frac{\epsilon}{5} \cdot f(I^*)$.
\end{lemma}

We first give a sketch of the proof of Lemma~\ref{lm:modular}.

We split the $I^*$ into two parts. The first part consists of elements with weights $w(e) \le \frac{\epsilon}{10k}MAX$. As we will show, the total weight of these elements can be bounded by $\frac{\epsilon}{5} \cdot f(I^*)$. As for the second group, we will show that each element can be mapped one-to-one to an element in $I_T$ with at least half its weight. 
~\\
Formally, 
let $I^*_{W}$ consist of all the elements in $I^*$ such that
$w(e) \le \frac{\epsilon}{10k}MAX$. 
We first bound $w(I^*_{W})$:
\begin{align}
    w(I^*_{W})
    =\sum_{e \in I^*_{W}} w(e) 
       \notag
    &\le |I^*_{W}| \cdot \frac{\epsilon}{10k}\cdot MAX
       \notag
    \\&\le |I^*| \cdot \frac{\epsilon}{10k} \cdot MAX
       \notag
   \\& \le \frac{\epsilon}{10} \cdot MAX < \frac{\epsilon}{5} \cdot f(I^*) \enspace .
   \label{eq:jul5_1442}
\end{align}
The last conclusion comes from the fact that $f(I^*) \ge \max_{e \in V} f(e) \in (\frac{MAX}{2}, MAX]$.

In order to bound $w(I^* \backslash I^*_{W}$), we will use the following lemmas:

\begin{lemma}
\label{lm:matroid:e_in_span_span}
     Let sets $E_1, E_2 \subseteq V$ and elements $e_1, e_2 \in V$. If $e_1 \in \spn(E_1)$ and $e_2 \in \spn(E_2 - e_2)$, then $e_1 \in \spn((E_1 \cup E_2) - e_2)$.
\end{lemma}
\begin{proof}
  Note that if we have two sets $S_1, S_2 \subseteq E$ such that $S_1 \subseteq S_2$, then $\spn(S_1) \subseteq \spn(S_2)$ and $rank(S_1) \le rank(S_2)$.  
  Now, using Definition \ref{def:matroid:span} with the fact $E_1 \subseteq (E_1 \cup E_2)$ gives us 
  $      e_1 \in \spn(E_1) \subseteq \spn(E_1 \cup E_2)$. 
  Therefore, we have $rank((E_1 \cup E_2) + e_1) = rank(E_1 \cup E_2)$.

  In a similar way, since $E_2-e_2 \subseteq ((E_1 \cup E_2) - e_2$), we can conclude that 
  $e_2 \in \spn(E_2-e_2) \subseteq \spn((E_1 \cup E_2)-e_2)$ 
  what implies that $rank(E_1 \cup E_2) = rank((E_1 \cup E_2) - e_2)$. 
  
  By using these two results, we conclude that $rank((E_1 \cup E_2) - e_2) = rank((E_1 \cup E_2) + e_1)$. Furthermore,
  \begin{align*}
      rank((E_1 \cup E_2) - e_2) \le rank((E_1 \cup E_2) - e_2 + e_1) \le rank((E_1 \cup E_2) + e_1)
  \end{align*}
  where the first and third parts are equal. Therefore, all of them are equal and $rank((E_1 \cup E_2) - e_2 + e_1) = rank((E_1 \cup E_2) - e_2)$, which implies that $e_1 \in \spn((E_1 \cup E_2) - e_2)$.
\end{proof}

\begin{lemma}\label{lm:hall}
  There is a function 
  $N: I^* \backslash I^*_{W} \to 2^{I_T}$ such that
  for all $e \in I^* \backslash I^*_{W}$, $e \in \spn(N(e))$ and  
  for all $e' \in N(e)$,
  $w(e) \le 2w(e')$.
\end{lemma}
\begin{proof}
  Define
  \begin{math}
      \wtilde{I}_i := \{e \in I^* \backslash I^*_{W}: z(e) \le i\}.
  \end{math}
  We prove by induction on $i \in [T]$ that there is 
  a function $N_i: \wtilde{I}_i \to 2^{I_i}$ such that
  $e \in \spn(N_i(e))$ and  
  $w(e) \le 2w(e')$ for all $e' \in N_i(e)$.
  
  The induction base holds trivially as $\wtilde{I}_0 = \emptyset$. Assume the claim holds for $i-1$. We show it holds for $i$.
  Let $e$ be an element of
  $\wtilde{I}_i$. We define $N_i(e)$ based on three cases as follows.
 
  \begin{itemize}
      \item \textbf{Assume that
      $\mathbf{e \in \wtilde{I}_{i-1}}$.} 
      If $I_i = I_{i-1} + e_i$ or $I_{i} = I_{i-1} + e_i - \hat{e}$ for some $\hat{e} \notin N_{i-1}(e)$, we set
      $N_i(e) = N_{i-1}(e)$.
      $N_i$ has the desirable properties for $e$ by the induction hypothesis.
      Otherwise, assuming that
      $I_{i} = I_{i-1} + e_i - \hat{e}$,  for some $\hat{e} \in N_{i-1}(e)$.
      By Lemma~\ref{lem:<=1circuit} there is a unique circuit in $I_{i-1} + e_i$, named $C$. 
      Define $N_i(e)$ as $(N_{i -1}(e) \cup C) - \hat{e} $.
      Since $\hat{e} \in N_{i-1}(e)$, by induction hypothesis, $w(e_i) \le 2w(\hat{e})$.
      Also, given independent invariant, $\hat{e} = \replacementTester{}(I_{i-1}, I'_{i-1}, e_i, w[I_{i-1}])$, i.e. $\hat{e} \gets \arg\min_{e' \in C} w(e')$. Thus,
      $w(e_i) \le 2w(\hat{e}) \le 2w(e')$ for all $e' \in C - \hat{e}$. Moreover, $w(e_i) \le 2w(e')$ for all $e' \in N_{i-1}(e)$ by induction hypothesis. Hence, $w(e_i) \le 2w(e')$ for all $e' \in \left( \left( N_{i -1}(e) \cup C \right) - \hat{e} \right)$.
      Furthermore, since $e \in \spn(N_{i-1}(e))$ and $\hat{e} \in \spn(C - \hat{e})$, we can use Lemma~\ref{lm:matroid:e_in_span_span} to conclude that
      $e \in \spn(N_{i}(e))$. 
      \item \textbf{Assume that $\mathbf{e = e_i}$.} 
      In this case, we set $N_i(e) = e$.
      \item \textbf{If neither of the two cases above hold,}
      then $z(e)=i$ but $e\neq e_i$.
      According to the survivor invariant, $e$ is not a promoting element for $L_i$.
      It follows that
      $I_i+e$ is not independent.
      By Lemma~\ref{lem:<=1circuit} there is a unique circuit in $I_i + e$. 
      Let $C$ denote this circuit, and let $N_i(e) = C - e$. It is clear that $e \in \spn(N_i(e))$, and $w(e) \le 2w(e')$ for all $e' \in C - e$ since otherwise, $e$ would be promoting element for $L_i$.
  \end{itemize}
  Finally, we set $N = N_T$ to get the desired function.
\end{proof}

\begin{lemma}
    \label{lm:matroid:hall_property_in_matroid}
    Assume that $E, E' \subseteq V$. If $E$ be an independent set such that $E \subseteq \spn(E')$, then $|E| \le |E'|$. 
\end{lemma}
\begin{proof}
    Given that $E$ is independent, we know that $|E| = rank(E)$.
    In addition, given that $E \subseteq \spn(E')$, we have $rank(E) \le rank\left(\spn\left(E'\right)\right)$.
    Therefore, $        |E| = rank(E) \le rank\left(\spn\left(E'\right)\right) = rank(E') \le |E'|$. 
\end{proof}

\begin{proof}[Proof of Lemma~\ref{lm:modular}]
    Let $N: I^* \backslash I^*_{W} \to 2^{I_{T}}$ be the function described in Lemma~\ref{lm:hall}.
    Recall that $e \in \spn(N(e))$ for all $e \in I^* \backslash I^*_{W}$.
    This further implies that for all
    $E \subseteq I^* \backslash I^*_{W}$, 
    we have $E \subseteq \spn(N(E))$.
    Therefore, as $E$ is independent, we can use Lemma~\ref{lm:matroid:hall_property_in_matroid} to conclude that $|E| \le |N(E)|$.
    
    By Hall's marriage theorem, we conclude that there is an injection $H: I^* \backslash I^*_{W} \to I_T$ such that
    $H(e) \in N(e)$ for all $e\in I^* \backslash I^*_{W}$. Therefore, 
    \begin{align*}
        w(I^* \backslash I^*_{W})
        =
        \sum_{e \in I^* \backslash I^*_{W}}
        w(e)
        \overset{(a)}{\le}
        \sum_{e \in I^* \backslash I^*_{W}}
        2w(H(e))
        \overset{(b)}{\le} \sum_{e' \in I_T}2w(e')
        = 2w(I_T) \enspace .
    \end{align*}
    where for $(a)$ we have used the fact that
    $w(e) \le 2w(e')$ for all $e' \in N(e)$, and for $(b)$ we have used the fact that $H$ is an injection. Summing the above inequality with \eqref{eq:jul5_1442} finishes the proof.
    \begin{align*}
        w(I^*) = w(I^* \backslash I^*_{W}) + w(I^*_W) \le 2w(I_T) + \frac{\epsilon}{5} \cdot f(I^*) \enspace .
    \end{align*}
\end{proof}

\section{Parameterized dynamic algorithm for submodular maximization under cardinality constraint}
\label{sec:klogk}
In this section, we present our dynamic algorithm for the maximum submodular problem under the cardinality constraint $k$.
The pseudo-code of our algorithm is provided in Algorithm \ref{alg:cardinality:offline_constraint}. 
The overview of our dynamic algorithm is given in Section "Our contribution"~\ref{sec:contrib}. 
The analysis of this algorithm is similar to the dynamic algorithm that we designed for the matroid constraint. 
Thus, we explain it in Appendix~\ref{sec:analysis:klogk}.

\begin{algorithm}[H]
  \caption{\levelingconstraint$(k, OPT)$ }
  \label{alg:cardinality:offline_constraint}
  \begin{algorithmic}[1]
    \Function{\init}{$V$}
        \State $\tau \gets \frac{OPT}{2k}$
        \State $I_{0} \gets \emptyset$ and $ R_{0} \gets V$
        \State
        $R_{1} \gets \{ e \in R_0 : \replacementTester{}(I_0, e) = True\}$
        \State Invoke \constLevel$(i = 1)$
    \EndFunction
    
 \rule{15cm}{0.4pt} 
    \Function{\constLevel}{$i$}
        \State Let $P$ be a random permutation of elements of $R_{i}$ and $\ell \gets i$ 
        \For{$e$ in $P$}\label{line:cardinality:iterate_P}
            \If{ \replacementTester$(I_{\ell-1}, e) = True$} \label{line:cardinality:check_e_is_promoting}
                \State $e_{\ell} \gets e$, \ \  $I_{\ell} \gets I_{\ell-1} + e_{\ell}$, \ \ and $z \gets \ell$ \label{line:cardinality:set_I_l}
                \State $\ell \gets \ell + 1$ \ \ and $R_{\ell} \gets  \emptyset$\label{line:cardinality:constLevel:increase_ell} 
            \Else
                \State Run binary search to find the lowest $z \in [i, \ell-1]$ such that \replacementTester$(I_z, e) = False$ \label{line:cardinality:const_level_matroid:binary_search}
            \EndIf
                \For{$r \gets i+1$ \textbf{to} $z$}  
                    \State $R_r \gets R_r + e$. \label{line:cardinality:constlevelmatroid:addR_bs}
                \EndFor
          \EndFor
        \State \Return $T \gets \ell-1$ which is the final $\ell$ that the for-loop above returns subtracted by one
    \EndFunction
    
 \rule{15cm}{0.4pt} 
 
    \Function{\replacementTester}{$I, e$}
          \If{$\marginalgain{e}{I} \ge \tau$  and  $|I| < k$}
             \State \Return True
          \EndIf
          \State \Return False
    \EndFunction
    \end{algorithmic}
\end{algorithm}


\begin{algorithm}
    \caption{\carupdates$(k, OPT)$ }
    \begin{algorithmic}[1]
    \Function{\deletev}{$v$}
        \State $R_0 \gets R_0 - v$
        \For{ $i \gets 1$ \textbf{to} $T$}
            \If{$v \notin R_i$}
                \State \Break
            \EndIf
            \State $R_i \gets R_i - v$
            \If{$e_i = v$}
            \State Invoke $\constLevel(i)$.
            \State \Break
            \EndIf
        \EndFor
    \EndFunction
    
\rule{15cm}{0.4pt} 

    \Function{\insertv}{$v$}
        \State $R_0 \gets R_0 + v$. 
        \For{$i \gets 1$ \textbf{to} $T+1$}
            \If{\replacementTester$(I_{i-1}, v)$ = False} \label{line:cardinality:insert:break}
                \State \Break
            \EndIf
            \State $R_{i} \gets R_{i} + v$.
            \State Let $p=1$ with probability $\frac{1}{|R_i|}$, and otherwise $p=0$. \label{line:cardinality:p:insert:klogk}
            \If{$p=1$} \label{line:cardinality:insert:if}
                \State {$e_i \gets v$, \quad $I_i \gets I_{i-1} + v$} \label{line:cardinality:insert:setI}
                \State {$R_{i+1} = \{e' \in R_i: \replacementTester{}(I_i, e') = True \}$} \label{line:cardinality:insert:setRi+1}
                \State {$\constLevel{}(i+1)$
                }
                \State \Break
            \EndIf
        \EndFor
    \EndFunction
    
  \end{algorithmic}
\end{algorithm}

\paragraph{Relaxing $OPT$ assumption.}

Our dynamic algorithm assumes the optimal value 
$OPT=\max_{I^* \subseteq V: |I^*|\le k}  f(I^*)$ is given as a parameter.
However, in reality, the optimal value is not known in advance and may change after every insertion or deletion. 
To remove this assumption in Algorithm~\ref{alg:cardinality:unknown:opt}, we run parallel instances of 
our dynamic algorithm for different guesses of the optimal value $OPT_t$  at any time $t$ of the sequence $\mathcal{S}_t$, such that 
$\max_{I^* \subseteq V_t: |I^*|\le k} f(I^*) \in (OPT_t/(1+\epsilon), OPT_t]$ in one of the runs. 
Recall that $V_t$ is the set of elements that have been inserted but not deleted from the beginning of the sequence till time $t$. 
These guesses that we take are $(1+\epsilon)^i$ where $i\in \mathbb{Z}$. 
If $\rho$ is the ratio between the maximum and minimum non-zero possible value 
of a subset of $V$ with at most $k$ elements, 
then the number of parallel instances of our algorithm will be 
$\mO(\log_{1+\eps}{\rho})$. 
This incurs an extra $\mO(\log_{1+\eps}{\rho})$-factor in the query complexity of our dynamic algorithm. 

In fact, we can replace this extra factor with 
an extra factor of $\mO(\log{(k)}/\epsilon)$ which is independent of $\rho$. 
To this end, we use the well-known technique that has been also used in~\cite{DBLP:conf/nips/LattanziMNTZ20}. 
In particular, for every element $e$, we add it to those instances $i$ 
for which we have $\frac{(1+\epsilon)^i}{2k} \leq f(e) \leq (1+\epsilon)^i$.
The reason is if the optimal value of $V_t$ is within the range $((1+\epsilon)^{i-1},(1+\epsilon)^i]$ and 
$f(e) > (1+\epsilon)^i$, then $f(e)$ is greater than the optimal value and can safely be ignored for the instance $i$ 
that corresponds to the guess $(1+\epsilon)^i$. 
On the other hand, we can safely ignore all elements $e$ whose $f(e) < \frac{(1+\epsilon)^i}{2k}=\tau$, 
since these elements will never be a promoting element in the run with $OPT = (1+\epsilon)^i$.
This essentially means that every element $e$ is added to at most 
$\mO(\log_{1+\epsilon}{(2k)}) = \mO(\log{(k)}/\epsilon)$ parallel instances. 
Thus, after every insertion or deletion, 
we need to update only $\mO(\log{(k)}/\epsilon)$ instances of our dynamic algorithm.

\begin{algorithm}[h] 
\caption{Unknown $OPT$} 
\begin{algorithmic}[1]
    \State Let $\mathcal{A}_i$ be the instance of our dynamic algorithm, for which $OPT=(1+\epsilon)^i$.
    
     \rule{15cm}{0.4pt} 
    
    \Function{UpdateWithoutKnowingOPT}{$e$}
        \For{\textbf{each} $i \in \left[\ceil{\log_{1+\epsilon}{f(e)}},\floor{\log_{1+\epsilon}{\left(2k\cdot f(e)\right)}}\right]$} \Comment{$\frac{(1+\epsilon)^i}{2k} \leq f(e) \leq (1+\epsilon)^i$}
            \State Invoke $\update(e)$ for instance $\mathcal{A}_i$.
        \EndFor
    \EndFunction
\end{algorithmic}
\label{alg:cardinality:unknown:opt}
\end{algorithm}

\newcommand{\Proc}{Proceedings of the~}
\newcommand{\AMS}{Annals of Mathematical Statistics}
\newcommand{\STOC}{Annual ACM Symposium on Theory of Computing (STOC)}
\newcommand{\FOCS}{IEEE Symposium on Foundations of Computer Science (FOCS)}
\newcommand{\SODA}{Annual ACM-SIAM Symposium on Discrete Algorithms (SODA)}
\newcommand{\SOCG}{Annual Symposium on Computational Geometry (SoCG)}
\newcommand{\ICALP}{Annual International Colloquium on Automata, Languages and Programming (ICALP)}
\newcommand{\ESA}{Annual European Symposium on Algorithms (ESA)}
\newcommand{\CCC}{Annual IEEE Conference on Computational Complexity (CCC)}
\newcommand{\RANDOM}{International Workshop on Randomization and Approximation Techniques in Computer Science (RANDOM)}
\newcommand{\APPROX}{International Workshop on  Approximation Algorithms for Combinatorial Optimization Problems  (APPROX)}
\newcommand{\PODS}{ACM SIGMOD Symposium on Principles of Database Systems (PODS)}
\newcommand{\SSDBM}{ International Conference on Scientific and Statistical Database Management (SSDBM)}
\newcommand{\ALENEX}{Workshop on Algorithm Engineering and Experiments (ALENEX)}
\newcommand{\BEATCS}{Bulletin of the European Association for Theoretical Computer Science (BEATCS)}
\newcommand{\CCCG}{Canadian Conference on Computational Geometry (CCCG)}
\newcommand{\CIAC}{Italian Conference on Algorithms and Complexity (CIAC)}
\newcommand{\COCOON}{Annual International Computing Combinatorics Conference (COCOON)}
\newcommand{\COLT}{Annual Conference on Learning Theory (COLT)}
\newcommand{\COMPGEOM}{Annual ACM Symposium on Computational Geometry}
\newcommand{\DCGEOM}{Discrete \& Computational Geometry}
\newcommand{\DISC}{International Symposium on Distributed Computing (DISC)}
\newcommand{\ECCC}{Electronic Colloquium on Computational Complexity (ECCC)}
\newcommand{\FSTTCS}{Foundations of Software Technology and Theoretical Computer Science (FSTTCS)}
\newcommand{\ICCCN}{IEEE International Conference on Computer Communications and Networks (ICCCN)}
\newcommand{\ICDCS}{International Conference on Distributed Computing Systems (ICDCS)}
\newcommand{\VLDB}{ International Conference on Very Large Data Bases (VLDB)}
\newcommand{\IJCGA}{International Journal of Computational Geometry and Applications}
\newcommand{\INFOCOM}{IEEE INFOCOM}
\newcommand{\IPCO}{International Integer Programming and Combinatorial Optimization Conference (IPCO)}
\newcommand{\ISAAC}{International Symposium on Algorithms and Computation (ISAAC)}
\newcommand{\ISTCS}{Israel Symposium on Theory of Computing and Systems (ISTCS)}
\newcommand{\JACM}{Journal of the ACM}
\newcommand{\LNCS}{Lecture Notes in Computer Science}
\newcommand{\RSA}{Random Structures and Algorithms}
\newcommand{\SPAA}{Annual ACM Symposium on Parallel Algorithms and Architectures (SPAA)}
\newcommand{\STACS}{Annual Symposium on Theoretical Aspects of Computer Science (STACS)}
\newcommand{\SWAT}{Scandinavian Workshop on Algorithm Theory (SWAT)}
\newcommand{\TALG}{ACM Transactions on Algorithms}
\newcommand{\UAI}{Conference on Uncertainty in Artificial Intelligence (UAI)}
\newcommand{\WADS}{Workshop on Algorithms and Data Structures (WADS)}
\newcommand{\SICOMP}{SIAM Journal on Computing}
\newcommand{\JCSS}{Journal of Computer and System Sciences}
\newcommand{\JASIS}{Journal of the American society for information science}
\newcommand{\PMS}{ Philosophical Magazine Series}
\newcommand{\ML}{Machine Learning}
\newcommand{\DCG}{Discrete and Computational Geometry}
\newcommand{\TODS}{ACM Transactions on Database Systems (TODS)}
\newcommand{\PHREV}{Physical Review E}
\newcommand{\NATS}{National Academy of Sciences}
\newcommand{\MPHy}{Reviews of Modern Physics}
\newcommand{\NRG}{Nature Reviews : Genetics}
\newcommand{\BullAMS}{Bulletin (New Series) of the American Mathematical Society}
\newcommand{\AMSM}{The American Mathematical Monthly}
\newcommand{\JAM}{SIAM Journal on Applied Mathematics}
\newcommand{\JDM}{SIAM Journal of  Discrete Math}
\newcommand{\JASM}{Journal of the American Statistical Association}
\newcommand{\JALG}{Journal of Algorithms}
\newcommand{\TIT}{IEEE Transactions on Information Theory}
\newcommand{\CM}{Contemporary Mathematics}
\newcommand{\JC}{Journal of Complexity}
\newcommand{\TSE}{IEEE Transactions on Software Engineering}
\newcommand{\TNDE}{IEEE Transactions on Knowledge and Data Engineering}
\newcommand{\JIC}{Journal Information and Computation}
\newcommand{\ToC}{Theory of Computing}
\newcommand{\MST}{Mathematical Systems Theory}
\newcommand{\Com}{Combinatorica}
\newcommand{\NC}{Neural Computation}
\newcommand{\TAP}{The Annals of Probability}
\newcommand{\TCS}{Theoretical Computer Science}
\newcommand{\IPL}{Information Processing Letter}
\newcommand{\Algorithmica}{Algorithmica}

\section{Acknowledgements}
The work is partially supported by DARPA QuICC, NSF AF:Small  \#2218678, and  NSF AF:Small  \#2114269.

\bibliographystyle{abbrv}
\bibliography{references}

\begin{thebibliography}{100}

\bibitem{DBLP:conf/stoc/AbboudA0PS19}
A.~Abboud, R.~Addanki, F.~Grandoni, D.~Panigrahi, and B.~Saha.
\newblock Dynamic set cover: improved algorithms and lower bounds.
\newblock In {\em Proceedings of the 51st Annual {ACM} {SIGACT} Symposium on Theory of Computing, {STOC} 2019, Phoenix, AZ, USA, June 23-26, 2019}, pages 114--125. {ACM}, 2019.

\bibitem{DBLP:conf/soda/AbrahamCK17}
I.~Abraham, S.~Chechik, and S.~Krinninger.
\newblock Fully dynamic all-pairs shortest paths with worst-case update-time revisited.
\newblock In {\em Proceedings of the Twenty-Eighth Annual {ACM-SIAM} Symposium on Discrete Algorithms, {SODA} 2017, Barcelona, Spain, Hotel Porta Fira, January 16-19}, pages 440--452. {SIAM}, 2017.

\bibitem{DBLP:conf/focs/AbrahamDKKP16}
I.~Abraham, D.~Durfee, I.~Koutis, S.~Krinninger, and R.~Peng.
\newblock On fully dynamic graph sparsifiers.
\newblock In {\em {IEEE} 57th Annual Symposium on Foundations of Computer Science, {FOCS} 2016, 9-11 October 2016, Hyatt Regency, New Brunswick, New Jersey, {USA}}, pages 335--344. {IEEE} Computer Society, 2016.

\bibitem{DBLP:conf/wsdm/AgrawalGHI09}
R.~Agrawal, S.~Gollapudi, A.~Halverson, and S.~Ieong.
\newblock Diversifying search results.
\newblock In {\em Proceedings of the Second International Conference on Web Search and Web Data Mining, {WSDM} 2009, Barcelona, Spain, February 9-11, 2009}, pages 5--14. {ACM}, 2009.

\bibitem{DBLP:conf/soda/AhnGM12}
K.~J. Ahn, S.~Guha, and A.~McGregor.
\newblock Analyzing graph structure via linear measurements.
\newblock In {\em Proceedings of the Twenty-Third Annual {ACM-SIAM} Symposium on Discrete Algorithms, {SODA} 2012, Kyoto, Japan, January 17-19, 2012}, pages 459--467. {SIAM}, 2012.

\bibitem{DBLP:conf/icalp/AlalufEFNS20}
N.~Alaluf, A.~Ene, M.~Feldman, H.~L. Nguyen, and A.~Suh.
\newblock Optimal streaming algorithms for submodular maximization with cardinality constraints.
\newblock In {\em 47th International Colloquium on Automata, Languages, and Programming, {ICALP} 2020, July 8-11, 2020, Saarbr{\"{u}}cken, Germany (Virtual Conference)}, volume 168 of {\em LIPIcs}, pages 6:1--6:19. Schloss Dagstuhl - Leibniz-Zentrum f{\"{u}}r Informatik, 2020.

\bibitem{DBLP:conf/soda/AssadiK18}
S.~Assadi and S.~Khanna.
\newblock Tight bounds on the round complexity of the distributed maximum coverage problem.
\newblock In {\em Proceedings of the Twenty-Ninth Annual {ACM-SIAM} Symposium on Discrete Algorithms, {SODA} 2018, New Orleans, LA, USA, January 7-10, 2018}, pages 2412--2431. {SIAM}, 2018.

\bibitem{DBLP:conf/stoc/AssadiOSS18}
S.~Assadi, K.~Onak, B.~Schieber, and S.~Solomon.
\newblock Fully dynamic maximal independent set with sublinear update time.
\newblock In {\em Proceedings of the 50th Annual {ACM} {SIGACT} Symposium on Theory of Computing, {STOC} 2018, Los Angeles, CA, USA, June 25-29, 2018}, pages 815--826. {ACM}, 2018.

\bibitem{DBLP:journals/jacm/BabaioffIKK18}
M.~Babaioff, N.~Immorlica, D.~Kempe, and R.~Kleinberg.
\newblock Matroid secretary problems.
\newblock {\em J. {ACM}}, 65(6):35:1--35:26, 2018.

\bibitem{DBLP:conf/kdd/BadanidiyuruMKK14}
A.~Badanidiyuru, B.~Mirzasoleiman, A.~Karbasi, and A.~Krause.
\newblock Streaming submodular maximization: massive data summarization on the fly.
\newblock In {\em The 20th {ACM} {SIGKDD} International Conference on Knowledge Discovery and Data Mining, {KDD} '14, New York, NY, {USA} - August 24 - 27, 2014}, pages 671--680. {ACM}, 2014.

\bibitem{badanidiyuru2011buyback}
A.~Badanidiyuru~Varadaraja.
\newblock Buyback problem-approximate matroid intersection with cancellation costs.
\newblock In {\em International Colloquium on Automata, Languages, and Programming}, pages 379--390. Springer, 2011.

\bibitem{DBLP:conf/stoc/BalcanH11}
M.~Balcan and N.~J.~A. Harvey.
\newblock Learning submodular functions.
\newblock In {\em Proceedings of the 43rd {ACM} Symposium on Theory of Computing, {STOC} 2011, San Jose, CA, USA, 6-8 June 2011}, pages 793--802. {ACM}, 2011.

\bibitem{DBLP:conf/pkdd/BalcanH12}
M.~Balcan and N.~J.~A. Harvey.
\newblock Learning submodular functions.
\newblock In {\em Machine Learning and Knowledge Discovery in Databases - European Conference, {ECML} {PKDD} 2012, Bristol, UK, September 24-28, 2012. Proceedings, Part {II}}, volume 7524 of {\em Lecture Notes in Computer Science}, pages 846--849. Springer, 2012.

\bibitem{DBLP:conf/stoc/BalkanskiRS17}
E.~Balkanski, A.~Rubinstein, and Y.~Singer.
\newblock The limitations of optimization from samples.
\newblock In {\em Proceedings of the 49th Annual {ACM} {SIGACT} Symposium on Theory of Computing, {STOC} 2017, Montreal, QC, Canada, June 19-23, 2017}, pages 1016--1027. {ACM}, 2017.

\bibitem{DBLP:conf/stoc/BalkanskiS18}
E.~Balkanski and Y.~Singer.
\newblock The adaptive complexity of maximizing a submodular function.
\newblock In {\em Proceedings of the 50th Annual {ACM} {SIGACT} Symposium on Theory of Computing, {STOC} 2018, Los Angeles, CA, USA, June 25-29, 2018}, pages 1138--1151. {ACM}, 2018.

\bibitem{DBLP:conf/icml/BalkanskiS18}
E.~Balkanski and Y.~Singer.
\newblock Approximation guarantees for adaptive sampling.
\newblock In {\em Proceedings of the 35th International Conference on Machine Learning, {ICML} 2018, Stockholmsm{\"{a}}ssan, Stockholm, Sweden, July 10-15, 2018}, volume~80 of {\em Proceedings of Machine Learning Research}, pages 393--402. {PMLR}, 2018.

\bibitem{pmlr-v202-banihashem23a}
K.~Banihashem, L.~Biabani, S.~Goudarzi, M.~Hajiaghayi, P.~Jabbarzade, and M.~Monemizadeh.
\newblock Dynamic constrained submodular optimization with polylogarithmic update time.
\newblock In {\em Proceedings of the 40th International Conference on Machine Learning}, volume 202 of {\em Proceedings of Machine Learning Research}, pages 1660--1691. PMLR, 23--29 Jul 2023.

\bibitem{nonbanihashem2023dynamic}
K.~Banihashem, L.~Biabani, S.~Goudarzi, M.~Hajiaghayi, P.~Jabbarzade, and M.~Monemizadeh.
\newblock Dynamic non-monotone submodular maximization.
\newblock In {\em Thirty-seventh Conference on Neural Information Processing Systems}, 2023.

\bibitem{DBLP:journals/pami/BarinovaLK12}
O.~Barinova, V.~S. Lempitsky, and P.~Kohli.
\newblock On detection of multiple object instances using hough transforms.
\newblock {\em {IEEE} Trans. Pattern Anal. Mach. Intell.}, 34(9):1773--1784, 2012.

\bibitem{DBLP:conf/esa/Baswana06}
S.~Baswana.
\newblock Dynamic algorithms for graph spanners.
\newblock In {\em Algorithms - {ESA} 2006, 14th Annual European Symposium, Zurich, Switzerland, September 11-13, 2006, Proceedings}, volume 4168 of {\em Lecture Notes in Computer Science}, pages 76--87. Springer, 2006.

\bibitem{DBLP:journals/talg/BaswanaKS12}
S.~Baswana, S.~Khurana, and S.~Sarkar.
\newblock Fully dynamic randomized algorithms for graph spanners.
\newblock {\em {ACM} Trans. Algorithms}, 8(4):35:1--35:51, 2012.

\bibitem{DBLP:conf/icml/BateniCEFMR19}
M.~Bateni, L.~Chen, H.~Esfandiari, T.~Fu, V.~S. Mirrokni, and A.~Rostamizadeh.
\newblock Categorical feature compression via submodular optimization.
\newblock In {\em Proceedings of the 36th International Conference on Machine Learning, {ICML} 2019, 9-15 June 2019, Long Beach, California, {USA}}, volume~97 of {\em Proceedings of Machine Learning Research}, pages 515--523. {PMLR}, 2019.

\bibitem{DBLP:journals/talg/BateniHZ13}
M.~Bateni, M.~T. Hajiaghayi, and M.~Zadimoghaddam.
\newblock Submodular secretary problem and extensions.
\newblock {\em {ACM} Trans. Algorithms}, 9(4):32:1--32:23, 2013.

\bibitem{DBLP:journals/corr/abs-2207-07607}
S.~Behnezhad.
\newblock Dynamic algorithms for maximum matching size.
\newblock {\em CoRR}, abs/2207.07607, 2022.

\bibitem{DBLP:conf/focs/BehnezhadDHSS19}
S.~Behnezhad, M.~Derakhshan, M.~Hajiaghayi, C.~Stein, and M.~Sudan.
\newblock Fully dynamic maximal independent set with polylogarithmic update time.
\newblock In {\em 60th {IEEE} Annual Symposium on Foundations of Computer Science, {FOCS} 2019, Baltimore, Maryland, USA, November 9-12, 2019}, pages 382--405. {IEEE} Computer Society, 2019.

\bibitem{DBLP:conf/focs/Bernstein09}
A.~Bernstein.
\newblock Fully dynamic {(2} + epsilon) approximate all-pairs shortest paths with fast query and close to linear update time.
\newblock In {\em 50th Annual {IEEE} Symposium on Foundations of Computer Science, {FOCS} 2009, October 25-27, 2009, Atlanta, Georgia, {USA}}, pages 693--702. {IEEE} Computer Society, 2009.

\bibitem{DBLP:phd/us/Bernstein16}
A.~Bernstein.
\newblock {\em Dynamic Algorithms for Shortest Paths and Matching}.
\newblock PhD thesis, Columbia University, {USA}, 2016.

\bibitem{DBLP:reference/algo/Bernstein16a}
A.~Bernstein.
\newblock Dynamic approximate-apsp.
\newblock In {\em Encyclopedia of Algorithms}, pages 602--605. 2016.

\bibitem{DBLP:conf/stoc/BernsteinDL21}
A.~Bernstein, A.~Dudeja, and Z.~Langley.
\newblock A framework for dynamic matching in weighted graphs.
\newblock In {\em {STOC} '21: 53rd Annual {ACM} {SIGACT} Symposium on Theory of Computing, Virtual Event, Italy, June 21-25, 2021}, pages 668--681. {ACM}, 2021.

\bibitem{DBLP:journals/talg/BernsteinFH21}
A.~Bernstein, S.~Forster, and M.~Henzinger.
\newblock A deamortization approach for dynamic spanner and dynamic maximal matching.
\newblock {\em {ACM} Trans. Algorithms}, 17(4):29:1--29:51, 2021.

\bibitem{DBLP:conf/icalp/BernsteinS15}
A.~Bernstein and C.~Stein.
\newblock Fully dynamic matching in bipartite graphs.
\newblock In {\em Automata, Languages, and Programming - 42nd International Colloquium, {ICALP} 2015, Kyoto, Japan, July 6-10, 2015, Proceedings, Part {I}}, volume 9134 of {\em Lecture Notes in Computer Science}, pages 167--179. Springer, 2015.

\bibitem{DBLP:conf/icalp/BernsteinBGNSS022}
A.~Bernstein, J.~van~den Brand, M.~P. Gutenberg, D.~Nanongkai, T.~Saranurak, A.~Sidford, and H.~Sun.
\newblock Fully-dynamic graph sparsifiers against an adaptive adversary.
\newblock In {\em 49th International Colloquium on Automata, Languages, and Programming, {ICALP} 2022, July 4-8, 2022, Paris, France}, volume 229 of {\em LIPIcs}, pages 20:1--20:20. Schloss Dagstuhl - Leibniz-Zentrum f{\"{u}}r Informatik, 2022.

\bibitem{DBLP:journals/talg/BhattacharyaGKL22}
S.~Bhattacharya, F.~Grandoni, J.~Kulkarni, Q.~C. Liu, and S.~Solomon.
\newblock Fully dynamic ({\(\Delta\)} +1)-coloring in \emph{O}(1) update time.
\newblock {\em {ACM} Trans. Algorithms}, 18(2):10:1--10:25, 2022.

\bibitem{DBLP:conf/soda/BhattacharyaHN17}
S.~Bhattacharya, M.~Henzinger, and D.~Nanongkai.
\newblock Fully dynamic approximate maximum matching and minimum vertex cover in \emph{O}(log\({}^{\mbox{3}}\) \emph{n}) worst case update time.
\newblock In {\em Proceedings of the Twenty-Eighth Annual {ACM-SIAM} Symposium on Discrete Algorithms, {SODA} 2017, Barcelona, Spain, Hotel Porta Fira, January 16-19}, pages 470--489. {SIAM}, 2017.

\bibitem{DBLP:conf/focs/BhattacharyaHN19}
S.~Bhattacharya, M.~Henzinger, and D.~Nanongkai.
\newblock A new deterministic algorithm for dynamic set cover.
\newblock In {\em 60th {IEEE} Annual Symposium on Foundations of Computer Science, {FOCS} 2019, Baltimore, Maryland, USA, November 9-12, 2019}, pages 406--423. {IEEE} Computer Society, 2019.

\bibitem{DBLP:conf/stoc/BhattacharyaHNT15}
S.~Bhattacharya, M.~Henzinger, D.~Nanongkai, and C.~E. Tsourakakis.
\newblock Space- and time-efficient algorithm for maintaining dense subgraphs on one-pass dynamic streams.
\newblock In {\em Proceedings of the Forty-Seventh Annual {ACM} on Symposium on Theory of Computing, {STOC} 2015, Portland, OR, USA, June 14-17, 2015}, pages 173--182. {ACM}, 2015.

\bibitem{DBLP:conf/soda/BhattacharyaHNW21}
S.~Bhattacharya, M.~Henzinger, D.~Nanongkai, and X.~Wu.
\newblock Dynamic set cover: Improved amortized and worst-case update time.
\newblock In {\em Proceedings of the 2021 {ACM-SIAM} Symposium on Discrete Algorithms, {SODA} 2021, Virtual Conference, January 10 - 13, 2021}, pages 2537--2549. {SIAM}, 2021.

\bibitem{DBLP:conf/icalp/BhattacharyaK21}
S.~Bhattacharya and P.~Kiss.
\newblock Deterministic rounding of dynamic fractional matchings.
\newblock In {\em 48th International Colloquium on Automata, Languages, and Programming, {ICALP} 2021, July 12-16, 2021, Glasgow, Scotland (Virtual Conference)}, volume 198 of {\em LIPIcs}, pages 27:1--27:14. Schloss Dagstuhl - Leibniz-Zentrum f{\"{u}}r Informatik, 2021.

\bibitem{DBLP:journals/corr/abs-2207-07438}
S.~Bhattacharya, P.~Kiss, T.~Saranurak, and D.~Wajc.
\newblock Dynamic matching with better-than-2 approximation in polylogarithmic update time.
\newblock {\em CoRR}, abs/2207.07438, 2022.

\bibitem{DBLP:conf/esa/BodwinK16}
G.~Bodwin and S.~Krinninger.
\newblock Fully dynamic spanners with worst-case update time.
\newblock In {\em 24th Annual European Symposium on Algorithms, {ESA} 2016, August 22-24, 2016, Aarhus, Denmark}, volume~57 of {\em LIPIcs}, pages 17:1--17:18. Schloss Dagstuhl - Leibniz-Zentrum f{\"{u}}r Informatik, 2016.

\bibitem{DBLP:conf/soda/BuchbinderFS15a}
N.~Buchbinder, M.~Feldman, and R.~Schwartz.
\newblock Online submodular maximization with preemption.
\newblock In {\em Proceedings of the Twenty-Sixth Annual {ACM-SIAM} Symposium on Discrete Algorithms, {SODA} 2015, San Diego, CA, USA, January 4-6, 2015}, pages 1202--1216. {SIAM}, 2015.

\bibitem{DBLP:conf/stoc/CarterW77}
L.~Carter and M.~N. Wegman.
\newblock Universal classes of hash functions (extended abstract).
\newblock In {\em Proceedings of the 9th Annual {ACM} Symposium on Theory of Computing, May 4-6, 1977, Boulder, Colorado, {USA}}, pages 106--112, 1977.

\bibitem{chakrabarti2015submodular}
A.~Chakrabarti and S.~Kale.
\newblock Submodular maximization meets streaming: Matchings, matroids, and more.
\newblock {\em Mathematical Programming}, 154(1):225--247, 2015.

\bibitem{DBLP:conf/icalp/CharikarS18}
M.~Charikar and S.~Solomon.
\newblock Fully dynamic almost-maximal matching: Breaking the polynomial worst-case time barrier.
\newblock In {\em 45th International Colloquium on Automata, Languages, and Programming, {ICALP} 2018, July 9-13, 2018, Prague, Czech Republic}, volume 107 of {\em LIPIcs}, pages 33:1--33:14. Schloss Dagstuhl - Leibniz-Zentrum f{\"{u}}r Informatik, 2018.

\bibitem{DBLP:conf/aaai/ChaturvediNZ21}
A.~Chaturvedi, H.~L. Nguyen, and L.~Zakynthinou.
\newblock Differentially private decomposable submodular maximization.
\newblock In {\em Thirty-Fifth {AAAI} Conference on Artificial Intelligence, {AAAI} 2021, Thirty-Third Conference on Innovative Applications of Artificial Intelligence, {IAAI} 2021, The Eleventh Symposium on Educational Advances in Artificial Intelligence, {EAAI} 2021, Virtual Event, February 2-9, 2021}, pages 6984--6992. {AAAI} Press, 2021.

\bibitem{DBLP:conf/focs/ChechikZ19}
S.~Chechik and T.~Zhang.
\newblock Fully dynamic maximal independent set in expected poly-log update time.
\newblock In {\em 60th {IEEE} Annual Symposium on Foundations of Computer Science, {FOCS} 2019, Baltimore, Maryland, USA, November 9-12, 2019}, pages 370--381. {IEEE} Computer Society, 2019.

\bibitem{DBLP:conf/soda/ChechikZ20}
S.~Chechik and T.~Zhang.
\newblock Dynamic low-stretch spanning trees in subpolynomial time.
\newblock In {\em Proceedings of the 2020 {ACM-SIAM} Symposium on Discrete Algorithms, {SODA} 2020, Salt Lake City, UT, USA, January 5-8, 2020}, pages 463--475. {SIAM}, 2020.

\bibitem{DBLP:conf/icalp/ChekuriGQ15}
C.~Chekuri, S.~Gupta, and K.~Quanrud.
\newblock Streaming algorithms for submodular function maximization.
\newblock In {\em Automata, Languages, and Programming - 42nd International Colloquium, {ICALP} 2015, Kyoto, Japan, July 6-10, 2015, Proceedings, Part {I}}, volume 9134 of {\em Lecture Notes in Computer Science}, pages 318--330. Springer, 2015.

\bibitem{DBLP:conf/stoc/ChekuriQ19}
C.~Chekuri and K.~Quanrud.
\newblock Parallelizing greedy for submodular set function maximization in matroids and beyond.
\newblock In {\em Proceedings of the 51st Annual {ACM} {SIGACT} Symposium on Theory of Computing, {STOC} 2019, Phoenix, AZ, USA, June 23-26, 2019}, pages 78--89. {ACM}, 2019.

\bibitem{DBLP:conf/soda/ChekuriVZ11}
C.~Chekuri, J.~Vondr{\'{a}}k, and R.~Zenklusen.
\newblock Multi-budgeted matchings and matroid intersection via dependent rounding.
\newblock In {\em Proceedings of the Twenty-Second Annual {ACM-SIAM} Symposium on Discrete Algorithms, {SODA} 2011, San Francisco, California, USA, January 23-25, 2011}, pages 1080--1097. {SIAM}, 2011.

\bibitem{DBLP:conf/focs/ChenGHPS20}
L.~Chen, G.~Goranci, M.~Henzinger, R.~Peng, and T.~Saranurak.
\newblock Fast dynamic cuts, distances and effective resistances via vertex sparsifiers.
\newblock In {\em 61st {IEEE} Annual Symposium on Foundations of Computer Science, {FOCS} 2020, Durham, NC, USA, November 16-19, 2020}, pages 1135--1146. {IEEE}, 2020.

\bibitem{DBLP:journals/corr/abs-2111-03198}
X.~Chen and B.~Peng.
\newblock On the complexity of dynamic submodular maximization.
\newblock In {\em Proceedings of the Fifty-Fourth Annual {ACM} on Symposium on Theory of Computing, {STOC} 2022, to appear}, 2022.

\bibitem{DBLP:conf/icml/ChenSMKWK14}
Y.~Chen, H.~Shioi, C.~F. Montesinos, L.~P. Koh, S.~A. Wich, and A.~Krause.
\newblock Active detection via adaptive submodularity.
\newblock In {\em Proceedings of the 31th International Conference on Machine Learning, {ICML} 2014, Beijing, China, 21-26 June 2014}, volume~32 of {\em {JMLR} Workshop and Conference Proceedings}, pages 55--63. JMLR.org, 2014.

\bibitem{DBLP:conf/soda/ChitnisCEHMMV16}
R.~Chitnis, G.~Cormode, H.~Esfandiari, M.~Hajiaghayi, A.~McGregor, M.~Monemizadeh, and S.~Vorotnikova.
\newblock Kernelization via sampling with applications to finding matchings and related problems in dynamic graph streams.
\newblock In {\em Proceedings of the Twenty-Seventh Annual {ACM-SIAM} Symposium on Discrete Algorithms, {SODA} 2016, Arlington, VA, USA, January 10-12, 2016}, pages 1326--1344. {SIAM}, 2016.

\bibitem{doi:10.1137/15M1032077}
R.~Chitnis, M.~Cygan, M.~Hajiaghayi, M.~Pilipczuk, and M.~Pilipczuk.
\newblock Designing fpt algorithms for cut problems using randomized contractions.
\newblock {\em SIAM Journal on Computing}, 45(4):1171--1229, 2016.

\bibitem{DBLP:conf/soda/ChitnisCHM15}
R.~H. Chitnis, G.~Cormode, M.~T. Hajiaghayi, and M.~Monemizadeh.
\newblock Parameterized streaming: Maximal matching and vertex cover.
\newblock In {\em Proceedings of the Twenty-Sixth Annual {ACM-SIAM} Symposium on Discrete Algorithms, {SODA} 2015, San Diego, CA, USA, January 4-6, 2015}, pages 1234--1251. {SIAM}, 2015.

\bibitem{DBLP:conf/focs/ChuzhoyGLNPS20}
J.~Chuzhoy, Y.~Gao, J.~Li, D.~Nanongkai, R.~Peng, and T.~Saranurak.
\newblock A deterministic algorithm for balanced cut with applications to dynamic connectivity, flows, and beyond.
\newblock In {\em 61st {IEEE} Annual Symposium on Foundations of Computer Science, {FOCS} 2020, Durham, NC, USA, November 16-19, 2020}, pages 1158--1167. {IEEE}, 2020.

\bibitem{cygan2015parameterized}
M.~Cygan, F.~V. Fomin, {\L}.~Kowalik, D.~Lokshtanov, D.~Marx, M.~Pilipczuk, M.~Pilipczuk, and S.~Saurabh.
\newblock {\em Parameterized algorithms}, volume~5.
\newblock Springer, 2015.

\bibitem{DBLP:conf/icml/BarbosaENW15}
R.~da~Ponte~Barbosa, A.~Ene, H.~L. Nguyen, and J.~Ward.
\newblock The power of randomization: Distributed submodular maximization on massive datasets.
\newblock In {\em Proceedings of the 32nd International Conference on Machine Learning, {ICML} 2015, Lille, France, 6-11 July 2015}, volume~37 of {\em {JMLR} Workshop and Conference Proceedings}, pages 1236--1244. JMLR.org, 2015.

\bibitem{DBLP:conf/focs/BarbosaENW16}
R.~da~Ponte~Barbosa, A.~Ene, H.~L. Nguyen, and J.~Ward.
\newblock A new framework for distributed submodular maximization.
\newblock In {\em {IEEE} 57th Annual Symposium on Foundations of Computer Science, {FOCS} 2016, 9-11 October 2016, Hyatt Regency, New Brunswick, New Jersey, {USA}}, pages 645--654. {IEEE} Computer Society, 2016.

\bibitem{10.5555/2464827}
R.~G. Downey and M.~R. Fellows.
\newblock {\em Parameterized Complexity}.
\newblock Springer Publishing Company, Incorporated, 2012.

\bibitem{DBLP:conf/iccv/DueckF07}
D.~Dueck and B.~J. Frey.
\newblock Non-metric affinity propagation for unsupervised image categorization.
\newblock In {\em {IEEE} 11th International Conference on Computer Vision, {ICCV} 2007, Rio de Janeiro, Brazil, October 14-20, 2007}, pages 1--8. {IEEE} Computer Society, 2007.

\bibitem{DBLP:conf/icml/DuettingFLNZ22}
P.~Duetting, F.~Fusco, S.~Lattanzi, A.~Norouzi{-}Fard, and M.~Zadimoghaddam.
\newblock Deletion robust submodular maximization over matroids.
\newblock In {\em International Conference on Machine Learning, {ICML} 2022, 17-23 July 2022, Baltimore, Maryland, {USA}}, volume 162 of {\em Proceedings of Machine Learning Research}, pages 5671--5693. {PMLR}, 2022.

\bibitem{DBLP:conf/stoc/DurfeeGGP19}
D.~Durfee, Y.~Gao, G.~Goranci, and R.~Peng.
\newblock Fully dynamic spectral vertex sparsifiers and applications.
\newblock In {\em Proceedings of the 51st Annual {ACM} {SIGACT} Symposium on Theory of Computing, {STOC} 2019, Phoenix, AZ, USA, June 23-26, 2019}, pages 914--925. {ACM}, 2019.

\bibitem{dutting2023fully}
P.~D{\"u}tting, F.~Fusco, S.~Lattanzi, A.~Norouzi-Fard, and M.~Zadimoghaddam.
\newblock Fully dynamic submodular maximization over matroids.
\newblock {\em arXiv preprint arXiv:2305.19918}, 2023.

\bibitem{DBLP:journals/mp/Edmonds71}
J.~Edmonds.
\newblock Matroids and the greedy algorithm.
\newblock {\em Math. Program.}, 1(1):127--136, 1971.

\bibitem{Edmonds1965TransversalsAM}
J.~Edmonds and D.~R. Fulkerson.
\newblock Transversals and matroid partition.
\newblock {\em Journal of Research of the National Bureau of Standards Section B Mathematics and Mathematical Physics}, page 147, 1965.

\bibitem{DBLP:conf/soda/EhsaniHKS18}
S.~Ehsani, M.~Hajiaghayi, T.~Kesselheim, and S.~Singla.
\newblock Prophet secretary for combinatorial auctions and matroids.
\newblock In {\em Proceedings of the Twenty-Ninth Annual {ACM-SIAM} Symposium on Discrete Algorithms, {SODA} 2018, New Orleans, LA, USA, January 7-10, 2018}, pages 700--714. {SIAM}, 2018.

\bibitem{DBLP:conf/kdd/El-AriniG11}
K.~El{-}Arini and C.~Guestrin.
\newblock Beyond keyword search: discovering relevant scientific literature.
\newblock In {\em Proceedings of the 17th {ACM} {SIGKDD} International Conference on Knowledge Discovery and Data Mining, San Diego, CA, USA, August 21-24, 2011}, pages 439--447. {ACM}, 2011.

\bibitem{DBLP:conf/nips/ElenbergDFK17}
E.~R. Elenberg, A.~G. Dimakis, M.~Feldman, and A.~Karbasi.
\newblock Streaming weak submodularity: Interpreting neural networks on the fly.
\newblock In {\em Advances in Neural Information Processing Systems 30: Annual Conference on Neural Information Processing Systems 2017, December 4-9, 2017, Long Beach, CA, {USA}}, pages 4044--4054, 2017.

\bibitem{DBLP:conf/soda/EneN19}
A.~Ene and H.~L. Nguyen.
\newblock Submodular maximization with nearly-optimal approximation and adaptivity in nearly-linear time.
\newblock In {\em Proceedings of the Thirtieth Annual {ACM-SIAM} Symposium on Discrete Algorithms, {SODA} 2019, San Diego, California, USA, January 6-9, 2019}, pages 274--282. {SIAM}, 2019.

\bibitem{DBLP:conf/icml/EneN20}
A.~Ene and H.~L. Nguyen.
\newblock Parallel algorithm for non-monotone dr-submodular maximization.
\newblock In {\em Proceedings of the 37th International Conference on Machine Learning, {ICML} 2020, 13-18 July 2020, Virtual Event}, volume 119 of {\em Proceedings of Machine Learning Research}, pages 2902--2911. {PMLR}, 2020.

\bibitem{DBLP:conf/nips/EneNV17}
A.~Ene, H.~L. Nguyen, and L.~A. V{\'{e}}gh.
\newblock Decomposable submodular function minimization: Discrete and continuous.
\newblock In {\em Advances in Neural Information Processing Systems 30: Annual Conference on Neural Information Processing Systems 2017, December 4-9, 2017, Long Beach, CA, {USA}}, pages 2870--2880, 2017.

\bibitem{DBLP:conf/stoc/EneNV19}
A.~Ene, H.~L. Nguyen, and A.~Vladu.
\newblock Submodular maximization with matroid and packing constraints in parallel.
\newblock In {\em Proceedings of the 51st Annual {ACM} {SIGACT} Symposium on Theory of Computing, {STOC} 2019, Phoenix, AZ, USA, June 23-26, 2019}, pages 90--101. {ACM}, 2019.

\bibitem{DBLP:conf/mfcs/FafianieK14}
S.~Fafianie and S.~Kratsch.
\newblock Streaming kernelization.
\newblock In {\em Mathematical Foundations of Computer Science 2014 - 39th International Symposium, {MFCS} 2014, Budapest, Hungary, August 25-29, 2014. Proceedings, Part {II}}, volume 8635 of {\em Lecture Notes in Computer Science}, pages 275--286. Springer, 2014.

\bibitem{DBLP:journals/jacm/Feige98}
U.~Feige.
\newblock A threshold of ln \emph{n} for approximating set cover.
\newblock {\em J. {ACM}}, 45(4):634--652, 1998.

\bibitem{DBLP:conf/nips/FeldmanK018}
M.~Feldman, A.~Karbasi, and E.~Kazemi.
\newblock Do less, get more: Streaming submodular maximization with subsampling.
\newblock In {\em Advances in Neural Information Processing Systems 31: Annual Conference on Neural Information Processing Systems 2018, NeurIPS 2018, December 3-8, 2018, Montr{\'{e}}al, Canada}, pages 730--740, 2018.

\bibitem{feldman2021streaming}
M.~Feldman, P.~Liu, A.~Norouzi-Fard, O.~Svensson, and R.~Zenklusen.
\newblock Streaming submodular maximization under matroid constraints.
\newblock {\em arXiv preprint arXiv:2107.07183}, 2021.

\bibitem{DBLP:conf/stoc/FeldmanNSZ20}
M.~Feldman, A.~Norouzi{-}Fard, O.~Svensson, and R.~Zenklusen.
\newblock The one-way communication complexity of submodular maximization with applications to streaming and robustness.
\newblock In {\em Proccedings of the 52nd Annual {ACM} {SIGACT} Symposium on Theory of Computing, {STOC} 2020, Chicago, IL, USA, June 22-26, 2020}, pages 1363--1374. {ACM}, 2020.

\bibitem{DBLP:series/txtcs/FlumG06}
J.~Flum and M.~Grohe.
\newblock {\em Parameterized Complexity Theory}.
\newblock Texts in Theoretical Computer Science. An {EATCS} Series. Springer, 2006.

\bibitem{fomin2021fast}
F.~V. Fomin and T.~Korhonen.
\newblock Fast fpt-approximation of branchwidth.
\newblock In {\em Proceedings of 54th Annual ACM Symposium on Theory of Computing (STOC)}, 2022.
\newblock to appear.

\bibitem{DBLP:conf/focs/GaoLP21}
Y.~Gao, Y.~P. Liu, and R.~Peng.
\newblock Fully dynamic electrical flows: Sparse maxflow faster than goldberg-rao.
\newblock In {\em 62nd {IEEE} Annual Symposium on Foundations of Computer Science, {FOCS} 2021, Denver, CO, USA, February 7-10, 2022}, pages 516--527. {IEEE}, 2021.

\bibitem{DBLP:journals/algorithmica/GharanV13}
S.~O. Gharan and J.~Vondr{\'{a}}k.
\newblock On variants of the matroid secretary problem.
\newblock {\em Algorithmica}, 67(4):472--497, 2013.

\bibitem{DBLP:conf/stoc/GuptaK0P17}
A.~Gupta, R.~Krishnaswamy, A.~Kumar, and D.~Panigrahi.
\newblock Online and dynamic algorithms for set cover.
\newblock In {\em Proceedings of the 49th Annual {ACM} {SIGACT} Symposium on Theory of Computing, {STOC} 2017, Montreal, QC, Canada, June 19-23, 2017}, pages 537--550. {ACM}, 2017.

\bibitem{10.5555/3174304.3175483}
A.~Gupta, E.~Lee, and J.~Li.
\newblock An fpt algorithm beating 2-approximation for k-cut.
\newblock In {\em Proceedings of the Twenty-Ninth Annual ACM-SIAM Symposium on Discrete Algorithms}, SODA '18, page 2821–2837, USA, 2018. Society for Industrial and Applied Mathematics.

\bibitem{DBLP:conf/focs/GuptaL20}
A.~Gupta and R.~Levin.
\newblock Fully-dynamic submodular cover with bounded recourse.
\newblock In {\em 61st {IEEE} Annual Symposium on Foundations of Computer Science, {FOCS} 2020, Durham, NC, USA, November 16-19, 2020}, pages 1147--1157. {IEEE}, 2020.

\bibitem{DBLP:conf/wine/GuptaRST10}
A.~Gupta, A.~Roth, G.~Schoenebeck, and K.~Talwar.
\newblock Constrained non-monotone submodular maximization: Offline and secretary algorithms.
\newblock In {\em Internet and Network Economics - 6th International Workshop, {WINE} 2010, Stanford, CA, USA, December 13-17, 2010. Proceedings}, volume 6484 of {\em Lecture Notes in Computer Science}, pages 246--257. Springer, 2010.

\bibitem{DBLP:conf/nips/HanCCW20}
K.~Han, Z.~Cao, S.~Cui, and B.~Wu.
\newblock Deterministic approximation for submodular maximization over a matroid in nearly linear time.
\newblock In {\em Advances in Neural Information Processing Systems 33: Annual Conference on Neural Information Processing Systems 2020, NeurIPS 2020, December 6-12, 2020, virtual}, 2020.

\bibitem{DBLP:conf/nips/HarveyLS20}
N.~J.~A. Harvey, C.~Liaw, and T.~Soma.
\newblock Improved algorithms for online submodular maximization via first-order regret bounds.
\newblock In {\em Advances in Neural Information Processing Systems 33: Annual Conference on Neural Information Processing Systems 2020, NeurIPS 2020, December 6-12, 2020, virtual}, 2020.

\bibitem{DBLP:conf/focs/HenzingerKN13}
M.~Henzinger, S.~Krinninger, and D.~Nanongkai.
\newblock Dynamic approximate all-pairs shortest paths: Breaking the o(mn) barrier and derandomization.
\newblock In {\em 54th Annual {IEEE} Symposium on Foundations of Computer Science, {FOCS} 2013, 26-29 October, 2013, Berkeley, CA, {USA}}, pages 538--547. {IEEE} Computer Society, 2013.

\bibitem{DBLP:journals/siamcomp/HenzingerKN16}
M.~Henzinger, S.~Krinninger, and D.~Nanongkai.
\newblock Dynamic approximate all-pairs shortest paths: Breaking the o(mn) barrier and derandomization.
\newblock {\em {SIAM} J. Comput.}, 45(3):947--1006, 2016.

\bibitem{DBLP:reference/algo/HenzingerKN16}
M.~Henzinger, S.~Krinninger, and D.~Nanongkai.
\newblock Dynamic approximate all-pairs shortest paths: Breaking the o(mn) barrierand derandomization.
\newblock In {\em Encyclopedia of Algorithms}, pages 600--602. 2016.

\bibitem{DBLP:journals/jacm/HenzingerK99}
M.~R. Henzinger and V.~King.
\newblock Randomized fully dynamic graph algorithms with polylogarithmic time per operation.
\newblock {\em J. {ACM}}, 46(4):502--516, 1999.

\bibitem{DBLP:conf/soda/KapronKM13}
B.~M. Kapron, V.~King, and B.~Mountjoy.
\newblock Dynamic graph connectivity in polylogarithmic worst case time.
\newblock In {\em Proceedings of the Twenty-Fourth Annual {ACM-SIAM} Symposium on Discrete Algorithms, {SODA} 2013, New Orleans, Louisiana, USA, January 6-8, 2013}, pages 1131--1142. {SIAM}, 2013.

\bibitem{4557271}
N.~Kashyap.
\newblock Code decomposition: Theory and applications.
\newblock In {\em 2007 IEEE International Symposium on Information Theory}, pages 481--485, 2007.

\bibitem{KT11}
K.-i. Kawarabayashi and M.~Thorup.
\newblock The minimum k-way cut of bounded size is fixed-parameter tractable.
\newblock In {\em Proceedings of the 2011 IEEE 52nd Annual Symposium on Foundations of Computer Science}, FOCS '11, page 160–169, USA, 2011. IEEE Computer Society.

\bibitem{DBLP:conf/icml/0001MZLK19}
E.~Kazemi, M.~Mitrovic, M.~Zadimoghaddam, S.~Lattanzi, and A.~Karbasi.
\newblock Submodular streaming in all its glory: Tight approximation, minimum memory and low adaptive complexity.
\newblock In {\em Proceedings of the 36th International Conference on Machine Learning, {ICML} 2019, 9-15 June 2019, Long Beach, California, {USA}}, volume~97 of {\em Proceedings of Machine Learning Research}, pages 3311--3320. {PMLR}, 2019.

\bibitem{DBLP:conf/icml/0001ZK18}
E.~Kazemi, M.~Zadimoghaddam, and A.~Karbasi.
\newblock Scalable deletion-robust submodular maximization: Data summarization with privacy and fairness constraints.
\newblock In {\em Proceedings of the 35th International Conference on Machine Learning, {ICML} 2018, Stockholmsm{\"{a}}ssan, Stockholm, Sweden, July 10-15, 2018}, volume~80 of {\em Proceedings of Machine Learning Research}, pages 2549--2558. {PMLR}, 2018.

\bibitem{DBLP:conf/kdd/KempeKT03}
D.~Kempe, J.~M. Kleinberg, and {\'{E}}.~Tardos.
\newblock Maximizing the spread of influence through a social network.
\newblock In {\em Proceedings of the Ninth {ACM} {SIGKDD} International Conference on Knowledge Discovery and Data Mining, Washington, DC, USA, August 24 - 27, 2003}, pages 137--146. {ACM}, 2003.

\bibitem{DBLP:journals/toc/KempeKT15}
D.~Kempe, J.~M. Kleinberg, and {\'{E}}.~Tardos.
\newblock Maximizing the spread of influence through a social network.
\newblock {\em Theory of Computing}, 11:105--147, 2015.

\bibitem{DBLP:conf/stoc/KleinbergW12}
R.~Kleinberg and S.~M. Weinberg.
\newblock Matroid prophet inequalities.
\newblock In {\em Proceedings of the 44th Symposium on Theory of Computing Conference, {STOC} 2012, New York, NY, USA, May 19 - 22, 2012}, pages 123--136. {ACM}, 2012.

\bibitem{DBLP:journals/geb/KleinbergW19}
R.~Kleinberg and S.~M. Weinberg.
\newblock Matroid prophet inequalities and applications to multi-dimensional mechanism design.
\newblock {\em Games Econ. Behav.}, 113:97--115, 2019.

\bibitem{DBLP:conf/focs/Korhonen21}
T.~Korhonen.
\newblock A single-exponential time 2-approximation algorithm for treewidth.
\newblock In {\em 62nd {IEEE} Annual Symposium on Foundations of Computer Science, {FOCS} 2021, Denver, CO, USA, February 7-10, 2022}, pages 184--192. {IEEE}, 2021.

\bibitem{DBLP:conf/bmvc/Krause13}
A.~Krause.
\newblock Submodularity in machine learning and vision.
\newblock In {\em British Machine Vision Conference, {BMVC} 2013, Bristol, UK, September 9-13, 2013}. {BMVA} Press, 2013.

\bibitem{DBLP:journals/topc/KumarMVV15}
R.~Kumar, B.~Moseley, S.~Vassilvitskii, and A.~Vattani.
\newblock Fast greedy algorithms in mapreduce and streaming.
\newblock {\em {ACM} Trans. Parallel Comput.}, 2(3):14:1--14:22, 2015.

\bibitem{DBLP:conf/aaai/KumariB21}
L.~Kumari and J.~A. Bilmes.
\newblock Submodular span, with applications to conditional data summarization.
\newblock In {\em Thirty-Fifth {AAAI} Conference on Artificial Intelligence, {AAAI} 2021, Thirty-Third Conference on Innovative Applications of Artificial Intelligence, {IAAI} 2021, The Eleventh Symposium on Educational Advances in Artificial Intelligence, {EAAI} 2021, Virtual Event, February 2-9, 2021}, pages 12344--12352. {AAAI} Press, 2021.

\bibitem{DBLP:conf/nips/KupferQBS20}
R.~Kupfer, S.~Qian, E.~Balkanski, and Y.~Singer.
\newblock The adaptive complexity of maximizing a gross substitutes valuation.
\newblock In {\em Advances in Neural Information Processing Systems 33: Annual Conference on Neural Information Processing Systems 2020, NeurIPS 2020, December 6-12, 2020, virtual}, 2020.

\bibitem{DBLP:conf/uai/KvetonWAEE14}
B.~Kveton, Z.~Wen, A.~Ashkan, H.~Eydgahi, and B.~Eriksson.
\newblock Matroid bandits: Fast combinatorial optimization with learning.
\newblock In {\em Proceedings of the Thirtieth Conference on Uncertainty in Artificial Intelligence, {UAI} 2014, Quebec City, Quebec, Canada, July 23-27, 2014}, pages 420--429. {AUAI} Press, 2014.

\bibitem{DBLP:conf/nips/LattanziMNTZ20}
S.~Lattanzi, S.~Mitrovic, A.~Norouzi{-}Fard, J.~Tarnawski, and M.~Zadimoghaddam.
\newblock Fully dynamic algorithm for constrained submodular optimization.
\newblock In {\em Advances in Neural Information Processing Systems 33: Annual Conference on Neural Information Processing Systems 2020, NeurIPS 2020, December 6-12, 2020, virtual}, 2020.

\bibitem{LattanziMNTZ20update}
S.~Lattanzi, S.~Mitrovic, A.~Norouzi{-}Fard, J.~Tarnawski, and M.~Zadimoghaddam.
\newblock Fully dynamic algorithm for constrained submodular optimization.
\newblock {\em CoRR}, abs/2006.04704v2, 2023.

\bibitem{DBLP:conf/stoc/LeeSV10}
J.~Lee, M.~Sviridenko, and J.~Vondr{\'{a}}k.
\newblock Matroid matching: the power of local search.
\newblock In {\em Proceedings of the 42nd {ACM} Symposium on Theory of Computing, {STOC} 2010, Cambridge, Massachusetts, USA, 5-8 June 2010}, pages 369--378. {ACM}, 2010.

\bibitem{DBLP:journals/mor/LeeSV10}
J.~Lee, M.~Sviridenko, and J.~Vondr{\'{a}}k.
\newblock Submodular maximization over multiple matroids via generalized exchange properties.
\newblock {\em Math. Oper. Res.}, 35(4):795--806, 2010.

\bibitem{DBLP:conf/bcb/LibbrechtBN18}
M.~W. Libbrecht, J.~A. Bilmes, and W.~S. Noble.
\newblock Choosing non-redundant representative subsets of protein sequence data sets using submodular optimization.
\newblock In {\em Proceedings of the 2018 {ACM} International Conference on Bioinformatics, Computational Biology, and Health Informatics, {BCB} 2018, Washington, DC, USA, August 29 - September 01, 2018}, page 566. {ACM}, 2018.

\bibitem{DBLP:conf/acl/LinB11}
H.~Lin and J.~A. Bilmes.
\newblock A class of submodular functions for document summarization.
\newblock In {\em The 49th Annual Meeting of the Association for Computational Linguistics: Human Language Technologies, Proceedings of the Conference, 19-24 June, 2011, Portland, Oregon, {USA}}, pages 510--520. The Association for Computer Linguistics, 2011.

\bibitem{DBLP:conf/soda/LiuV19}
P.~Liu and J.~Vondr{\'{a}}k.
\newblock Submodular optimization in the mapreduce model.
\newblock In {\em 2nd Symposium on Simplicity in Algorithms, SOSA@SODA 2019, January 8-9, 2019 - San Diego, CA, {USA}}, volume~69 of {\em {OASICS}}, pages 18:1--18:10. Schloss Dagstuhl - Leibniz-Zentrum f{\"{u}}r Informatik, 2019.

\bibitem{10.2307/2371070}
S.~MacLane.
\newblock Some interpretations of abstract linear dependence in terms of projective geometry.
\newblock {\em American Journal of Mathematics}, 58(1):236--240, 1936.

\bibitem{Maeda1970}
F.~Maeda and S.~Maeda.
\newblock {\em Matroid Lattices}, pages 56--71.
\newblock Springer Berlin Heidelberg, Berlin, Heidelberg, 1970.

\bibitem{DBLP:journals/jcp/MagosMP06}
D.~Magos, I.~Mourtos, and L.~S. Pitsoulis.
\newblock The matching predicate and a filtering scheme based on matroids.
\newblock {\em J. Comput.}, 1(6):37--42, 2006.

\bibitem{marx2008parameterized}
D.~Marx.
\newblock Parameterized complexity and approximation algorithms.
\newblock {\em The Computer Journal}, 51(1):60--78, 2008.

\bibitem{DBLP:journals/mst/McGregorV19}
A.~McGregor and H.~T. Vu.
\newblock Better streaming algorithms for the maximum coverage problem.
\newblock {\em Theory Comput. Syst.}, 63(7):1595--1619, 2019.

\bibitem{10.2307/24901346}
G.~J. MINTY.
\newblock On the axiomatic foundations of the theories of directed linear graphs, electrical networks and network-programming.
\newblock {\em Journal of Mathematics and Mechanics}, 15(3):485--520, 1966.

\bibitem{DBLP:conf/stoc/MirrokniZ15}
V.~S. Mirrokni and M.~Zadimoghaddam.
\newblock Randomized composable core-sets for distributed submodular maximization.
\newblock In {\em Proceedings of the Forty-Seventh Annual {ACM} on Symposium on Theory of Computing, {STOC} 2015, Portland, OR, USA, June 14-17, 2015}, pages 153--162. {ACM}, 2015.

\bibitem{DBLP:conf/aaai/MirzasoleimanJ018}
B.~Mirzasoleiman, S.~Jegelka, and A.~Krause.
\newblock Streaming non-monotone submodular maximization: Personalized video summarization on the fly.
\newblock In {\em Proceedings of the Thirty-Second {AAAI} Conference on Artificial Intelligence, (AAAI-18), the 30th innovative Applications of Artificial Intelligence (IAAI-18), and the 8th {AAAI} Symposium on Educational Advances in Artificial Intelligence (EAAI-18), New Orleans, Louisiana, USA, February 2-7, 2018}, pages 1379--1386. {AAAI} Press, 2018.

\bibitem{DBLP:conf/icml/MirzasoleimanK017}
B.~Mirzasoleiman, A.~Karbasi, and A.~Krause.
\newblock Deletion-robust submodular maximization: Data summarization with "the right to be forgotten".
\newblock In {\em Proceedings of the 34th International Conference on Machine Learning, {ICML} 2017, Sydney, NSW, Australia, 6-11 August 2017}, volume~70 of {\em Proceedings of Machine Learning Research}, pages 2449--2458. {PMLR}, 2017.

\bibitem{DBLP:conf/nips/Monemizadeh20}
M.~Monemizadeh.
\newblock Dynamic submodular maximization.
\newblock In {\em Advances in Neural Information Processing Systems 33: Annual Conference on Neural Information Processing Systems 2020, NeurIPS 2020, December 6-12, 2020, virtual}, 2020.

\bibitem{DBLP:conf/stoc/NanongkaiS17}
D.~Nanongkai and T.~Saranurak.
\newblock Dynamic spanning forest with worst-case update time: adaptive, las vegas, and o(n\({}^{\mbox{1/2 - {\(\epsilon\)}}}\))-time.
\newblock In {\em Proceedings of the 49th Annual {ACM} {SIGACT} Symposium on Theory of Computing, {STOC} 2017, Montreal, QC, Canada, June 19-23, 2017}, pages 1122--1129. {ACM}, 2017.

\bibitem{DBLP:conf/focs/NanongkaiSW17}
D.~Nanongkai, T.~Saranurak, and C.~Wulff{-}Nilsen.
\newblock Dynamic minimum spanning forest with subpolynomial worst-case update time.
\newblock In {\em 58th {IEEE} Annual Symposium on Foundations of Computer Science, {FOCS} 2017, Berkeley, CA, USA, October 15-17, 2017}, pages 950--961. {IEEE} Computer Society, 2017.

\bibitem{DBLP:journals/talg/NeimanS16}
O.~Neiman and S.~Solomon.
\newblock Simple deterministic algorithms for fully dynamic maximal matching.
\newblock {\em {ACM} Trans. Algorithms}, 12(1):7:1--7:15, 2016.

\bibitem{DBLP:journals/mp/NemhauserWF78}
G.~L. Nemhauser, L.~A. Wolsey, and M.~L. Fisher.
\newblock An analysis of approximations for maximizing submodular set functions - {I}.
\newblock {\em Math. Program.}, 14(1):265--294, 1978.

\bibitem{DBLP:conf/stoc/OnakR10}
K.~Onak and R.~Rubinfeld.
\newblock Maintaining a large matching and a small vertex cover.
\newblock In {\em Proceedings of the 42nd {ACM} Symposium on Theory of Computing, {STOC} 2010, Cambridge, Massachusetts, USA, 5-8 June 2010}, pages 457--464. {ACM}, 2010.

\bibitem{DBLP:books/daglib/0070636}
J.~G. Oxley.
\newblock {\em Matroid theory}.
\newblock Oxford University Press, 1992.

\bibitem{DBLP:books/ph/PapadimitriouS82}
C.~H. Papadimitriou and K.~Steiglitz.
\newblock {\em Combinatorial Optimization: Algorithms and Complexity}.
\newblock Prentice-Hall, 1982.

\bibitem{DBLP:conf/nips/Peng21}
B.~Peng.
\newblock Dynamic influence maximization.
\newblock In {\em Advances in Neural Information Processing Systems 34: Annual Conference on Neural Information Processing Systems 2021, NeurIPS 2021, December 6-14, 2021, virtual}, pages 10718--10731, 2021.

\bibitem{DBLP:conf/aaai/RadanovicS0F18}
G.~Radanovic, A.~Singla, A.~Krause, and B.~Faltings.
\newblock Information gathering with peers: Submodular optimization with peer-prediction constraints.
\newblock In {\em Proceedings of the Thirty-Second {AAAI} Conference on Artificial Intelligence, (AAAI-18), the 30th innovative Applications of Artificial Intelligence (IAAI-18), and the 8th {AAAI} Symposium on Educational Advances in Artificial Intelligence (EAAI-18), New Orleans, Louisiana, USA, February 2-7, 2018}, pages 1603--1610. {AAAI} Press, 2018.

\bibitem{DBLP:conf/focs/Rauch92}
M.~Rauch.
\newblock Fully dynamic biconnectivity in graphs.
\newblock In {\em 33rd Annual Symposium on Foundations of Computer Science, Pittsburgh, Pennsylvania, USA, 24-27 October 1992}, pages 50--59. {IEEE} Computer Society, 1992.

\bibitem{DBLP:conf/stoc/SawlaniW20}
S.~Sawlani and J.~Wang.
\newblock Near-optimal fully dynamic densest subgraph.
\newblock In {\em Proccedings of the 52nd Annual {ACM} {SIGACT} Symposium on Theory of Computing, {STOC} 2020, Chicago, IL, USA, June 22-26, 2020}, pages 181--193. {ACM}, 2020.

\bibitem{DBLP:journals/jmlr/SchreiberBN20}
J.~M. Schreiber, J.~A. Bilmes, and W.~S. Noble.
\newblock apricot: Submodular selection for data summarization in python.
\newblock {\em J. Mach. Learn. Res.}, 21:161:1--161:6, 2020.

\bibitem{DBLP:conf/cikm/SiposSSJ12}
R.~Sipos, A.~Swaminathan, P.~Shivaswamy, and T.~Joachims.
\newblock Temporal corpus summarization using submodular word coverage.
\newblock In {\em 21st {ACM} International Conference on Information and Knowledge Management, CIKM'12, Maui, HI, USA, October 29 - November 02, 2012}, pages 754--763. {ACM}, 2012.

\bibitem{DBLP:conf/focs/Solomon16}
S.~Solomon.
\newblock Fully dynamic maximal matching in constant update time.
\newblock In {\em {IEEE} 57th Annual Symposium on Foundations of Computer Science, {FOCS} 2016, 9-11 October 2016, Hyatt Regency, New Brunswick, New Jersey, {USA}}, pages 325--334. {IEEE} Computer Society, 2016.

\bibitem{DBLP:journals/talg/SolomonW20}
S.~Solomon and N.~Wein.
\newblock Improved dynamic graph coloring.
\newblock {\em {ACM} Trans. Algorithms}, 16(3):41:1--41:24, 2020.

\bibitem{DBLP:conf/icml/StanZ0K17}
S.~Stan, M.~Zadimoghaddam, A.~Krause, and A.~Karbasi.
\newblock Probabilistic submodular maximization in sub-linear time.
\newblock In {\em Proceedings of the 34th International Conference on Machine Learning, {ICML} 2017, Sydney, NSW, Australia, 6-11 August 2017}, volume~70 of {\em Proceedings of Machine Learning Research}, pages 3241--3250. {PMLR}, 2017.

\bibitem{DBLP:journals/spm/TohidiACGLK20}
E.~Tohidi, R.~Amiri, M.~Coutino, D.~Gesbert, G.~Leus, and A.~Karbasi.
\newblock Submodularity in action: From machine learning to signal processing applications.
\newblock {\em {IEEE} Signal Process. Mag.}, 37(5):120--133, 2020.

\bibitem{DBLP:conf/stoc/BrandGJLLPS22}
J.~van~den Brand, Y.~Gao, A.~Jambulapati, Y.~T. Lee, Y.~P. Liu, R.~Peng, and A.~Sidford.
\newblock Faster maxflow via improved dynamic spectral vertex sparsifiers.
\newblock In {\em {STOC} '22: 54th Annual {ACM} {SIGACT} Symposium on Theory of Computing, Rome, Italy, June 20 - 24, 2022}, pages 543--556. {ACM}, 2022.

\bibitem{DBLP:conf/focs/BrandN19}
J.~van~den Brand and D.~Nanongkai.
\newblock Dynamic approximate shortest paths and beyond: Subquadratic and worst-case update time.
\newblock In {\em 60th {IEEE} Annual Symposium on Foundations of Computer Science, {FOCS} 2019, Baltimore, Maryland, USA, November 9-12, 2019}, pages 436--455. {IEEE} Computer Society, 2019.

\bibitem{DBLP:conf/icml/WeiIB15}
K.~Wei, R.~K. Iyer, and J.~A. Bilmes.
\newblock Submodularity in data subset selection and active learning.
\newblock In {\em Proceedings of the 32nd International Conference on Machine Learning, {ICML} 2015, Lille, France, 6-11 July 2015}, volume~37 of {\em {JMLR} Workshop and Conference Proceedings}, pages 1954--1963. JMLR.org, 2015.

\bibitem{DBLP:conf/bcb/YangBN20}
W.~Yang, J.~A. Bilmes, and W.~S. Noble.
\newblock Submodular sketches of single-cell rna-seq measurements.
\newblock In {\em {BCB} '20: 11th {ACM} International Conference on Bioinformatics, Computational Biology and Health Informatics, Virtual Event, USA, September 21-24, 2020}, pages 61:1--61:6. {ACM}, 2020.

\bibitem{zhang2022coresets}
G.~Zhang, N.~Tatti, and A.~Gionis.
\newblock Coresets remembered and items forgotten: submodular maximization with deletions.
\newblock In {\em 2022 IEEE International Conference on Data Mining (ICDM)}, pages 676--685. IEEE, 2022.

\end{thebibliography}

\appendix 
\section{Some of the Proofs regarding Invariants of the algorithm for submodular matroid maximization}
\label{sec:invar_appendix}

\subsection{Proof of Theorem \ref{thm:invariants:leveling}}
\label{subs:thm:invariants:levelin}
\begin{lemma}[Survivor invariant]
\label{lm:survivor:leveling}
If before calling $\MatroidConstLevel{}(j)$, the level invariants partially hold for the first $j$ levels,
then after its execution, the survivor invariant fully holds.
\end{lemma}

\begin{proof} 

First of all, we assume that $R_j \neq \emptyset$, otherwise $T = j-1$ and we are done. 
As we have in Algorithm \MatroidConstLevel{}, let $P$ be a random permutation of the set $R_j$. 
Let us fix an arbitrary element $e \in P$ and suppose that 
at the time when we see $e \in P$, the current level is $L_{\ell}$ for $\ell \ge j$. 
We have two cases. 
Either $e$ is a promoting element for the level $L_{\ell-1}$ or it is not promoting for the level $L_{\ell-1}$. 

First, assume that $e$ is a promoting element for the level $L_{\ell-1}$. 
We then let $e_{\ell}$ be $e$, perform a set of computations, and then start the new level. 
In particular, the element $e$ is not added to $R_{\ell+1}$ and so, it will not appear in any set $R_{z > \ell}$.
Recall that Lemma~\ref{lm:binary_search_argument} proves if $e$ 
is not a promoting element with respect to a level $L_{z}$, 
it will not be a promoting element for the next level $L_{x}$ where $z\le x \le T$.
On the other hand, since $e$ is a promoting element for the level $L_{\ell-1}$, 
we add $e$ to all previous sets $R_{j+1},\cdots, R_{\ell}$. 

Next, we consider the latter case where $e$ is not a promoting element for the level $L_{\ell-1}$. 
That is, $\suit(e,L_{\ell-1})$ is $False$.
This essentially means that if we inductively apply the argument of Lemma~\ref{lm:binary_search_argument}, 
there exists an integer $z \in [j,\ell)$ for which 
$\boolsuit(e,L_{z-1})$ is $True$, but $\boolsuit(e,L_{z})$ is $False$. 
This means $e$ is a promoting element for all levels $L_j,\cdots,L_{z-1}$ and 
it is not promoting for levels $L_{z},\cdots,L_T$. 
According to function $\MatroidConstLevel$, 
we insert the element $e$ into sets $R_{j+1}, \cdots, R_{z}$. 
Hence, after the execution of $\MatroidConstLevel{}(j)$, 
the survivor invariant holds. 
\end{proof}

\begin{lemma}[Independent invariant]
\label{lm:independent:leveling}
If before calling $\MatroidConstLevel{}(j)$, the level invariants partially hold for the first $j$ levels,
then after its execution, the independent invariant fully holds. 
\end{lemma}

\begin{proof} 

In the execution of $\MatroidConstLevel{}(j)$, the variable $\ell$ is set to $j, j+1,\cdots, T,T+1$. Therefore, for each $\ell \in  [j,T]$, we set $I_{\ell}$ to  $(I_{\ell-1} + e_\ell) \backslash y$ in Line \ref{line:constlevelmatroid:increase_ell}, where $y$ is defined as \replacementTester$(I_{\ell-1}, I'_{\ell-1}, e, w[I_{\ell - 1}])$ in Line \ref{line:const_level_matroid:def_y}.
Adding this to the assumption of lemma implies that $I_{\ell} = (I_{\ell-1} + e_\ell) - \replacementTester(I_{\ell-1}, I'_{\ell-1}, e, w[I_{\ell - 1}])$ holds for any $\ell \in [T]$.

Next, we prove $I'_\ell = \cup_{m \le \ell} I_m$ for any $\ell\in[T]$ using induction.  
By the assumption of lemma, $I'_{\ell} = \cup_{m \le {\ell}} I_{m}$ holds for any  $i\leq j-1$.
For the induction step, assume $\ell  \in [j,T]$ and 
$I'_{\ell-1} = \cup_{m \le {\ell-1}} I_{m}$ holds. 
In Line \ref{line:constlevelmatroid:increase_ell} we set $I_{\ell}=(I_{\ell-1} + e) \backslash y$.
Since $e\notin I+{\ell-1}$, it means $I_{\ell} \setminus I_{\ell-1}=e$.
We then set $I'_{\ell}=I'_{\ell-1} + e$ in Line \ref{line:constlevelmatroid:increase_ell}. Putting everything together we have
$$I'_{\ell} = I'_{\ell-1} + e=\cup_{m \le {\ell-1}} I_{m}+e = \cup_{m \le {\ell-1}} I_{m} + (I_{\ell} \setminus I_{\ell-1}) = \cup_{m \le {\ell}} I_{m} \enspace .$$
It completes the proof the of lemma.
\end{proof}

\begin{lemma}[Weight invariant]
\label{lm:weight:leveling}
If before calling $\MatroidConstLevel{}(j)$, the level invariants partially hold for the first $j$ levels,
then after its execution, the weight invariant fully holds. 
\end{lemma}

\begin{proof} 
To prove the lemma, we need to show $e_\ell \in R_\ell$ and  $w(e_\ell) = f(I'_{\ell-1} + e_\ell) - f(I'_{\ell-1})$ hold for each $\ell \in [j,T]$.
Recall that in the execution of $\MatroidConstLevel{}(j)$, after constructing the level $L_\ell$, we increase the variable $\ell$ by one. Hence, the variable $\ell$ is set to $j, j+1,\cdots, T,T+1$ during the execution of $\MatroidConstLevel{}(j)$.

We first prove $e_\ell \in R_\ell$ holds for each $\ell \in[j,T]$. Let $\ell $ be a fixed integer in  $[j,T]$.  Since $e_\ell$ is an element of $P$, and $P$ is a random permutation of $R_j$, we have $e_\ell \in R_j$. 
We know that $\replacementTester(I_{\ell-1}, I'_{\ell-1}, e_\ell, w[I_{\ell-1}]) \ne \err$. Then the monotone property that we prove in Lemma~\ref{lm:binary_search_argument} implies that $\replacementTester(I_{m-1}, I'_{m-1}, e_{\ell}, w[I_{m-1}]) \ne \err$ for any $m \in [i,\ell]$. Also, it is obvious that $e_\ell \neq e_{m}$ when $m<\ell$. Recall that $R_{m}= \{ e \in R_{m-1} - e_{m-1}: \replacementTester{}(I_{m-1}, I'_{m-1}, e, w[I_{m-1}]) \ne \err\}$. Therefore, by a simple induction on  $m$, we  can  show  that $e_\ell \in R_{m}$ holds for each  $m\in[i,\ell]$, which implies $e_\ell \in R_\ell$.

Moreover, we fix the weight $w(e_\ell)=f(I'_{\ell - 1} + e_\ell) - f(I'_{\ell -1})$ for each $\ell\in[j,T]$ in Line \ref{line:const_level_matroid:w}.
Adding this to the assumption of Lemma finishes the proof.
\end{proof}
\begin{lemma}[Terminator invariant]
\label{lm:terminator:leveling}
If before calling $\MatroidConstLevel{}(j)$, the level invariants partially hold for the first $j$ levels,
then after its execution, the terminator invariant fully holds. 
\end{lemma}

\begin{proof}
According to Line~\ref{line:constlevelmatroid:addR_bs} and the variable $z$, if we add an element $e$ to $R_r$ at some point of time, then $r\leq z\leq \ell-1$ holds at that moment.
Since the variable $\ell$ never decreases during the execution of $\MatroidConstLevel{}(j)$ and we return $\ell-1$ as $T$ at the end, we can conduct that no element has been added to  $R_{T+1}$, and then $R_{T+1}=\emptyset$, which means the terminator invariant holds.   
\end{proof}

\subsection{Proof of Lemma \ref{mat_insert_level}}
\label{subs:mat_insert_level}

To prove the lemma, we first mention some useful facts and then show that the starter, weight, independent, and survivor invariants partially hold. Finally, we prove that all level invariants hold.

We begin with defining variables $i^*$ and $j^*$ as follows.

\begin{itemize}
    \item \textbf{$i^*$: } If during the execution of \insertv$(v)$ there is  $i\in [T]$ such that $e_{i}$ has been set to be $v$, which also implies that we have invoked \MatroidConstLevel $(i + 1)$, then we set $i^*$ to be $i$. Otherwise, we set $i^*$ to be $T + 1$. 
    \item \textbf{$j^*$: } Let $j^*$ be the largest $i\in [0, T^-+1]$ such that we have added $v$ to $R_i^-$.
\end{itemize}

We consider these two cases in this proof. 
\begin{itemize}
    \item Case 1: $i^* \leq T$, which means $e_{i^*} = v$ and therefore $j^* = i^*$. It also means that we have invoked \MatroidConstLevel $(i^* + 1)$.
    \item Case 2: $i^* = T + 1$, which means \MatroidConstLevel{} has never been invoked  during the insertion of $v$.
    Note that in this case, $T = T^-$ and it is not possible for $j^*$ to be equal to $T^- + 1$, since that would have caused invoking \MatroidConstLevel.  Thus, $j^* < T^-+1 = T+1 =i^*$.
\end{itemize}

Considering our algorithm in \insertv$(v)$, 
it is clear that for any $i < i^*$, we have not made any kind of change in $e_i^-$, $w(e_i)$, $I^-_i$, or $I'^-_i$ at least until \MatroidConstLevel{} is invoked, if it ever gets invoked. 
Additionally, according to \MatroidConstLevel, we know that if we have invoked \MatroidConstLevel$(i^* + 1)$, there has not been any alteration to the variables regarding previous levels. 
Hence, we can conduct the following facts.

\begin{fact}
\label{fact:iprev:e}
For any $i \in [1, i^*)$, we have $e_i = e_i^-$.
\end{fact}

\begin{fact}
\label{fact:iprev:w}
For any $i \in [1, i^*)$, we have $w(e_i) = w^-(e_i)$.
\end{fact}

\begin{fact}
\label{fact:iprev:I}
For any $i \in [0, i^*)$, we have $I_i = I_i^-$.
\end{fact}

\begin{fact}
\label{fact:iprev:I'}
For any $i \in [0, i^*)$, we have ${I'}_i =  {I'}_i^-$. 
\end{fact}

By the definition of $j^*$, we have added the element $v$ to the set $R_i^-$, for each $i \in [0, j^*]$.
Recall that $j^* \leq i^*$, and by invoking \MatroidConstLevel$(i^* + 1)$, there has not been any alteration to the variables regarding previous levels. It leads to the following fact.
\begin{fact}
\label{fact:iprev:R+v}
For any $i \in [0, j^*]$, we have $R_i = R_i^- + v$.
\end{fact}

We know that if Case 2 holds, which means \MatroidConstLevel{} has never been invoked during the insertion of $v$, we have $R_i = R_i^-$  for any $i \in [j^* + 1, T + 1]$. Recall that in Case 2, $i^* = T + 1$, and therefore $[j^*+1, T+1] = [j^*+1, i^*]$. Also if Case 1 holds, $j^* = i^*$, so $[j^* + 1, i^*] = \emptyset$. Thus, independent of the case, we can have the following fact.

\begin{fact}
\label{fact:iprev:R}
 For any $i \in [j^*+1, i^*]$, we have $R_i = R_i^-$.
\end{fact}

In the following, we first prove that the starter invariant holds after executing $\insertv(v)$. We next show that the weight, independent, and survivor invariants partially hold for the first $i^*+1$ levels. Finally, we complete the proof by proving that all the level invariants hold.

\textbf{Starter invariant.}
To show that the starter invariant holds after $\insertv(v)$, we need to prove $R_0= V$ and $I_0={I'}_0=\emptyset$.
By the assumption of this lemma, we have $R_0^-=V^-$, and Fact~\ref{fact:iprev:R+v} results that $R_0 = R_0^- + v$. Thus $R_0 = R_0^- + v = V^- + v = V$.

Again by the assumption of this lemma, $I_0^-={I'}_0^-=\emptyset$.
Due to Fact~\ref{fact:iprev:I}, we have $I_0=I^-_0$, and therefore, it is clear that $I_0 = I_0^- = \emptyset$.
Similarly, we have ${I'}_0={I'}^-_0$ because of Fact~\ref{fact:iprev:I'}, and then ${I'}_0 = {I'}_0^- = \emptyset$.

\textbf{Weight invariant (partially).}
To show that weight invariant partially holds for the first $i^*$ levels, we first prove $e_i\in R_i$ and we prove then $w(e_i) = f(I'_{i - 1} + e_i) - f(I'_{i - 1})$.

By the assumption of this lemma, we know that $e_i^- \in R_i^-$ for $i \in [1, i^*)$.
Besides, according to Fact~\ref{fact:iprev:R} and Fact~\ref{fact:iprev:R+v}, for any $i \in [0,i^*]$, either $R_i=R_i^-+v$ or $R_i=R_i^-$, and thus $R_i^- \subseteq R_i$. Also for any $i \in [1, i^*)$ we have $e_i=e_i^-$ by Fact~\ref{fact:iprev:e}. Putting everything  together, we have $e_i=e_i^- \in R_i^- \subseteq R_i$ for any $i\in[1,i^*)$. which means $e_i \in R_i$. 

Next we need to show 
$w(e_i)= f(I'_{i - 1} + e_i) - f(I'_{i - 1})$.
For any $i\in[1, i^*)$, we have $w(e_i) = w^-(e_i)$ by Fact~\ref{fact:iprev:w}. Besides, $e_i=e_i^-$ by Fact~\ref{fact:iprev:e}, and then $w^-(e_i)=w^-(e_i^-)$.
Moreover, the assumption of this lemma implies that $w^-(e_i^-)=f(I'^-_{i - 1} + e_i^-) - f(I'^-_{i - 1})$.
Also $I'^-_{i - 1}=I'_{i - 1}$ for any $i\in [1,i^*)$ because of Fact~\ref{fact:iprev:I'}. Adding it to $e_i=e_i^-$ implies $f(I'^-_{i - 1} + e_i^-) - f(I'^-_{i - 1}) = f(I'_{i - 1} + e_i) - f(I'_{i - 1})$.
Putting everything together, for any $i\in [1,i^*)$ we have

\[
w(e_i) = w^-(e_i) = w^-(e_i^-)= f(I'^-_{i - 1} + e_i^-) - f(I'^-_{i - 1}) = f(I'_{i - 1} + e_i) - f(I'_{i - 1})  \enspace .
\]

\textbf{Independent invariant (partially).}
Now we show that the independent invariant partially holds for the first $i^*$ levels after $\insertv(v)$.
To do this, we prove $I_i = I_{i - 1} + e_i - \replacementTester{}(I_{i - 1},  {I'}_{i - 1}, e_i, w[I_{i - 1}])$ and  ${I'}_i = \cup_{j \le i} I_j$ holds for all $i\in[1,i^*)$.

Using Fact~\ref{fact:iprev:I}, we have $I_i = I_i^-$ for any $i\in  [1,i^*)$. Also, $I_i^- = I_{i - 1}^- + e_i^- - \replacementTester{}(I_{i - 1}^-, {I'}_{i - 1}^-, e_i^-, w^-[I_{i - 1}^-])$ by
the assumption of this lemma.  Then, 
$$ I_i = I_i^- = I_{i - 1}^- + e_i^- - \replacementTester{}(I_{i - 1}^-, {I'}_{i - 1}^-, e_i^-, w^-[I_{i - 1}^-]) \enspace.
$$

Recall that $w[X]$ means $w$ restricted to domain $X$.
Now, we show that $w^-[I_{i - 1}^-]=w[I_{i - 1}]$ holds for any $i\in[1,i^*)$.
By Fact~\ref{fact:iprev:I}, $I_{i-1} = I_{i-1}^-$, which means the domain of $w^-[I_{i - 1}^-]$ and $w[I_{i - 1}]$ are equal.
Moreover, since the independent invariant holds before the insertion of $v$ by the assumption of this lemma, we have $I_{i - 1}^- \subseteq \{e_1^-, \cdots, e_{i - 1}^-\}$.
Besides, for any $j\leq i-1 <i^*$, Fact~\ref{fact:iprev:e} implies that $e_j^-=e_j$, which results in $I_{i - 1}^- \subseteq \{e_1, \cdots, e_{i - 1}\}$.
We also have $w(e_j)=w^-(e_j)$ for any $j\leq i-1 <i^*$ using Fact~\ref{fact:iprev:w}. It completes the proof of $w^-[I_{i - 1}^-]=w[I_{i - 1}]$.
Hence,

$$I_i = I_{i - 1}^- + e_i^- - \replacementTester{}(I_{i - 1}^-, {I'}_{i - 1}^-, e_i^-, w[I_{i - 1}]) \enspace.$$
For any $i\in[1,i^*)$, we have $I_{i - 1}=I_{i - 1}^-$ by Fact~\ref{fact:iprev:I}, ${I'}_{i - 1}={I'}_{i - 1}^-1$ by Fact~\ref{fact:iprev:I'}, and $e_i=e_i^-$ by Fact~\ref{fact:iprev:e}. Then we can further conclude that:
$$
I_i = I_{i - 1} + e_i - \replacementTester{}(I_{i - 1},  {I'}_{i - 1}, e_i, w[I_{i - 1}]) \enspace.
$$

Finally, we prove ${I'}_i = \cup_{j \le i} I_j$ for  all $i \in [1,i^*)$.
Because of Fact~\ref{fact:iprev:I'}, ${I'}_i =  {I'}_i^-$ holds  for any $i\in[1,i^*)$. Besides, ${I'}_i^- = \cup_{j \le i} I_j^-$ a result of the assumption of this lemma.
Also, $I_j^-=I_j$ for any $j\leq i <i^*$ due to Fact~\ref{fact:iprev:I}.
Putting everything together we have
$$
{I'}_i =  {I'}_i^- = \cup_{j \le i} I_j^- = \cup_{j \le i} I_j  \enspace .
$$

\textbf{Survivor invariant (partially). }
Next, we show that the survivor invariant partially holds for the first $i^*$ levels by proving that $R_i=\{ e \in R_{i-1} - e_{i-1}: \replacementTester{}(I_{i-1}, {I'}_{i-1}, e, w[I_{i-1}]) \ne \err\}$ holds for any $i\in[1,i^*]$.
In the following, we first consider $i \in [1, j^*]$ and then $i \in [j^*+1, i^*]$.
Recall that $j^*$ is the largest $j$ that we added $v$ to $R_j$ in $\insertv(v)$.

First we study $i \in [1, j^*]$. 
Using Fact~\ref{fact:iprev:R+v}, $R_i = R_i^- + v$ holds for each $i \in [1, j^*]$, and $R_i^-=\{ e \in R_{i-1}^- - e_{i-1}^-: \replacementTester{}(I_{i-1}^-, {I'}_{i-1}^-, e, w^-[I_{i-1}^-]) \ne \err\}$ according  to the assumption of this lemma. Besides, by the definition of $j^*$  and according to the break condition in Line~\ref{line:mat:insert:break}, we can conduct that $\replacementTester{}(I_{i-1}^-, {I'}_{i-1}^-, v, w^-[I_{i-1}^-]) \ne \err$. Putting everything together we have,
\begin{align*}
    R_i &= R_i^- + v \\
    &= \{ e \in R_{i-1}^- - e_{i-1}^-: \replacementTester{}(I_{i-1}^-, {I'}_{i-1}^-, e, w^-[I_{i-1}^-]) \ne \err\} + v
    \\
    &= \{ e \in R_{i-1}^- + v - e_{i-1}^-: \replacementTester{}(I_{i-1}^-, {I'}_{i-1}^-, e, w^-[I_{i-1}^-]) \ne \err\} \enspace.
\end{align*}
As stated in the previous part, we know that $w^-[I_{i - 1}^-] = w[I_{i - 1}]$ for any $i$ that $i-1< j^*\leq i^*$. 
Using this fact alongside $R_{i-1}=R_{i-1}^- + v$ by Fact~\ref{fact:iprev:R+v}, $e_{i-1}=e_{i-1}^-$ by Fact~\ref{fact:iprev:e}, $I_{i-1}=I_{i-1}^-$ by Fact~\ref{fact:iprev:I}, and ${I'}_{i-1}={I'}_{i-1}^-$ by Fact~\ref{fact:iprev:I'}, we have: 
$$
R_i = \{ e \in R_{i-1} - e_{i-1}: \replacementTester{}(I_{i-1}, {I'}_{i-1}, e, w[I_{i - 1}]) \ne \err\} \enspace.
$$

Recall that if case 1 holds, $i^*=j^*$, and then the survivor invariant partially holds for the first $i^*$ levels.
Otherwise, if Case 2 holds, it remains to study $i\in [j^*+1, i^*]=[j^*+1,T+1]$.

It holds that $R_i = R_i^-$ for any $i \in [j^*, i^*]$  according to Fact~\ref{fact:iprev:R}. Adding  it to the assumption of this lemma, we have: 
$$
    R_i = R_i^-
    = \{ e \in R_{i-1}^- - e_{i-1}^-: \replacementTester{}(I_{i-1}^-, {I'}_{i-1}^-, e, w^-[I_{i-1}^-]) \ne \err\} \enspace .
$$

Besides, $e_{i-1}=e_{i-1}^-$ by Fact~\ref{fact:iprev:e}, $I_{i-1}=I_{i-1}^-$ by Fact~\ref{fact:iprev:I}, and ${I'}_{i-1}={I'}_{i-1}^-$ by Fact~\ref{fact:iprev:I'}. Using these fact alongside $w^-[I_{i - 1}^-] = w[I_{i - 1}]$, which was stated beforehand we have
$$
    R_i
    = \{ e \in R_{i-1}^- - e_{i-1}: \replacementTester{}(I_{i-1}, {I'}_{i-1}, e, w[I_{i-1}]) \ne \err\} \enspace .
$$
We have either $i=j^*+1$ or $i\in (j^*+1, i^*]$. If $i\in (j^*+1, i^*]$, then $R_{i - 1}=R_{i - 1}^-$ by Fact~\ref{fact:iprev:R}. Hence,
$$
    R_i
    = \{ e \in R_{i-1} - e_{i-1}: \replacementTester{}(I_{i-1}, {I'}_{i-1}, e, w[I_{i-1}]) \ne \err\} \enspace .
$$
Now consider $i=j^*+1$. According to Fact~\ref{fact:iprev:R+v}, $R_{j^*}=R_{j^*}^-+v$, and then $R_{j^*}^-=R_{j^*}\setminus\{v\}$. Due to the definition of $j^*$, we know $v$ is not added to $R_{j^*+1}$. Hence, according to the break condition in Line~\ref{line:mat:insert:break}, we can conduct that  $\replacementTester{}(I_{j^*}, {I'}_{j^*}, e, w[I_{j^*}]) = \err$. 
Putting everything we have
\begin{align*}
    R_{j^*+1} &= \{ e \in R_{j^*}^- - e_{j^*}: \replacementTester{}(I_{j^*}, {I'}_{j^*}, e, w^-[I_{j^*}]) \ne \err\} \\
     &= \{ e \in R_{j^*} \backslash \{v\} - e_{j^*}: \replacementTester{}(I_{j^*}, {I'}_{j^*}, e, w^-[I_{j^*}]) \ne \err\} \\
    &= \{ e \in R_{j^*} - e_{j^*}: \replacementTester{}(I_{j^*}, {I}_{j^*}, e, w[I_{j^*}]) \ne \err\} \enspace ,
\end{align*}
as claimed.

\textbf{Completing the proof.}
Now having everything in hand, we can complete the proof of this lemma. Above, we show that the starter, survivor, weight, and independent invariants partially hold for the first $i^*$ levels.

If Case 1 holds, it means we set $e_{i^*}$ to $v$. 
We also set $w(e_{i^*}) = f(I'_{i^* - 1} + e_{i^*}) - f(I'_{i^* -1})$ in Line~\ref{line:mat:insert:setw},
$I_{i^*} = I_{i^*-1} + e_{i^*} - \replacementTester{}(I_{i^*-1}, I'_{i^*-1}, e_{i^*}, w[I_{i^* - 1}])$
and $I'_i = I'_i + e_{i^*}$ in Line~\ref{line:mat:insert:setI}.
Then in the Line~\ref{line:mat:insert:setRi+1},
we set $R_{i^*+1} = \{e' \in R_{i^*}: \replacementTester{}(I_{i^*}, I'_{i^*}, e',w[I_{i^*}]) \ne \err \}$.
It implies that the level invariants partially hold for the first $i^*+1$ levels.
Next, we invoke \MatroidConstLevel$(i^*+1)$, and then all the invariants hold by Theorem~\ref{thm:invariants:leveling}. 

Otherwise, if Case 2 holds, $i^*=T+1=T^-+1$.
It means the starter,  survivor, weight, and independent invariants hold and it remains to show that the terminator invariant holds to complete the proof.
Recall that in this case, $j^*<T^-+1$, which implies $(T^-+1)\in[j^*+1, i^*]$ and then $R_{T^-+1}=R_{T^-+1}^-$ by Fact~\ref{fact:iprev:R+v}.
Also, the assumption of this lemma implies $R_{T^-+1}^-=\emptyset$.  Therefore,
$R_{T+1}=R_{T^-+1}=R_{T^-+1}^-=\emptyset$, which means the terminator invariant holds and completes the proof.

\subsection{Proof of Lemma \ref {mat_delete_level}}
\label{subs:mat_delete_level}

Let $i^*$ be the level for which \MatroidConstLevel{} is invoked, and if \MatroidConstLevel{} is never invoked during the execution of \deletev$(v)$, we set $i^*$ to be $T + 1$.
Considering our algorithm in \deletev$(v)$, we know that in levels before $i^*$, or in other words in each level $i \in [1, i^*)$ we do not make any change in our data structure other than removing the element $v$ from $R_{i}^-$ if it has this element in it. 
Hence, we have the following facts about $e$, $w$, $I$, and $I'$, which are similar to the Facts~\ref{fact:iprev:e},~\ref{fact:iprev:w},~\ref{fact:iprev:I},~and~\ref{fact:iprev:I'} in the proof of Lemma~\ref{mat_insert_level}.

\begin{fact}
\label{fact:prev:e}
For any $i \in [1, i^*)$, it holds that $e_i = e_i^-$. 
\end{fact}

\begin{fact}
\label{fact:prev:w}
For any $i \in [1, i^*)$, it holds that $w(e_i) = w^-(e_i)$.
\end{fact}

\begin{fact}
\label{fact:prev:I}
For any $i \in [0, i^*)$, it holds that $I_i = I_i^-$.
\end{fact}

\begin{fact}
\label{fact:prev:I'}
For any $i \in [0, i^*)$, it holds that ${I'}_i =  {I'}_i^-$.
\end{fact}

In \deletev$(v)$, we removes the element $v$ from $R_{i}^-$ for any $i\in[0,\min(i^*,T^-)]$.
By the definition of $i^*$, we have $i^*\leq T^-+1$, and $i^*=T^-+1$ happens only if \MatroidConstLevel{} has never been invoked during the execution of \deletev$(v)$. In this case, 
$R_{T^-+1} = R_{T^-+1}^-$, and since $R_{T^-+1}^- = \emptyset$ according to the assumption of this lemma, we have $R_{T^-+1} = R_{T^-+1}^-=\emptyset$.
Therefore, $R_i = R_i^- \backslash \{v\}$ also holds when $i=T^-+1$, and as $\min(i^*,T^-+1)=i^*$, we can conduct the following fact.

\begin{fact}
\label{fact:prev:R}
For any $i \in [0, i^*]$, it holds that $R_i = R_i^- \backslash \{v\}$.
\end{fact}

\textbf{Starter invariant. }
Similar to Lemma~\ref{mat_insert_level}, we prove the starter invariant holds, which means $R_0=V$ and $I_0=I'_0=\emptyset$.

Fact~\ref{fact:prev:R} implies that $R_0 = R_0^- \backslash \{v\}$.
Besides, $ R_0^- = V^-$ since the starter invariant holds before the deletion of $v$ by the assumption.
We also know $V=V^- \backslash \{v\}$. Therefore, 
$R_0 = R_0^- \backslash \{v\}=V^- \backslash \{v\}= V$.

We have $I_0=I_0^-$ by Fact~\ref{fact:prev:I} and ${I'}_0={I'}_0^-$ by Fact~\ref{fact:prev:I'}.
Adding these to $I_0^-={I'}_0^-=\emptyset$, which is the assumption of lemma, implies that $I_0 = {I'}_0 = \emptyset$.

\textbf{Weight invariant  (partially). }
We next prove that the weight invariant partially holds for the first $i^*$ levels.
To do this, we show $e_i \in R_i$ and $w(e_i)=f(I'_{i - 1} + e_i) - f(I'_{i - 1})$ hold for any $i\in [1, i^*)$.

We first prove  $e_i \in R_i$ holds for each $i \in (1, i^*)$.
According to the definition of $i^*$ and \deletev{} Algorithm, we know that $e_i^- \ne v$ for any $i < i^*$.
Moreover, the assumption of this lemma implies that $e_i^-\in R_i^-$.
Hence we can say that $e_i^- \in R_i^- \backslash \{v\}$. 
We also know $R_i^- \backslash \{v\} =  R_i$ for any $i\in[0,i^*)$ according to Fact~\ref{fact:prev:R}.
Adding it to $e_i= e_i^-$, which is a result of  Fact~\ref{fact:prev:e}, we have
$$
e_i= e_i^- \in R_i^- \backslash \{v\}  =  R_i \enspace .
$$

We next show $w(e_i) = f(I'_{i - 1} + e_i) - f(I'_{i - 1})$ in  the  same way of Lemma~\ref{mat_insert_level}.
For any $i\in[1,i^*)$,  we have $w(e_i) = w^-(e_i)$ by Fact~\ref{fact:prev:w}, and $e_i=e_i^-$ by Fact~\ref{fact:prev:e}.
Moreover, $w^-(e_i^-) = f(I'^-_{i - 1} + e_i^-) - f(I'^-_{i - 1})$ because of the assumption of this lemma.
Also, $I'^-_{i - 1}=I'_{i - 1}$ holds for any $i\in[1,i^*)$ because of Fact~\ref{fact:prev:I'}.
Putting everything together, for any $i\in[1,i^*)$ we have
$$
w(e_i) = w^-(e_i) 
= w^-(e_i^-)
= f(I'^-_{i - 1} + e_i^-) - f(I'^-_{i - 1}) 
= f(I'_{i - 1} + e_i) - f(I'_{i - 1}) \enspace, 
$$
which completes the proof.

\textbf{Independent invariant (partially). }
We show that the independent invariant partially holds for the first $i^*$ levels,  which means $I_i = I_{i-1} + e_i - \replacementTester{}(I_{i-1}, I'_{i-1}, e_i, w[I_{i - 1}])$ and $I'_i = \cup_{j \le i} I_j$ hold for any $i\in[1,i^*)$.

Same as Lemma~\ref{mat_insert_level}, we first prove $I_i = I_{i-1} + e_i - \replacementTester{}(I_{i-1}, I'_{i-1}, e_i,, w[I_{i - 1}])$ holds for any $i\in[1,i^*)$.
By the assumption of this lemma, $I_{i}^-=I_{i - 1}^- + e_i^- - \replacementTester{}(I_{i - 1}^-, {I'}_{i - 1}^-, e_i^-, w^-[I_{i - 1}^-])$ holds.
Besides, $I_i = I_i^-$  for any any $i\in[1,i^*)$ according to Fact~\ref{fact:prev:I}.
Therefore,
$$ I_i = I_i^- = I_{i - 1}^- + e_i^- - \replacementTester{}(I_{i - 1}^-, {I'}_{i - 1}^-, e_i^-, w^-[I_{i - 1}^-])  \enspace.
$$
Recall that $w[X]$ means $w$ restricted to domain $X$.
Again same as Lemma~\ref{mat_insert_level}, we show that $w^-[I_{i - 1}^-]=w[I_{i - 1}]$ holds for any $i\in[1,i^*)$.
By Fact~\ref{fact:prev:I}, $I_{i-1} = I_{i-1}^-$, which means the domain of $w^-[I_{i - 1}^-]$ and $w[I_{i - 1}]$ are equal.
Moreover, since the independent invariant holds before the deletion of $v$ by the assumption of this lemma, we have $I_{i - 1}^- \subseteq \{e_1^-, \cdots, e_{i - 1}^-\}$.
Besides, for any $j\leq i-1 <i^*$, Fact~\ref{fact:prev:e} implies that $e_j^-=e_j$, which results in $I_{i - 1}^- \subseteq \{e_1, \cdots, e_{i - 1}\}$.
We also have $w(e_j)=w^-(e_j)$ for any $j\leq i-1 <i^*$ using Fact~\ref{fact:prev:w}. It completes the proof of $w^-[I_{i - 1}^-]=w[I_{i - 1}]$.
Hence,
$$
I_i = I_{i - 1}^- + e_i^- - \replacementTester{}(I_{i - 1}^-, {I'}_{i - 1}^-, e_i^-, w[I_{i - 1}]) \enspace.
$$
Adding it to $I_{i - 1}^-=I_{i - 1}$ by Fact~\ref{fact:prev:I}, 
$e_i^-=e_i$ by Fact~\ref{fact:prev:e}, 
and ${I'}_{i - 1}^-={I'}_{i - 1}$ by Fact~\ref{fact:prev:I'},
for any $i\in[1,i^*)$ we have
$$
I_i = I_{i - 1} + e_i - \replacementTester{}(I_{i - 1},  {I'}_{i - 1}, e_i, w[I_{i - 1}]) \enspace.
$$

Next, we show $I'_i = \cup_{j \le i} I_j$ holds for any $i\in[1,i^*)$.
By Fact~\ref{fact:prev:I'}, ${I'}_i =  {I'}_i^-$  holds for any $i\in[1,i^*)$.
Also, ${I'}_i^- = \cup_{j \le i} I_j^-$ is a result of the assumption of this lemma.
Moreover, Fact~\ref{fact:prev:I} implies $I_j^-=I_j$ for any $j\leq i <i^*$.
Therefore, 
$$
{I'}_i =  {I'}_i^- = \cup_{j \le i} I_j^- = \cup_{j \le i} I_j  \enspace.
$$

\textbf{Survivor invariant (partially). }
We next show that the survivor invariant partially holds for the first $i^*$ levels.
To do this, we prove
$R_i = \{ e \in R_{i-1} - e_{i-1}: \replacementTester{}(I_{i-1}, {I'}_{i-1}, e, w[I_{i - 1}]) \ne \err\}$ holds for any $i\in[1,i^*]$.

Using Fact~\ref{fact:prev:R},  we have $R_i = R_i^- \backslash \{v\}$ for each $i \in [1, i^*]$.
Also, the assumption of this lemma implies that 
$R_i^- = \{ e \in R_{i-1}^- - e_{i-1}^-: \replacementTester{}(I_{i-1}^-, {I'}_{i-1}^-, e, w^-[I_{i-1}^-]) \ne \err\}$. Thus,
\begin{align*}
    R_i &= R_i^- \backslash \{v\} \\
    &= \{ e \in R_{i-1}^- - e_{i-1}^-: \replacementTester{}(I_{i-1}^-, {I'}_{i-1}^-, e, w^-[I_{i-1}^-]) \ne \err\} \backslash \{v\} \\
    &= \{ e \in R_{i-1}^- \backslash \{v\} - e_{i-1}^-: \replacementTester{}(I_{i-1}^-, {I'}_{i-1}^-, e, w^-[I_{i-1}^-]) \ne \err\} \enspace.
\end{align*}

As we stated in the previous, we  know that $w^-[I_{i - 1}^-] = w[I_{i - 1}]$ holds when $i-1<i^*$.
Using this fact alongside $R_{i-1}^-=R_{i-1}$ by Fact~\ref{fact:prev:R}, $e_{i-1}^-=e_{i-1}$ by Fact~\ref{fact:prev:e}, $I_{i-1}^-=I_{i-1}$ by Fact~\ref{fact:prev:I}, and ${I'}_{i-1}^-={I'}_{i-1}$ by Fact~\ref{fact:prev:I'}, for each $i\in[1,i^*]$ we have
$$
R_i = \{ e \in R_{i-1} - e_{i-1}: \replacementTester{}(I_{i-1}, {I'}_{i-1}, e, w[I_{i - 1}]) \ne \err\} \enspace.
$$

\textbf{Completing the proof. }
Above, we show that the starter, survivor, weight, and independent invariants hold for the first $i^*$ levels.
To complete the proof of this lemma, we show all level invariants hold.
If $i^* \ne T + 1$, which means that we have invoked \MatroidConstLevel$(i^*)$, our proof is complete by Theorem \ref{thm:invariants:leveling}. 
Otherwise, if $i^* = T + 1$, we just need to show the terminator invariant holds to complete our proof, which means $R_{T + 1} = \emptyset$. 
Note that in this case, $T=T^-$.
Fact~\ref{fact:prev:R} implies that $R_{T + 1}=R_{T + 1}^-$, and the assumption of lemma implies that $R_{T^- + 1}^- = \emptyset$.
Hence,
$$
R_{T + 1} = R_{T + 1}^- = R_{T^- + 1}^- = \emptyset \enspace,
$$
which completes the proof.
\\

\section{Analysis of dynamic algorithm for maximum submodular under cardinality constraint}
\label{sec:analysis:klogk}
In this section, we prove the correctness of \constLevel{}, \deletev{}, and \insertv{} algorithms. We will also compute the query complexity of each one of them. 
Except for the approximation guarantee proof and some parts of the query complexity section, most of the theorems, lemmas, and their proof in this section are similar to Section \ref{sec:matroid:analysis} where we analyze our dynamic algorithm for matroid constraint.
First, we define the following random variables. 

\begin{tcolorbox}[width=\linewidth, colback=white!80!gray,boxrule=0pt,frame hidden, sharp corners]

\textbf{Random Variables:}
\begin{itemize}
    \item  We define the random variable $\bE_i$ for the promoting element $e_i$ that we pick at level $L_i$. 
     \item  We denote by $\bR_i$ the random variable that corresponds to the set $R_i$ of elements at the level $L_i$ and its 
     value is denoted by $R_i$ which is the set of elements that are in the set $R_i$. 
    \item  We define the random variable $\bT$ for the index of the last level that our algorithm creates. 
    Indeed, for a level $L_i$ to be created entirely, the value of the random variable $\bT$ should be $\bT \ge i$.
    
    \item We define
    $H_i = (e_1, \dots, e_{i-1}, R_0, \dots, R_{i})$
    as \emph{the partial configuration up to level $i$}.
    Note that $R_{i}$ is included in this definition, while $e_{i}$ is not.
    We let 
    $\bH_i := (\bE_1, \dots, \bE_{i-1}, \bR_0,  \bR_{1}, \dots, \bR_{i})$ be a random variable that corresponds to $H_i$.
\end{itemize}

\end{tcolorbox}

We break the analysis of our algorithms into a few steps.

\paragraph{Step 1: Analysis of binary search.} 
In the first step, we prove that the binary search that 
we use to speed up the process of finding the right levels for non-promoting elements  
works. Indeed, we prove that if $e \in V$ is a promoting element for a level $L_{z-1}$, 
it is promoting  for all levels $L_{r \le z-1}$ and 
if $e$ is not promoting  for the level $L_{z}$, 
it is not promoting  for all levels $L_{r \ge z}$. 
Therefore, because of this monotonicity property, we can do a binary search to find the smallest $z \in [i,\ell-1]$ so that 
$e$ is promoting  for the level $L_{z-1}$, but it is not promoting  for the level $L_{z}$.

\paragraph{Step 2: Maintaining invariants.}
We define five invariants, and we show that these invariants \emph{hold} when \init{} is run, and our whole data structure gets built, and \emph{are preserved} after every insertion and deletion of an element.

\begin{tcolorbox}[width=\linewidth, colback=white!80!gray,boxrule=0pt,frame hidden, sharp corners]
\textbf{Invariants:} 
\begin{enumerate}
    \item \textbf{Level invariants.} 
    \begin{enumerate}
        \item \textbf{Starter.} $R_0=V$ and $I_0 = \emptyset$
        \item \textbf{Survivor.}  For $1 \leq i \leq T + 1$, $R_i = \{ e \in R_{i-1} - e_{i-1}: \replacementTester{}(I_{i-1}, e) = True\}$ 
        \item  \textbf{Cardinality.} For $1 \leq i \leq T$, $I_i = I_{i-1} + e_i$ where $\replacementTester{}(I_{i-1}, e) = True$
        \item \textbf{Terminator.} $R_{T+1}=\emptyset$
    \end{enumerate}
    \item  \textbf{Uniform invariant.} For all $i \ge 1$, condition on the random variables $\bH_i$, the element $e_i$ is chosen uniformly at random from the set $R_i$. That is, $ \Pr{\bE_i = e |  \bT \geq i \text{ and } \bH_i = H_i  }= \frac{1}{|R_i|}\cdot \ind{e\in R_i}  $.
\end{enumerate}
\end{tcolorbox}

The survivors invariant says that all elements that are added to $R_i$ at a level $L_i$ 
are promoting elements for that level. 
In another words, those elements of the set $R_{i-1} - e_{i-1}$ 
that are not promoting will be filtered out and not be seen in $R_i$. 
The terminator invariant shows that the recursive construction of levels stops when the survivor set becomes empty. 
The cardinality invariant shows that the sets $I_i$ are constructed by adding a new element that promotes the previous set. That means the new element adds at least $\tau$ to the submodular value of $I_{i-1}$, and the size of $I_i$ remains at most $k$.
Intuitively, these level invariants provide us the approximation guarantee.

The uniform invariant asserts that for every level $L_{i \in [T]}$, 
condition on $\bT \geq i \text{ and } \bH_i = H_i $, 
the element $e_i$ is chosen uniformly at random from the set $R_i$. 
That is, $\Pr{e_i = e | \bT \geq i \text{ and } \bH_i = H_i } = \frac{1}{|R_i|} \cdot \ind{e \in R_i}$.
Intuitively, this invariant provides us with the randomness that we need  
to fool the adversary in the (fully) dynamic model 
which in turn helps us to develop a dynamic algorithm 
for the submodular maximization under the cardinality constraint.

\paragraph{Step 3: Query complexity.}
In the third part of the proof, we bound 
the worst-case expected  query complexity of the leveling algorithm, and later,  we show that if the \textbf{uniform} invariant holds, we can bound the worst-case expected query complexity of the insertion and deletion operations. 

\paragraph{Step 4: Approximation guarantee.}
Finally, in the last step of the proof, we show that if the \textbf{level} invariants 
is fulfilled, we can report a set $I_T$ with $|I_T|\le k$ 
whose submodular value is an $(2+\epsilon)$-approximation of the optimal value.

\subsection{Monotone property of promotable levels and binary search argument}
Recall that we defined the function 
$\suit(I_j, e)$ for an element $e \in V$ 
with respect to the level $L_j$ which 
\begin{itemize}
    \item returns $True$ if $\marginalgain{e}{I_j} \ge \tau$  and  $|I_j| < k$;
    \item returns $False$ otherwise.
\end{itemize}

\begin{lemma}
\label{lm:cardinality:binary_search_argument}
Let $L_j$ be an arbitrary level of the Algorithm~\levelingconstraint{}, where $1 \le j \le T$. 
Let $e \in V$ be an arbitrary element of the ground set. 
If $\suit(I_{j-1}, e)$ returns $False$, then 
$\suit(I_j, e)$ returns $False$. 
\end{lemma}   

\begin{proof}
Recall that in Line \ref{line:cardinality:set_I_l}, we set $I_j = I_{j-1} + e_j$.
Suppose that $\suit(I_{j-1}, e)$ returns $False$. 
It means that either $\marginalgain{e}{I_{j-1}} \ge \tau$  or  $|I_{j-1}| < k$ not hold.
If $|I_{j-1}| >=k$, since $|I_j| = |I_{j-1}|+1$, we can conclude that $|I_j| > k$ which means $\suit(I_j, e) = False$.
Moreover, if $\marginalgain{e}{I_{j-1}} < \tau$, since $I_{j-1} \subset I_j$ and $f$ is submodular, we have $\marginalgain{e}{I_j} \le \marginalgain{e}{I_{j-1}} < \tau$, which means that $\suit(I_j, e) = False$.
\end{proof}

Using Lemma \ref{lm:cardinality:binary_search_argument}, and by applying a simple induction, 
we can show the function $\suit(I_j, e)$ 
is monotone.
Thus, for every arbitrary element $e$, 
it is possible to perform a binary search on an interval $[i, \ell-1]$ 
to find the smallest $z\in [i, \ell-1]$ such that $\suit(I_{z-1}, e) = True$ and $\suit(I_z, e) = False$.

\subsection{Correctness of invariants after \constLevel{} is called}
In this section, we show that the invariants that we defined above will hold 
at the end of the algorithm $\constLevel{}(j)$. 
There are many similarities between this section and Section \ref{sec:level:proofs:matroid}.
However, in Section \ref{sec:level:proofs:matroid}, where we prove the correctness of invariants after calling $\MatroidConstLevel{}$, we have independent and weight invariants instead of the cardinality invariant.

We first explain what we mean by stating that level invariants partially hold.

\begin{tcolorbox}[width=\linewidth, colback=white!80!gray,boxrule=0pt,frame hidden, sharp corners]
\begin{definition}
    For $j \ge 1$, we say that 
    \emph{the level invariants partially hold for the first $j$ levels} if the followings hold.
    \begin{enumerate}
        \item \textbf{Starter.} $R_0 = V$ and $I_0 = \emptyset$
        \item \textbf{Survivor.}  For $1 \leq i \le j$, $R_i = \{ e \in R_{i-1} - e_{i-1}: \replacementTester{}(I_{i-1}, e) = True\}$ 
        \item  \textbf{Cardinality.} For $1 \leq i \le j-1$, $I_i = I_{i-1} + e_i$ where $\suit{}(I_{i-1}, e_i)=True$
    \end{enumerate}
\end{definition}

\end{tcolorbox}

We want to prove the following theorem that says after invoking $\constLevel{}(j)$, 
all level invariants hold. 
We break the proof of this theorem into four parts where 
we separately prove the survivor, cardinality, and terminator invariant hold in Lemmas~\ref{lm:cardinality:survivor:leveling},~\ref{lm:cardinality:cardinality:leveling}, and\ref{lm:cardinality:terminator:leveling}, respectively. 
Also, note that the starter invariant holds by the assumption of our theorem.

\begin{theorem}
\label{thm:cardinality:invariants:leveling}
If before calling $\constLevel{}(j)$, the level invariants partially hold for the first $j$ levels,
then after the execution of $\constLevel{}(j)$, level invariants fully hold.
\end{theorem}

\begin{lemma}[Survivor invariant]
\label{lm:cardinality:survivor:leveling}
If before calling $\constLevel{}(j)$, the level invariants partially hold for the first $j$ levels,
then after its execution, the survivor invariant fully holds.
\end{lemma}

\begin{proof}
First of all, we assume that $R_j \neq \emptyset$, otherwise $T = j-1$ and we are done. 
As we have in Algorithm \constLevel{}, let $P$ be a random permutation of the set $R_j$. 
Let us fix an arbitrary element $e \in P$ and suppose that 
at the time when we see $e \in P$, the current level is $L_{\ell}$ for $\ell \ge j$. 
We have two cases. 
Either $e$ is a promoting element for the level $L_{\ell-1}$ or it is not promoting for the level $L_{\ell-1}$. 

First, assume that $e$ is a promoting element for the level $L_{\ell-1}$. 
We then let $e_{\ell}$ be $e$, perform a set of computations, and then start the new level. 
In particular, the element $e$ is not added to $R_{\ell+1}$ and so, it will not appear in any set $R_{z > \ell}$.
Recall that Lemma~\ref{lm:cardinality:binary_search_argument} proves if $e$ 
is not a promoting element with respect to a level $L_{z}$, 
it will not be a promoting element for the next level $L_{x}$ where $z\le x \le T$.
On the other hand, since $e$ is a promoting element for the level $L_{\ell-1}$, 
we add $e$ to all previous sets $R_{j+1},\cdots, R_{\ell}$. 

Next, we consider the latter case where $e$ is not a promoting element for the level $L_{\ell-1}$. 
That is, $\suit(I_{\ell-1}, e)$ is $False$.
This essentially means that if we inductively apply the argument of Lemma~\ref{lm:cardinality:binary_search_argument}, 
there exists an integer $z \in [j,\ell)$ for which 
$\suit(I_{z-1}, e)$ is $True$, but $\suit(I_{z}, e)$ is $False$. 
This means $e$ is a promoting element for all levels $L_j,\cdots,L_{z-1}$ and 
it is not promoting for levels $L_{z},\cdots,L_T$. 
According to function $\constLevel$, 
we insert the element $e$ into sets $R_{j+1}, \cdots, R_{z}$. 
Hence, after the execution of $\constLevel{}(j)$, 
the survivors invariant holds. 
\end{proof}

\begin{lemma}[Cardinality invariant]
\label{lm:cardinality:cardinality:leveling}
If before calling $\constLevel{}(j)$, the level invariants partially hold for the first $j$ levels, 
then after its execution, the cardinality invariant fully holds. 
\end{lemma}

\begin{proof}
In the execution of $\constLevel{}(j)$, the variable $\ell$ is set to $j, j+1,\cdots, T,T+1$. Therefore, for each $\ell \in  [j,T]$, we set $I_{\ell} = I_{\ell-1} + e_\ell$ in Line \ref{line:cardinality:set_I_l}, where according to Line \ref{line:cardinality:check_e_is_promoting}, $\replacementTester(I_{\ell-1}, e_{\ell})=True$.
\end{proof}
\begin{lemma}[Terminator invariant]
\label{lm:cardinality:terminator:leveling}
If before calling $\constLevel{}(j)$, the level invariants partially hold for the first $j$ levels,
then after its execution, the terminator invariant fully holds. 
\end{lemma}

\begin{proof}
According to Line~\ref{line:cardinality:constlevelmatroid:addR_bs} and the variable $z$, if we add an element $e$ to $R_r$ at some point of time, then $r\leq z\leq \ell-1$ holds at that moment.
Since the variable $\ell$ never decreases during the execution of $\constLevel{}(j)$ and we return $\ell-1$ as $T$ at the end, we can conduct that no element has been added to  $R_{T+1}$, and then $R_{T+1}=\emptyset$, which means the terminator invariant holds. 
\end{proof}

\begin{lemma}[Uniform invariant]
\label{lm:cardinality:uniform:leveling}
If $\constLevel{}(j)$ is invoked and 
the level invariants are going to hold after its execution, then for any $j \ge i$ we have $ \Pr{\bE_i = e |  \bT \geq i \text{ and } \bH_i = H_i  }= \frac{1}{|R_i|}\cdot \ind{e\in R_i}  $.
\end{lemma}

\begin{proof}
In the beginning of $\constLevel{}(j)$, we take a random permutation of elements in $R_j$. Making a random permutation is equal to sampling all elements without replacement. 
    In other words, instead of fixing a random permutation $P$ of $R_j$ and iterating through $P$ in Line \ref{line:cardinality:iterate_P}, we repeatedly sample a random element $e$ from the unseen elements of $R_j$ until we have seen all of the elements. 
    Thus, for proving this lemma, we are assuming that our algorithm uses sampling without replacement.

    Given this view, we make the following claims.
    
    \textbf{Observation 1.}
    $e_i$ is the first element of $R_i$ seen in the permutation.~\\
            This is because before $e_i$ is seen,
            the value of $\ell$ is at most $i$. It is also clear from the algorithm that when an element $e$ is considered, it can only be added to sets $R_{x}$ for $x\le \ell$, both when $\suit{}(I_{\ell-1},e)=Ture$ and when $\suit{}(I_{\ell-1},e)=False$.
            Furthermore, $e$ can only be added to $R_{\ell}$ if $e=e_{\ell}$. Therefore, no element can be added to $R_i$ before $e_i$ is seen.
            
    \textbf{Observation 2.}
    Once $e_1, \dots, e_{i-1}$ have been seen, the set $R_i$ is uniquely determined.\\
    Note that $R_i$ is uniquely determined \emph{even though the algorithm has not observed its elements yet}. 
    This is because regardless of the randomness of $\constLevel{}(j)$, the level invariants will hold after its execution. 
    This implies that the content of the set $R_i$ only depends on the value of $(\bE_1, \dots, \bE_{i - 1})$, which is not going to change after it is set to be equal to $(e_1, \dots, e_{i - 1})$. 
    
    Let the random variable $\bM_i$ denote the sequence of elements that our algorithm observes until setting $\bE_{i-1}$ to be $e_{i - 1}$, including $e_{i-1}$ itself.
    In other words, if $e_{i-1}$ is the $x$-th element of the permutation $P$, $M_i$ is the first $x$ elements of $P$.
    ~\\
    Based on the above facts, conditioned on $\bM_i=M_i$, 
    \textbf{(a) }the value of $\bR_i$, or in other words $R_i$ is uniquely determined.
    \textbf{(b)} $e_i$ is going to be the first element of $R_i$ that the algorithm observes. Therefore,
    since we assumed that the algorithm uses sampling without replacement, $\bE_i$ is going to have a uniform distribution over $R_i$, i.e.,
    \begin{align*}
        \Pr{\bE_i = e | 
        \bT \ge i,
        \bM_i = M_i
        }
        = \frac{1}{|R_i|} \ind{e \in R_i} \enspace .
    \end{align*}

    By the law of total probability, we have 
    \begin{align*}
        \Pr{\bE_i = e_i | 
        \bT \ge i,
        \bH_i = H_i
        }
        &=
        \Exu{M_i}{
        \Pr{\bE_i = e_i | 
        \bT \ge i,
        \bH_i = H_i,
        \bM_i=M_i
        }
        } \enspace , 
    \end{align*}
    where the expectation is taken over all $M_i$ with positive probability.

    Also, note that knowing that $\bM_i = M_i$ uniquely determines the value of $\bH_i$ as well. This is because $M_i$ includes $(e_1, \dots, e_{i-1})$ and, with similar reasoning to what we used for Observation 2, we can say that $R_1, \dots, R_{i}$ are uniquely determined by $(e_1, \dots, e_{i-1})$.
    
Since we are only considering $M_i$ with positive probability, and $\bH_i$ is a function of $\bM_i$ given the discussion above, all the forms of $M_i$ that we consider in our expectation are the ones that imply $\bH_i = H_i$. Therefore, we can drop the condition $\bH_i=H_i$ from the condition $\bH_i=H_i, \bM_i = M_i$, which implies
\[
        \Pr{\bE_i = e_i | \bT \ge i, \bH_i = H_i }
        = \Exu{M_i}{ \Pr{\bE_i = e_i | \bT \ge i, \bM_i=M_i} }
        = \Exu{M_i}{ \frac{1}{|R_i|} \ind{e_i \in R_i} }
        =\frac{1}{|R_i|} \ind{e_i \in R_i} \enspace ,
\]
    as claimed.
\end{proof}

\subsection{Correctness of invariants after an update}
In our dynamic model, 
we consider a sequence $\mathcal{S}$ of updates of elements of an underlying ground set $V$ 
where at time $t$ of the sequence $\mathcal{S}$, we observe an update which can be the deletion of an element $e \in V$ or insertion of an element $e \in V$. 
We assume that an element $e$ can be deleted at time $t$, if after the first time $e$ is inserted, it is not deleted until time $t$. 
In this section, similar to Section \ref{sec:update:matroid}, we prove that all invariants hold after each update.
While most of this section is identical to Section \ref{sec:update:matroid}, we have the cardinality invariant instead of independent and weight invariants in this section.

We use several random variables for our analysis, including $\bE_i$, $\bR_i$, $\bT$, and $\bH_i$. Upon observing an update at time $t$, 
we should distinguish between each of these random variables and their corresponding values before and after the update.
To do so, we use the notations $\bY^-$ and $Y^{-}$ to denote a random variable and its value before time $t$ when $e$ is either deleted or inserted, and we keep using $\bY$ and $Y$ to denote them at the current time after the execution of update.
As an example,
    $\bH_i^- := (\bE_1^-, \dots, \bE_{i-1}^-, \bR_0^-,  \bR_{1}^-, \dots, \bR_{i}^-)$ 
    is the random variable that corresponds to the partial configuration $H_i^- = (e_1^-, \dots, e_{i-1}^-, R_0^-, \dots, R_{i}^-)$.

\subsubsection{Correctness of invariants after every insertion}

We first consider the case when the update at time $t$ of the sequence $\mathcal{S}$ is an insertion of an element $v$. 
In this section, we prove the following theorem.

\begin{theorem}
\label{cardinality:mat_insert:invariants}
If before the insertion of an element $v$, the level invariants and uniform invariant hold, then they also hold after the execution of \insertv$(v)$. 
\end{theorem}

We break the proof of this theorem into Lemmas ~\ref{cardinality:mat_insert_level} and ~\ref{cardinality:mat_insert_uni}.

\begin{lemma} [Level invariants]
\label{cardinality:mat_insert_level}

If before the insertion of an element $v$ the level invariants (i.e., starter, survivor, cardinality, and terminator) hold, 
then they also hold after the execution of \insertv$(v)$. 
\end{lemma}

\begin{proof}
To prove the lemma, we first mention some useful facts and then show that the starter, cardinality, and survivor invariants partially hold. Finally, we prove that all level invariants hold.

We begin with defining variables $i^*$ and $j^*$ as follows.

\begin{itemize}
    \item \textbf{$i^*$: } If during the execution of \insertv$(v)$ there is  $i\in [T]$ such that $e_{i}$ has been set to be $v$, which also implies that we have invoked \constLevel $(i + 1)$, then we set $i^*$ to be $i$. Otherwise, we set $i^*$ to be $T + 1$. 
    \item \textbf{$j^*$: } Let $j^*$ be the largest $i\in [0, T^-+1]$ such that we have added $v$ to $R_i^-$.
\end{itemize}

We consider these two cases in this proof. 
\begin{itemize}
    \item Case 1: $i^* \leq T$, which means $e_{i^*} = v$ and therefore $j^* = i^*$. It also means that we have invoked \constLevel $(i^* + 1)$.
    \item Case 2: $i^* = T + 1$, which means \constLevel{} has never been invoked  during the insertion of $v$.
    Note that in this case, $T = T^-$ and therefore, $j^* < T^-+1 = T+1 =i^*$.
\end{itemize}

Considering our algorithm in \insertv$(v)$, 
it is clear that for any $i < i^*$, we have not made any kind of change in $e_i^-$ or $I^-_i$ at least until \constLevel{} is invoked, if it ever gets invoked. 
Additionally, according to \constLevel, we know that if we have invoked \constLevel$(i^* + 1)$, there has not been any alteration to the variables regarding previous levels. 
Hence, we can conduct the following facts.

\begin{fact}
\label{fact:cardinality:iprev:e}
For any $i \in [1, i^*)$, we have $e_i = e_i^-$.
\end{fact}

\begin{fact}
\label{fact:cardinality:iprev:I}
For any $i \in [0, i^*)$, we have $I_i = I_i^-$.
\end{fact}

By the definition of $j^*$, we have added the element $v$ to the set $R_i^-$, for each $i \in [0, j^*]$.
Recall that $j^* \leq i^*$, and by invoking constLevel$(i^* + 1)$, there has not been any alteration to the variables regarding previous levels. It leads to the following fact.
\begin{fact}
\label{fact:cardinality:iprev:R+v}
For any $i \in [0, j^*]$, we have $R_i = R_i^- + v$.
\end{fact}

We know that if Case 2 holds, which means \constLevel{} has never been invoked during the insertion of $v$, we have $R_i = R_i^-$  for any $i \in [j^* + 1, T + 1]$. Recall that in Case 2, $i^* = T + 1$, and therefore $[j^*+1, T+1] = [j^*+1, i^*]$. Also if Case 1 holds, $j^* = i^*$, so $[j^* + 1, i^*] = \emptyset$. Thus, independent of the case, we can have the following fact.

\begin{fact}
\label{fact:cardinality:iprev:R}
 For any $i \in [j^*+1, i^*]$, we have $R_i = R_i^-$.
\end{fact}

In the following, we first prove that the starter invariant holds after executing $\insertv(v)$. We next show that the cardinality and survivor invariants partially hold for the first $i^*+1$ levels. Finally, we complete the proof by proving that all the level invariants hold.

\textbf{Starter invariant.}
To show that the starter invariant holds after $\insertv(v)$, we need to prove $R_0= V$ and $I_0=\emptyset$.
By the assumption of this lemma, we have $R_0^-=V^-$, and Fact~\ref{fact:cardinality:iprev:R+v} results that $R_0 = R_0^- + v$. Thus $R_0 = R_0^- + v = V^- + v = V$.

Again by the assumption of this lemma, $I_0^-=\emptyset$.
Due to Fact~\ref{fact:cardinality:iprev:I}, we have $I_0=I^-_0$, and therefore, it is clear that $I_0 = I_0^- = \emptyset$.

\textbf{Cardinality invariant (partially).}
Now we show that the cardinality invariant partially holds up to the level~$L_{i^*}$ after $\insertv(v)$.
To do this, we prove $I_i = I_{i - 1} + e_i$ and $\replacementTester{}(I_{i - 1}, e_i)=True$ holds for all $i\in[1,i^*)$.

Using Fact~\ref{fact:cardinality:iprev:I}, we have $I_i = I_i^-$ for any $i\in  [1,i^*)$. Also, $I_i^- = I_{i - 1}^- + e_i^-$ by
the assumption of this lemma.  Then, 
$$ I_i = I_i^- = I_{i - 1}^- + e_i^- \enspace.
$$

For any $i\in[1,i^*)$, we have $I_{i - 1}=I_{i - 1}^-$ by Fact~\ref{fact:cardinality:iprev:I} and $e_i=e_i^-$ by Fact~\ref{fact:cardinality:iprev:e}. Then we can further conclude that:
$$
I_i = I_{i - 1} + e_i \enspace.
$$

Finally, we prove $\replacementTester{}(I_{i - 1}, e_i)=True$ for any $i\in[1,i^*)$. By the assumption of this lemma, we know that $\replacementTester{}(I_{i - 1}^-, e_i^-)=True$. In addition, for any $i\in[1,i^*)$, we have $I_{i - 1}=I_{i - 1}^-$ by Fact~\ref{fact:cardinality:iprev:I} and $e_i=e_i^-$ by Fact~\ref{fact:cardinality:iprev:e}. Thus, $\replacementTester{}(I_{i - 1}, e_i)=True$ for any $i\in[1,i^*)$.

\textbf{Survivor invariant (partially). }
Next, we show that the survivor invariant partially holds for the first $i^*$ levels by proving that $R_i=\{ e \in R_{i-1} - e_{i-1}: \replacementTester{}(I_{i-1}, e) = True\}$ holds for any $i\in[1,i^*]$.
In the following, we first consider $i \in [1, j^*]$ and then $i \in [j^*+1, i^*]$.
Recall that $j^*$ is the largest $j$ that we added $v$ to $R_j$ in $\insertv(v)$.

First we study $i \in [1, j^*]$. 
Using Fact~\ref{fact:cardinality:iprev:R+v}, $R_i = R_i^- + v$ holds for each $i \in [1, j^*]$, and $R_i^-=\{ e \in R_{i-1}^- - e_{i-1}^-: \replacementTester{}(I_{i-1}^-, e) =True\}$ according  to the assumption of this lemma. Besides, by the definition of $j^*$  and according to the break condition in Line~\ref{line:cardinality:insert:break}, we can conduct that $\replacementTester{}(I_{i-1}^-, v) = True$. Putting everything together we have,
\begin{align*}
    R_i &= R_i^- + v \\
    &= \{ e \in R_{i-1}^- - e_{i-1}^-: \replacementTester{}(I_{i-1}^-, e) = True\} + v
    \\
    &= \{ e \in R_{i-1}^- + v - e_{i-1}^-: \replacementTester{}(I_{i-1}^-, e) = True\} \enspace.
\end{align*}
Using $R_{i-1}=R_{i-1}^- + v$ by Fact~\ref{fact:cardinality:iprev:R+v}, $e_{i-1}=e_{i-1}^-$ by Fact~\ref{fact:cardinality:iprev:e}, and $I_{i-1}=I_{i-1}^-$ by Fact~\ref{fact:cardinality:iprev:I}, we have: 
$$
R_i = \{ e \in R_{i-1} - e_{i-1}: \replacementTester{}(I_{i-1}, e) = True\} \enspace.
$$

Recall that if case 1 holds, $i^*=j^*$, and then the survivor invariant partially holds for the first $i^*$ levels.
Otherwise, if Case 2 holds, it remains to study $i\in [j^*+1, i^*]=[j^*+1,T+1]$.

It holds that $R_i = R_i^-$ for any $i \in [j^*, i^*]$  according to Fact~\ref{fact:cardinality:iprev:R}. Adding  it to the assumption of this lemma, we have: 
$$
    R_i = R_i^-
    = \{ e \in R_{i-1}^- - e_{i-1}^-: \replacementTester{}(I_{i-1}^-, e) = True\} \enspace .
$$

Besides, $e_{i-1}=e_{i-1}^-$ by Fact~\ref{fact:cardinality:iprev:e} and $I_{i-1}=I_{i-1}^-$ by Fact~\ref{fact:cardinality:iprev:I}. Thus,
$$
    R_i
    = \{ e \in R_{i-1}^- - e_{i-1}: \replacementTester{}(I_{i-1}, e) = True\} \enspace .
$$
We have either $i=j^*+1$ or $i\in (j^*+1, i^*]$. If $i\in (j^*+1, i^*]$, then $R_{i - 1}=R_{i - 1}^-$ by Fact~\ref{fact:cardinality:iprev:R}. Hence,
$$
    R_i
    = \{ e \in R_{i-1} - e_{i-1}: \replacementTester{}(I_{i-1}, e) = True\} \enspace .
$$
Now consider $i=j^*+1$. According to Fact~\ref{fact:cardinality:iprev:R+v}, $R_{j^*}=R_{j^*}^-+v$, and then $R_{j^*}^-=R_{j^*}\setminus\{v\}$. Due to the definition of $j^*$, we know $v$ is not added to $R_{j^*+1}$. Hence, according to the break condition in Line~\ref{line:cardinality:insert:break}, we can conduct that  $\replacementTester{}(I_{j^*}, e) = False$. 
Putting everything we have
\begin{align*}
    R_{j^*+1} &= \{ e \in R_{j^*}^- - e_{j^*}: \replacementTester{}(I_{j^*}, e) = True\} \\
     &= \{ e \in R_{j^*} \backslash \{v\} - e_{j^*}: \replacementTester{}(I_{j^*}, e) = True\} \\
    &= \{ e \in R_{j^*} - e_{j^*}: \replacementTester{}(I_{j^*}, e) = True\} \enspace ,
\end{align*}
which finishes the proof.

\textbf{Completing the proof.}
Now having everything in hand, we can complete the proof of this lemma. Above, we show that the starter, survivor, and cardinality invariants partially hold for the first $i^*$ levels.

If Case 1 holds, it means we set $e_{i^*}=v$ and
$I_{i^*} = I_{i^*-1} + e_{i^*}$ in Line~\ref{line:cardinality:insert:setI}.
Then in the Line~\ref{line:cardinality:insert:setRi+1},
we set $R_{i^*+1} = \{e' \in R_{i^*}: \replacementTester{}(I_{i^*}, e') = True \}$.
It means that the invariants partially hold for the first $i^*+1$ levels.
Next, we invoke \constLevel$(i^*+1)$, and then all the invariants hold by Theorem~\ref{thm:cardinality:invariants:leveling}. 

Otherwise, if Case 2 holds, $i^*=T+1=T^-+1$.
It means the starter,  survivor, and cardinality invariants hold and it remains to show the terminator invariant to complete the proof.
Recall that in this case, $j^*<T^-+1$, which implies $(T^-+1)\in[j^*+1, i^*]$ and then $R_{T^-+1}=R_{T^-+1}^-$ by Fact~\ref{fact:cardinality:iprev:R+v}.
Also, the assumption of this lemma implies $R_{T^-+1}^-=\emptyset$.  Therefore,
$R_{T+1}=R_{T^-+1}=R_{T^-+1}^-=\emptyset$, which means the terminator invariant holds and completes the proof.

\end{proof}

\begin{lemma} [Uniform invariant]
\label{cardinality:mat_insert_uni}
If before the insertion of an element $v$ the level and uniform invariants hold, then the uniform invariant also holds after the execution of \insertv$(v)$. 
\end{lemma}

\begin{proof} 
By the assumption that the uniform invariant holds before the insertion of the element $v$, we mean that for any arbitrary $i$ and any arbitrary element $e$, the following holds:
\[
    \Pr{\bE_i^- = e | \bT^- \geq i, \bH_i^- = H_i^-} = 
    \frac{1}{|R_i^{-}|} \cdot \ind{e \in R_i^{-}} \enspace.
\]
We aim to prove that given %
our assumptions,
after the execution of \insertv$(v)$, for each arbitrary $i$ and each arbitrary element $e$, we have 
\[
    \Pr{\bE_i = e | \bT \geq i, \bH_i = H_i} = 
    \frac{1}{|R_i|} \cdot \ind{e \in R_i} \enspace .
\]

Note that $\Pr{\bE_i = e | \bT \geq i, \bH_i = H_i}$, is only defined when 
$\Pr{\bT \geq i, \bH_i = H_i} > 0$, which means that given the input and considering the behavior of our algorithm including its random choices, it is possible to reach a state where $\bT \geq i$ and $\bH_i = H_i$.
In this proof, we use $\bold{p}_i$ to denote the variable $p_i$ used in the \insertv{} as a random variable.

Fix any arbitrary $i$ and any arbitrary element $e$. Since $\bH_i^- = (\bE_1^-, \dots, \bE_{i - 1}^-, \bR_0^-, \bR_1^-, \dots, \bR_i^-)$ refers to our data structure levels before the insertion of the element $v$, it is clear that the following facts hold about $\bH_i^-$. 

\begin{fact}
\label{cardinality:ebeforeinsert}
For any $j < i$, $\bold{e}_j^- \neq v$.
\end{fact}

\begin{fact}
\label{cardinality:rbeforeinsert}
For any $j \leq i$, $v \notin \bold{R}_j^-$.
\end{fact}

We consider the following cases based on which of the following holds for $H_i = (e_1, \dots, e_{i - 1}, R_0, R_1, \dots, R_i)$:
\begin{itemize}
    \item Case 1: If the $e_j = v$ for some $j < i$.
    \item Case 2: If $v \notin \{e_1, \dots, e_{i - 1}\}$.
\end{itemize}

By Claims \ref{cardinality:insertcase1} and \ref{cardinality:insertcase2}, we prove that no matter the case, $\Pr{\bE_i = e | \bT \geq i, \bH_i = H_i}$ is equal to $\frac{1}{|R_i|} \cdot \ind{e \in R_i}$, which completes the proof of the Lemma.

\begin{claim}
\label{cardinality:insertcase1}
If $H_i$ is such that there is a $1 \le j < i$ that $e_j = v$, then  
$\Pr{\bE_i = e | \bT \geq i, \bH_i = H_i} = \frac{1}{|R_i|} \cdot \ind{e \in R_i}$.
\end{claim}

\begin{proof}
We know that, $\bold{p_j}$ 
must have been equal to 1, as otherwise, instead of having $\bold{e_j} = e_j = v$, we would have had $\bE_j = \bE_j^-$, which would not have been equal to $v$ as stated in Fact \ref{cardinality:ebeforeinsert}. According to our algorithm, since $\bold{p_j}$ has been equal to $1$, we have invoked \constLevel($j + 1$). 
By Lemma~\ref{cardinality:mat_insert_level}, we know that the level invariants hold at the end of the execution of \insertv{}, which is also the end of the execution of \constLevel{$(j + 1)$}. 
Thus, Lemma~\ref{lm:cardinality:uniform:leveling}, proves that $\Pr{\bE_i = e | \bT \geq i, \bH_i = H_i} = \frac{1}{|R_i|}\cdot \ind{e \in R_i}$.
\end{proof}

\begin{claim}
\label{cardinality:insertcase2}
If $H_i$ is such that $e_j \neq v$ for any $1 \le j < i$, then  
$\Pr{\bE_i = e | \bT \geq i, \bH_i = H_i} = \frac{1}{|R_i|} \cdot \ind{e \in R_i}$.
\end{claim}

\begin{proof}
We first define $H_i^-$ based on $H_i$ as $ H_i^- := 
    (R_0 \backslash \{v\} , \dots, R_i \backslash \{v\}, e_1, \dots, e_{i-1})$ 
and prove the following claim. 
\begin{claim}
\label{cardinality:insertclaim}
The events $[\bT \geq i, \bH_i = H_i]$ and $[\bT^- \geq i, \bH_i^- = H_i^-, \bold{p_1} = 0, \dots, \bold{p_{i - 1}} = 0]$ are equivalent and imply each other, thusly they are interchangeable.   
\end{claim}

\begin{proof}
First, we show that if $\bT \geq i, \bH_i = H_i$, then $\bT^- \geq i, \bH_i^- = H_i^-, \bold{p_1} = 0, \dots, \bold{p_{i - 1}} = 0$. 
Considering that case 2 holds for $H_i$, $\bH_i = H_i$, means that for any $j < i$, $\bE_i \neq v$, which means there is no $j < i$ with $\bold{p_j} = 1$. Note that if $\bold{p_j} = 1$, then we would have set $\bold{e_j}$ to be equal to $v$, and we would have invoked \constLevel($j + 1$). Thus, in addition to knowing that for any $j < i$, $\bold{p_j} = 0$, we also know that, we have not invoked \constLevel($j + 1$) for any $j < i$. As for any $j < i$, $\bold{p_j} = 0$ and \constLevel($j + 1$) was not invoked, we have the following results: 
\begin{enumerate}
    \item 
    Level $i$ also existed before the insertion of $v$, i.e. $\bT^- \geq i$.
    \item 
    We have made no change in the values of $(\bold{e_1}, \dots, \bold{e_{i - 1}})$, and they still have the values they had before the insertion of $v$, i.e. for any  $j < i$, $\bold{e_j} = \bold{e_j^-}$, and so $\bold{e_j^-} = e_j$. 
    \item
    All the change we might have made in our data structure is limited to adding the element $v$ to a subset of $\{\bold{R_0^-}, \dots, \bold{R_i^-}\}$. Hence, for any $j \leq i$, whether $\bold{R_j}$ is equal to $\bold{R_j^-}$ or $\bold{R_j^-} \cup \{v\}$, $\bold{R_j^-} = \bold{R_j} \backslash \{v\} = R_j \backslash \{v\}$. 
\end{enumerate}
So far, we have proved that throughout our algorithm, we reach the state, where $\bT \geq i, \bH_i = H_i$, only if $\bT^- \geq i, \bH_i^- = H_i^-, \bold{p_1} = 0, \dots, \bold{p_{i - 1}} = 0$.

We know that in our insertion algorithm, there is not any randomness other than setting the value of $\bold{p_j}$ as long as we have not invoked \constLevel, which only happens when for a $j$, $\bold{p_j}$ is set to be 1. It means that the value of $\bH_i$ can be determined uniquely if we know the value of $\bH_i^-$, and we know that $\bold{p_1}, \dots, \bold{p_{i - 1}}$ are all equal to $0$. 
Since we have assumed that $\bT \geq i, \bH_i = H_i$ is a valid and reachable state in our algorithm, $\bT^- \geq i, \bH_i^- = H_i^-$ must have been a reachable state as well. Plus, $\bT^- \geq i, \bH_i^- = H_i^-, \bold{p_1} = 0, \dots, \bold{p_{i - 1}} = 0$, should imply that $\bT \geq i$ and $ \bH_i = H_i$. Otherwise, $\bT \geq i, \bH_i = H_i$ could not be a reachable state, which is in contradiction with our assumption. 
\end{proof}

Now, we proceed to calculate $\Pr{\bE_i = e|\bT \geq i, \bH_i = H_i}$. 
As stated above, considering that Case 2 holds for $H_i$, we know that 
$\bT \geq i, \bH_i = H_i$ implies that \constLevel{} has not been invoked for any $j < i$. Thus, the value of $\bE_i$ will be determined based on the random variable $\bold{p_i}$. And we have: 
$$\Pr{\bE_i = e|\bT \geq i, \bH_i = H_i} = 
\sum_{p_i \in \{0, 1\}}
(\Pr{\bold{p_i} = p_i|\bT \geq i, \bH_i = H_i} \cdot
\Pr{\bE_i = e|\bT \geq i, \bH_i = H_i, \bold{p_i} = p_i}) $$
According to the algorithm, if $v \in H_i$, then $\Pr{\bold{p_i} = 1|\bT \geq i, \bH_i = H_i}$ is equal to $\frac{1}{|R_i|}$. Otherwise, if $v \notin H_i$, then $\bold{p_i}$ would be zero by default, and $\Pr{\bold{p_i} = 1|\bT \geq i, \bH_i = H_i} = 0$. Hence, we can say that: 
$$\Pr{\bold{p_i} = 1|\bT \geq i, \bH_i = H_i} = 
\frac{1}{|R_i|} \cdot \ind{v \in R_i}.$$
Additionally, Having $\bT \geq i, \bH_i = H_i$, 
 if $\bold{p_i} = 1$, then $\bE_i$ would be $v$. 
Otherwise, if $\bold{p_i} = 0$, then $\bE_i^-$ would remain unchanged, i.e. $\bE_i = \bE_i^-$.  
Hence, $\Pr{\bE_i = e|\bT \geq i, \bH_i = H_i}$ is equal to 
$$
\frac{1}{|R_i|} \cdot \ind{v \in R_i} \cdot
\Pr{\bE_i = e|\bT \geq i, \bH_i = H_i, \bold{p_i} = 1} 
+ (1 - \frac{1}{|R_i|} \cdot \ind{v \in R_i})\cdot
\Pr{\bE_i^- = e|\bT \geq i, \bH_i = H_i, \bold{p_i} = 0}
.$$
We consider the following cases based on the value of $e$: 
\begin{itemize}
\item Case (i): $e = v$.

In this case 
$\Pr{\bE_i = e|\bT \geq i, \bH_i = H_i, \bold{p_i} = 1} = 1$, and $\Pr{\bE_i^- = e|\bT \geq i, \bH_i = H_i, \bold{p_i} = 0} = 0$. Thus, we have: 
$$\Pr{\bE_i = e|\bT \geq i, \bH_i = H_i} = 
\frac{1}{|R_i|} \cdot \ind{v \in R_i} \cdot 1
+ (1 - \frac{1}{|R_i|} \cdot \ind{v \in R_i}) \cdot 0 = \frac{1}{|R_i|} \cdot \ind{v \in R_i} = \frac{1}{|R_i|} \cdot \ind{e \in R_i}.
 $$
\item Case (ii): $e \neq v$
In this case, $\Pr{\bE_i = e|\bT \geq i, \bH_i = H_i, \bold{p_i} = 1} = 0$. So we have: 
$$\Pr{\bE_i = e|\bT \geq i, \bH_i = H_i} =
\frac{1}{|R_i|} \cdot \ind{v \in R_i} \cdot 0
+ (1 - \frac{1}{|R_i|} \cdot \ind{v \in R_i}) \cdot 
\Pr{\bE_i^- = e|\bT \geq i, \bH_i = H_i, \bold{p_i} = 0}. $$
According to the claim that we proved beforehand, $\bT \geq i, \bH_i = H_i$ and $\bT^- \geq i, \bH_i^- = H_i^-, \bold{p_1} = 0, \dots, \bold{p_{i - 1}} = 0$ are interchangeable. So we have: 
$$\Pr{\bE_i = e|\bT \geq i, \bH_i = H_i} =
(1 - \frac{1}{|R_i|} \cdot \ind{v \in R_i}) \cdot 
\Pr{\bE_i^- = e| \bT^- \geq i, \bH_i^- = H_i^-, \bold{p_1} = 0, \dots, \bold{p_i} = 0}.
$$
Since for any $j \leq i$, $\bE_i^-$ and $\bold{p_i}$ are independent random variables, we have: 
\begin{align*}
\Pr{\bE_i = e|\bT \geq i, \bH_i = H_i} =
(1 - \frac{1}{|R_i|} \cdot \ind{v \in R_i}) \cdot 
\Pr{\bE_i^- = e| \bT^- \geq i, \bH_i^- = H_i^-} \\
= (1 - \frac{1}{|R_i|} \cdot \ind{v \in R_i}) \cdot 
\left(\frac{1}{|R_i^-|} \cdot \ind{e \in R_i^-}\right) \enspace, 
\end{align*}
where the last equality holds because of the assumption stated in Lemma. From the definition of $H_i^-$, we have $R_i^- = R_i \backslash \{v\}$. Therefore, 
$$\Pr{\bE_i = e|\bT \geq i, \bH_i = H_i} = 
\frac{|R_i| - \ind{v \in R_i}}{|R_i|} \cdot 
\left(\frac{1}{|R_i| - \ind{v \in R_i} } \cdot \ind{e \in R_i \backslash \{v\}}\right) \enspace .
$$
And since, $e \neq v$, we have: 
$$\Pr{\bE_i = e|\bT \geq i, \bH_i = H_i} =
\frac{1}{|R_i|} \cdot \ind{e \in R_i} \enspace . $$
\end{itemize}
\end{proof}
As stated before, proof of these claims completes the Lemma's proof.
\end{proof}

\subsubsection{Correctness of invariants after every deletion}

Now, we consider the case when the update at time $t$ of the sequence $\mathcal{S}$, is a deletion of an element $v$, and prove the following theorem. 

\begin{theorem}
\label{cardinality:mat_delete:invariants}
If before the deletion of an element $v$, the level invariants and the uniform invariant hold, then they also hold after the execution of \deletev$(v)$. 
\end{theorem}

Similar to Theorem~\ref{cardinality:mat_insert:invariants}, we break the proof of this theorem into 
Lemmas ~\ref{cardinality:mat_delete_level} and ~\ref{cardinality:mat_delete_uni}.

\begin{lemma}[Level invariants]
\label{cardinality:mat_delete_level}
If before the deletion of an element $v$ the level invariants (i.e., starter, survivor, cardinality, and terminator) hold, 
then they also hold after the execution of \deletev$(v)$. 
\end{lemma}

\begin{proof}

Let $i^*$ be the level for which \constLevel{} is invoked, and if \constLevel{} is never invoked during the execution of \deletev$(v)$, we set $i^*$ to be $T + 1$.
Considering our algorithm in \deletev$(v)$, we know that in levels before $i^*$, or in other words in each level $i \in [1, i^*)$ we do not make any change in our data structure other than removing the element $v$ from $R_{i}^-$ if it has this element in it. 
Hence, we have the following facts about $e$ and $I$, which are similar to the Facts~\ref{fact:cardinality:iprev:e} and~\ref{fact:cardinality:iprev:I} in the proof of Lemma~\ref{cardinality:mat_insert_level}.

\begin{fact}
\label{fact:cardinality:prev:e}
For any $i \in [1, i^*)$, it holds that $e_i = e_i^-$. 
\end{fact}

\begin{fact}
\label{fact:cardinality:prev:I}
For any $i \in [0, i^*)$, it holds that $I_i = I_i^-$.
\end{fact}

In \deletev$(v)$, we removes the element $v$ from $R_{i}^-$ for any $i\in[0,\min(i^*,T^-)]$.
By the definition of $i^*$, we have $i^*\leq T^-+1$, and $i^*=T^-+1$ happens only if \constLevel{} has never been invoked during the execution of \deletev$(v)$. In this case, 
$R_{T^-+1} = R_{T^-+1}^-$, and since $R_{T^-+1}^- = \emptyset$ according to the assumption of this lemma, we have $R_{T^-+1} = R_{T^-+1}^-=\emptyset$.
Therefore, $R_i = R_i^- \backslash \{v\}$ also holds when $i=T^-+1$, and as $\min(i^*,T^-+1)=i^*$, we can conduct the following fact.

\begin{fact}
\label{fact:cardinality:prev:R}
For any $i \in [0, i^*]$, it holds that $R_i = R_i^- \backslash \{v\}$.
\end{fact}

\textbf{Starter invariant. }
Similar to Lemma~\ref{cardinality:mat_insert_level}, we prove the starter invariant holds, which means $R_0=V$ and $I_0=\emptyset$.

Fact~\ref{fact:cardinality:prev:R} implies that $R_0 = R_0^- \backslash \{v\}$.
Besides, $ R_0^- = V^-$ since the starter invariant holds before the deletion of $v$ by the assumption.
We also know $V=V^- \backslash \{v\}$. Therefore, 
$R_0 = R_0^- \backslash \{v\}=V^- \backslash \{v\}= V$.

We have $I_0=I_0^-$ by Fact~\ref{fact:cardinality:prev:I}.
Adding it to $I_0^-=\emptyset$, which is an assumption of this lemma, implies that $I_0 = \emptyset$.

\textbf{Cardinality invariant (partially). }
We show that the cardinality invariant partially holds for the first $i^*$ levels,  which means $I_i = I_{i-1} + e_i$ and $\replacementTester{}(I_{i-1}, e_i)$ hold for any $i\in[1,i^*)$.

Same as Lemma~\ref{cardinality:mat_insert_level}, we first prove $I_i = I_{i-1} + e_i$ holds for any $i\in[1,i^*)$.
By the assumption of this lemma, $I_{i}^-=I_{i - 1}^- + e_i^-$ holds.
Besides, $I_i = I_i^-$  for any any $i\in[1,i^*)$ according to Fact~\ref{fact:cardinality:prev:I}.
Therefore,
$$ I_i = I_i^- = I_{i - 1}^- + e_i^-  \enspace.
$$
Adding it to $I_{i - 1}^-=I_{i - 1}$ by Fact~\ref{fact:cardinality:prev:I} and 
$e_i^-=e_i$ by Fact~\ref{fact:cardinality:prev:e},
for any $i\in[1,i^*)$ we have
$$
I_i = I_{i - 1} + e_i \enspace.
$$

Next, we show $\replacementTester{}(I_{i - 1}, e_i)=True$ for any $i\in[1,i^*)$. By the assumption of this lemma, we know that $\replacementTester{}(I_{i - 1}^-, e_i^-)=True$. In addition, for any $i\in[1,i^*)$, we have $I_{i - 1}=I_{i - 1}^-$ by Fact~\ref{fact:cardinality:prev:I} and $e_i=e_i^-$ by Fact~\ref{fact:cardinality:prev:e}. Thus, $\replacementTester{}(I_{i - 1}, e_i)=True$ for any $i\in[1,i^*)$.

\textbf{Survivor invariant (partially). }
We next show that the survivor invariant partially holds for the first $i^*$ levels.
To do this, we prove
$R_i = \{ e \in R_{i-1} - e_{i-1}: \replacementTester{}(I_{i-1}, e) = True\}$ holds for any $i\in[1,i^*]$.

Using Fact~\ref{fact:cardinality:prev:R},  we have $R_i = R_i^- \backslash \{v\}$ for each $i \in [1, i^*]$.
Also, the assumption of this lemma implies that 
$R_i^- = \{ e \in R_{i-1}^- - e_{i-1}^-: \replacementTester{}(I_{i-1}^-, e) = True\}$. Thus,
\begin{align*}
    R_i &= R_i^- \backslash \{v\} \\
    &= \{ e \in R_{i-1}^- - e_{i-1}^-: \replacementTester{}(I_{i-1}^-, e) = True\} \backslash \{v\} \\
    &= \{ e \in R_{i-1}^- \backslash \{v\} - e_{i-1}^-: \replacementTester{}(I_{i-1}^-, e) = True\} \enspace.
\end{align*}

Using $R_{i-1}^-=R_{i-1}$ by Fact~\ref{fact:cardinality:prev:R}, $e_{i-1}^-=e_{i-1}$ by Fact~\ref{fact:cardinality:prev:e}, and $I_{i-1}^-=I_{i-1}$ by Fact~\ref{fact:cardinality:prev:I}, for each $i\in[1,i^*]$ we have
$$
R_i = \{ e \in R_{i-1} - e_{i-1}: \replacementTester{}(I_{i-1}, e) = True\} \enspace.
$$

\textbf{Completing the proof. }
Above, we show that the starter, survivor, and cardinality invariants hold for the first $i^*$ levels.
To complete the proof of this lemma, we show all level invariants hold.
If $i^* \ne T + 1$, which means that we have invoked \constLevel$(i^*)$, our proof is complete by Theorem \ref{thm:cardinality:invariants:leveling}. 
Otherwise, if $i^* = T + 1$, we just need to show the terminator invariant holds to complete our proof, which means $R_{T + 1} = \emptyset$. 
Note that in this case, $T=T^-$.
Fact~\ref{fact:cardinality:prev:R} implies that $R_{T + 1}=R_{T + 1}^-$, and the assumption of lemma implies that $R_{T^- + 1}^- = \emptyset$.
Hence,
$$
R_{T + 1} = R_{T + 1}^- = R_{T^- + 1}^- = \emptyset \enspace,
$$
which completes the proof.
\end{proof}

\begin{lemma} [Uniform invariant]
\label{cardinality:mat_delete_uni}
If before the deletion of an element $v$, the level and uniform invariants hold, then the uniform invariant also holds after the execution of \deletev$(v)$. 
\end{lemma}

\begin{proof}
In other words, we want to prove that if for any $i$ and any element $e$ 
\[
    \Pr{\bE_i^- = e | \bT^- \geq i, \bH_i^- = H_i^-} = 
    \frac{1}{|R_i^{-}|} \cdot \ind{e \in R_i^{-}} \enspace ,
\]
then, after execution \deletev$(v)$, for each $i$ and each element $e$, we have 
\[
    \Pr{\bE_i = e | \bT \geq i, \bH_i = H_i} = 
    \frac{1}{|R_i|} \cdot \ind{e \in R_i} \enspace .
\]

Fix any arbitrary $i$ and $e$. 
We define a random variable $\bold{X_i}$ attaining values from the set $ \{0, 1, 2\}$, as follows:
    \begin{enumerate}
    \item If the execution of \deletev$(v)$ has terminated after invoking \constLevel$(j)$, then we set $\bold{X_i}$ to $2$.
    \item If the execution of \deletev$(v)$ has terminated in a level $L_{j \leq i}$ because $v \notin R^-_j$, then we set $\bold{X_i}$ to $1$.
    \item Otherwise, we set $\bold{X_i}$ to $0$. 
    That is, this case occurs if $v \in R_i^-$ and \deletev$(v)$ terminates because in a level $L_{j > i}$, 
    either $e_j=v$ or $v \notin R_j$.
\end{enumerate}

In Claims \ref{propos:cardinality:mat:delete1}, \ref{propos:cardinality:mat:delete2}, and \ref{cardinality:sevomi},
we show that for each value $X_i \in \{0, 1, 2\}$, 
$\Pr{\bE_i = e | \bT \geq i, \bH_i = H_i, \bold{X_i} = X_i} = \frac{1}{|R_i|} \cdot \ind{e \in R_i}$. 
This would imply the statement of our Lemma and completes the proof since 
\begin{equation*}
    \Pr{\bE_i = e | \bT \geq i, \bH_i = H_i}=
\Exu{X_i\sim \bold{X_i} }{\Pr{\bE_i = e | \bT \geq i, \bH_i = H_i, \bold{X_i}=X_i}}
\end{equation*}
by the law of total probability.

\begin{claim}
\label{propos:cardinality:mat:delete1}
$\Pr{\bE_i = e | \bT \geq i, \bH_i = H_i, \bold{X_i} = 0} = \frac{1}{|R_i|} \cdot \ind{e \in R_i}$.    
\end{claim}

\begin{proof}
First, we prove the following claim. 

\begin{claim}
\label{prop:cardinality:not:in:x0}
If $\bold{X_i} = 0$,  then for every $j < i$, $e_j \ne v$ and $v \notin R_i$. 
\end{claim}

\begin{proof}

Since $\bold{X_i}=0$, then $\constLevel{}(j)$ has not been invoked for any $j \le i$. 
Thus, $\bold{e_j^-} = \bold{e_j} = e_j$ for any $j < i$. 
However, if $e_j = v$ for a level index $j < i$, then $\bold{e_j^-} = v$ would have held for that $j < i$, 
which means that $\constLevel{}(j)$ would have been executed for that $j$ . 
This contradicts the assumption that $\bold{X_i}=0$. 
Therefore, for all $j < i$, we must have  $e_j \ne v$ proving the first part of this claim.  

Next, we prove the second part. 
Since $\bold{X_i} = 0$, the algorithm \deletev$(v)$ neither has called \constLevel{} nor it terminates its execution until level $L_i$. 
Thus, $\bold{R_i} = \bold{R_i^-} - v$, which implies that $v \notin \bold{R_i}$. 
However, if we had $v \in R_i$, then the event $[\bH_i = H_i, \bold{X_i = 0}]$ would have been impossible.
\end{proof}

Using Claim~\ref{prop:cardinality:not:in:x0}, 
we know that $e_j \ne v$ for $j < i$ and $v \notin R_i$. 
However, we also know that $v \in R_j^-$ for $j \le i$. 
Thus, we can define $H_i^- = (e_1^-, \dots, e_{i-1}^-, R_0^-  , \dots, R_i^- )$ 
based on $H_i = (e_1, \dots, e_{i-1}, R_0  , \dots, R_i)$ as follows:
\begin{align*}
    H_i^- =     (e_1, \dots, e_{i-1}, R_0 \cup \{v\} , \dots, R_i \cup \{v\}) \enspace. 
\end{align*}

\begin{claim}
    \label{pro:cardinality:events:eqn}
   Two events $[\bT \geq i, \bH_i = H_i, \bold{X_i} = 0]$ and $[\bT^- \geq i, \bH_i^- = H_i^-, \bE_i^- \neq v]$ are equivalent (i.e., they imply each other).
\end{claim}

\begin{proof}
We first prove that the event $[\bT \geq i, \bH_i = H_i, \bold{X_i} = 0]$ 
implies the event $[\bT^- \geq i, \bH_i^- = H_i^-,\bE_i^- \neq v]$. 
Indeed, since $\bold{X_i}=0 \ne 2$ we know that  
the algorithm $\constLevel{}(j)$ was not invoked for any $j \le i$ and 
the element $v$ was contained in $\bold{R_j^-}$ for all $j \le i$. 
In this case, according to the algorithm \deletev$(v)$, 
we conclude that for any $j \le i$, we have $\bold{e_j^-} \ne v $ and  $\bold{e_j^-} = \bold{e_j}$, and $\bold{R_j} = \bold{R_j^-} - v$.
This means that $\bold{R_j^-} = \bold{R_j} \cup \{v\}$. 
Therefore, since $\bH_i = H_i$, 
we must have  $\bH_i^- = H_i^-$, $\bE_i^- \neq v$, and  $\bE_i^-=\bE_i$.

Next, we prove the other way around. 
That is, the event $[\bT^- \geq i, \bH_i^- = H_i^-,\bE_i^- \neq v]$ 
implies the event $[\bT \geq i, \bH_i = H_i, \bold{X_i} = 0]$. 
Indeed, since $\bH_i^- = H_i^- = (e_1, \dots, e_{i-1}, R_0 \cup \{v\} , \dots, R_i \cup \{v\})$, 
then, for any $j \leq i$, $v \in \bold{R_j^-}$ and for any $j < i$, $\bold{e_j^-} = e_j$. 

Recall from Claim~\ref{prop:cardinality:not:in:x0} that 
for all $j < i$, $e_j \ne v$ and $v \notin R_i$. 
Thus, for any $j < i$, we know that $\bold{e_j^-} \ne v$. 
However, we also know that $\bE_i^- \neq v$. 
Thus, $\bold{e_j^-} \ne v$ for any $j \leq i$.
This essentially means that the algorithm \deletev$(v)$ neither invokes \constLevel{} 
nor terminates its execution till the level $L_i$. 
This implies that $\bold{X_i} = 0$. 
On the other hand, the algorithm \deletev$(v)$ only removes $v$ from $R_i^-$ and 
does not make any change in $\bold{e_1^-, \dots, e_i^-}$. 
Thus, $\bold{R_i} = R_i^- - \{v\}= R_i \cup \{v\} - v = R_i$ and 
$\bE_i = \bE_i^-$.  Therefore,  we have $\bH_i = H_i$. 
\end{proof}

Therefore, we have the following corollary.
\begin{corollary}
\label{cor:cardinality:eqn:events}
$\Pr{\bE_i = e | \bT \geq i, \bH_i = H_i, \bold{X_i} = 0} =  \Pr{\bE_i^- = e | \bT^- \geq i, \bH_i^- = H_i^-, \bE_i^- \neq v} $. 
\end{corollary}

Thus, in order to prove 
$\Pr{\bE_i = e | \bT \geq i, \bH_i = H_i, \bold{X_i} = 0} = \frac{1}{|R_i|} \cdot \ind{e \in R_i}$, 
we can prove 
\[
\Pr{\bE_i^- = e | \bT^- \geq i, \bH_i^- = H_i^-, \bE_i^- \neq v} = \frac{1}{|R_i|} \cdot \ind{e \in R_i} \enspace .
\]

Recall that the assumption of this lemma is  
$ \Pr{\bE_i^- = e | \bT^- \geq i, \bH_i^- = H_i^-} = 
  \frac{1}{|R_i^{-}|} \cdot \ind{e \in R_i^{-}} $. 
That is, conditioned on the event $[\bT^- \geq i, \bH_i^- = H_i^-]$, 
the random variable $\bE_i^- \sim U(R_i^-)$ is a uniform random variable over the set $R_i^-$. 
(i.e., the value $e_i$ of the random variable $\bE_i^-$ takes ones of the elements of the set $R_i^-$ uniformly at random.)  
However, since $X_i =0$ and using Claim~\ref{pro:cardinality:events:eqn}, we have $\bE_i^- \ne v$. 
Thus, conditioned on the event $ [\bT^- \geq i, \bH_i^- = H_i^-, \bE_i^- \ne v]$, 
we have that the random variable $\bE_i^- \sim U(R_i^- \backslash \{v\} ) = U(R_i)$ 
should be a uniform random variable over the set $R_i^- \backslash \{v\} = R_i$. 
Indeed, we have

\begin{align*}
\Pr{\bE_i^- = e | \bT^- \geq i, \bH_i^- = H_i^-, \bE_i^- \neq v} 
&= \dfrac {\Pr{\bE_i^- = e, \bE_i^- \neq v| \bT^- \geq i, \bH_i^- = H_i^-}}
    {\Pr{\bE_i^- \neq v| \bT^- \geq i, \bH_i^- = H_i^-}} 
= \dfrac{\frac{1}{|R_i^{-}|} \cdot \ind{e \in R_i^{-} \backslash \{v\}}}{1 - \frac{1}{|R_i^{-}|}} \\
&= \dfrac{1}{|R_i^{-}| - 1} \cdot \ind{e \in R_i^{-} \backslash \{v\}} 
= \dfrac{1}{|R_i|} \cdot \ind{e \in R_i} \enspace ,
\end{align*}
where the second equality holds because of our assumption that the uniform invariant holds before the deletion, and the fourth invariant holds because $R_i^- = R_i \cup \{v\}$ and $v \notin R_i$ proving the case $X=0$.
\end{proof}

\begin{claim}
\label{propos:cardinality:mat:delete2}
$\Pr{\bE_i = e | \bT \geq i, \bH_i = H_i, \bold{X_i} = 1} = \frac{1}{|R_i|} \cdot \ind{e \in R_i}$.
\end{claim}

\begin{proof}
We will be conditioning on possible values of $\bH_i^-$.  
\begin{align*}
    \Pr{\bE_i = e | \bT \geq i, \bH_i = H_i, X_i = 1} = \Exu{H_i^-}{\Pr{\bE_i = e | \bT \geq i, \bH_i = H_i, \bold{X_i} = 1, \bH_i^- = H_i^-}} \enspace ,
\end{align*}
where the expectation is taken over all $H_i$ for which $\Pr{\bT \geq i, \bH_i = H_i, \bold{X_i} = 1, \bH_i^- = H_i^-} > 0$.
For all such $H_i^-$, we claim that
this can be further rewritten as
$\Pr{\bT \geq i, \bH_i^- = H_i^-}$.
This is because \deletev$(v)$ is executed deterministically if it does not invoke 
the algorithm \constLevel{}. Furthermore, the value of $\bold{X_i}$ is
deterministically determined by $\bH_i^-$.
Therefore, for any value of $H_i^-$, either
$\bH_i^-=H_i^-$ implies $\bold{X_i}\ne 1$, in which case $\Pr{\bT \geq i, \bH_i^- = H_i^-, \bold{X_i} = 1} = 0$, which is in contradiction with our assumption, or 
$\bH_i^- = H_i^-$ imply
$\bold{X_i}=1$. Therefore, for all such $H_i^-$ implies $\bold{X_i}=1$, which also means that \constLevel{} never gets invoked, in which case $\bH_i$ is uniquely determined. Hence $\bH_i^- = H_i^-$ should also imply that $\bH_i = H_i$, as otherwise $\Pr{\bT \geq i, \bH_i^- = H_i^-, \bH_i = H_i} = 0$.
We therefore obtain: 
\begin{align*}
    \Pr{\bE_i = e | \bT \geq i, \bH_i = H_i, \bold{X_i} = 1, \bH_i^- = H_i^-} = \Pr{\bE_i = e | \bT \geq i, \bH_i^- = H_i^-} 
\end{align*}
as claimed.

Also, we know that $\bH_i = H_i, \bold{X_i} = 1$, implies that: 
\begin{align*}
    \bT^- = \bT, \enspace
    \bold{R_i^-} = \bold{R_i}, \enspace
    \bE_i^- = \bE_i, \enspace
\end{align*}
since it means that the execution of \deletev$(v)$ has terminated before level $i$, thus no change has been made for that level.
Therefore, for a $H_i^-$ used in our expectation, we know that $\bT \geq i, \bH_i^- = H_i^-$ also implies 
\begin{align*}
    \bT^- \geq i, \enspace
    R_i^- = R_i, \enspace
    \bE_i^- = \bE_i, \enspace
\end{align*}
we have:
\begin{align*}
    \Pr{\bE_i = e | \bT \geq i, \bH_i^- = H_i^-} = \Pr{\bE_i^- = e | \bT^- \geq i, \bH_i^- = H_i^-} = \frac{1}{|R_i^{-}|} \cdot \ind{e \in R_i^{-}} = \frac{1}{|R_i|} \cdot \ind{e \in R_i}, 
\end{align*}
where the third equality holds because of our assumption that the uniform invariant holds before the deletion of element $v$.
Therefore, $\Pr{\bE_i = e | \bT \geq i, \bH_i = H_i, \bold{X_i} = 1} = \frac{1}{|R_i|} \cdot \ind{e \in R_i}$.
    
\end{proof}

\begin{claim}
\label{cardinality:sevomi}
    $\Pr{\bE_i = e | \bT \geq i, \bH_i = H_i, \bold{X_i} = 2}  = \frac{1}{|R_i|}\cdot \ind{e \in R_i}$.
\end{claim}

\begin{proof}
By Lemma~\ref{cardinality:mat_delete_level}, we know that the level invariants hold at the end of the execution of \deletev{}, which is also the end of the execution of \constLevel{$(j)$}. 
Using Lemma~\ref{lm:cardinality:uniform:leveling}, we know that 
since the level invariants are going to hold after the execution of $\constLevel{}(j)$, 
for $i$ which is greater than $j$, we have: 
$$\Pr{\bE_i = e | \bT \geq i, \bH_i = H_i, \bold{X_i} = 2}  = \frac{1}{|R_i|}\cdot \ind{e \in R_i} \enspace, $$

which proves this claim. 
\end{proof}

\end{proof}

\subsection{Application of Uniform Invariant: Query complexity}
In terms of the query complexity of this algorithm, it's important to note that verifying whether an element $e$ is promoting for a level $L_z$ requires a single oracle query.
The binary search that we perform needs $O(\log T)$ number of such suitability checks for the element $e$. 
Thus, if we initiate the leveling algorithm with a set $R_i$, 
our algorithm needs $O(|R_i|\cdot \log(T))$ oracle queries to build the levels $L_i,\cdots, L_T$. 

\begin{lemma}
\label{lm:cardinality:number_of_levels}
  The number of levels $T$ is at most $k$.
\end{lemma}
\begin{proof}
  Given the starter and the cardinality invariants, we have $|I_0|=0$ and $|I_i|=|I_{i-1}|+1$.
  We can conclude by induction that $|I_{i}| = i$. 
  Since we have $I_T=I_{T-1}+e_T$ where $\suit{}(I_{T-1},e_T)=True$ by cardinality invariant, the element $e_T$ is promoting for level $L_{T-1}$. Therefore, $T-1<k$ which means $T<=k$.
\end{proof}

\begin{lemma}
\label{lm:cardinality:level_query_complexity}
  The query complexity of calling $\constLevel{}(i)$ is at most
  $\mO\left(\log\left(k\right)\cdot|R_i|\right)$.
\end{lemma}
\begin{proof}
  The algorithm $\constLevel{}(i)$ iterates over all elements in $R_i$.
  For each element $e$, it first calls the \replacementTester{} function, and select $e$ if it is a promoting element, i.e. $\replacementTester(I_{\ell-1}, e) = True$. In this case, we only need one query call for checking whether the element is promoting or not.
  However, if $e$ is not a promoting element, it reaches Line \ref{line:cardinality:const_level_matroid:binary_search} and runs the binary search on interval $[i,\ell-1]$.
  Based on Lemma~\ref{lm:cardinality:number_of_levels}, the length of this interval is  $\mO\left(k\right)$.
  Therefore, the number of steps in binary search is at most $\mO\left(\log\left(k\right)\right)$.
  In each step of the binary search, the algorithm calls $\replacementTester$ one time, which takes one query call.
  Thus, for each element, we need $\mO\left(\log\left(k\right)\right)$, and for all elements, we need $\mO\left(\log\left(k\right)\cdot|R_i|\right)$ query calls.
\end{proof}

\begin{lemma}
  For a specified value of $OPT$, the query complexity of each update operation in \carupdates~is at most
  $\mO\left(k\log\left(k\right)\right)$.
\end{lemma}
\begin{proof}
  Based on uniform invariant, when we insert/delete an element, for each natural number $i \le T$, we call $\constLevel{}(i)$ with probability $\dfrac{1}{|R_i|} \cdot \ind{e \in R_i}$ which is at most $\dfrac{1}{|R_i|}$. 
  Using Lemma~\ref{lm:cardinality:level_query_complexity}, the query complexity for calling $\constLevel{}(i)$ is $\mO\left(\log\left(k\right)\cdot|R_i|\right)$. Therefore, the expected number of queries caused by level $i$ is bounded by $\dfrac{1}{|R_i|} \cdot \mO\left(\log\left(k\right)\cdot|R_i|\right) = \mO\left(\log\left(k\right)\right)$.
  As the Lemma~\ref{lm:cardinality:number_of_levels} bounded the number of levels by $T = O\left(k\right)$, we calculate the expected number of query calls for each update by summing the expected number of query calls at each level:
  \begin{align*}
      \sum_{i=1}^T \mO\left(\log\left(k\right)\right) \le \mO\left(k\log\left(k\right)\right) \enspace .
  \end{align*}
\end{proof}

In order to obtain an algorithm that works regardless of the value of $OPT$, we guess $OPT$ up to a factor of $1+\epsilon$ using parallel runs.
Each element is inserted only to $\mO(\log(k)/\epsilon)$ copies of the algorithm. Therefore, we obtain the total query complexity claimed in Theorem \ref{thm:cardinality:query_complexity}.

\begin{theorem}
\label{thm:cardinality:query_complexity}
The expected query complexity of each insert/delete for all runs is $\mO\left(\frac{k}{\epsilon}\log^2\left(k\right)\right)$.
\end{theorem}

\subsection{Application of Level Invariants: Approximation guarantee}

Recall that we run parallel instances of our algorithm with different guesses of $OPT$.
In this section, we prove that if the level invariants hold, then after each update, there is a run such that the submodular value of the set $I_T$ in this run is a $(2+\epsilon)$-approximation 
of the optimal value. Formally, we state this claim as follows: 

\begin{theorem}
Suppose that the level invariants hold in every run of our algorithm.
Let $I_T$ be the selected set of the final level $L_T$ in each run. 
If $I^* \subseteq V$ is an optimal subset of size at most $k$ that achieves the optimal value, then there is a run such that the set $I_T$ in that run satisfies $(2 + \epsilon) \cdot f(I_T) \ge f(I^*)$.
\end{theorem}

\begin{proof}
First of all, as described earlier, we run parallel instances of our algorithm with different guesses of $OPT$ such that after each update, there is a run in which $f(I^*) \in (\frac{OPT}{1+\epsilon}, OPT]$.
Now, we show that in this run, $(2 + \epsilon) \cdot f(I_T) \ge f(I^*)$.

Since the terminator invariant holds, $R_{T+1}=\emptyset$, which means no element promotes level $L_T$.
Thus, either $T=k$ or $\marginalgain{e}{I_T}<\tau$ for every element $e \in V$. 
If $T = k$, then since the cardinality invariant holds, we have $f(I_T) = \Sigma_{i = 1}^{T} f(I_i) - f(I_{i-1}) = \Sigma_{i = 1}^{T} \marginalgain{e_i}{I_{i-1}} \geq k\tau = \frac{OPT}{2} \ge \frac{f(I^*)}{2}$. 

In the other case,  $\marginalgain{e}{I_T}<\tau$ for every element $e \in V$. 
By the submodularity and monotonicity of the function f, we have
$$
f(I^*) \leq f(I^* \cup I_{T}) \leq 
f(I_T) + \Sigma_{e \in I^* \setminus I_{T}} \marginalgain{e}{I_T} < f(I_T) + k\tau = f(I_T) + \frac{OPT}{2} \enspace .
$$
Given that in this run $\frac{OPT}{1+\epsilon} < f(I^*)$, we have $f(I^*) < f(I_T) + \frac{1+\epsilon}{2}\cdot f(I^*)$. Thus, we can conclude that $\frac{2}{1-\epsilon} \cdot f(I_T) > f(I^*)$.
\end{proof}

\section{Parameterized Lower Bound}
\label{sec:ower:bound}
In above, we presented our dynamic $0.5$-approximation algorithm that has an amortized query complexity of $\mO(k\log{k})$ 
if we know the optimal value of the sequence $\mathcal{S}$ after every insertion or deletion, and 
incurs an extra $\mO(\log(k)/\epsilon)$-factor in the case that we do not know the optimal value.
One may ask if we can obtain a dynamic algorithm for this problem that 
provides better than $0.5$-approximation factor having a query complexity that is still linear, or even polynomial in $k$.
Interestingly, we show there is no dynamic algorithm that maintains a $(0.5+\epsilon)$-approximate submodular solution of the sequence $\mathcal{S}$
using a query complexity that is an arbitrary function $g(k)$ of $k$ (e.g., not even doubly exponentially in $k$). 
This hardness holds even when we know the optimal value of the sequence after every insertion or deletion.   
Thus, the approximation ratio of our parameterized dynamic algorithm is tight.
We first state the lower bound due to Chen and Peng \cite[Theorem 1.1]{DBLP:journals/corr/abs-2111-03198} in 
the following lemma.

\begin{lemma}[Theorem 1.1 of \cite{DBLP:journals/corr/abs-2111-03198}]
\label{lem:chen&peng:klogk}
For any constant $\epsilon>0$, there is a constant $C_{\epsilon}>0$ with the following property. 
 When $k \ge C_{\epsilon}$, any randomized algorithm that achieves an approximation ratio of 
 $(0.5+\epsilon)$ for dynamic submodular maximization under cardinality constraint $k$ 
requires amortized query complexity $n^{\alpha_{\epsilon}}/k^3$, where
$\alpha_{\epsilon} = \tilde{\Omega}(\epsilon)$ and $n$ is the number of elements in $V$. 
\end{lemma}

Chen and Peng proved their theorem by considering a sequence that has the optimal value $1$ after every insertion or deletion. 
Building on this lower bound, we next prove the following theorem. 

\begin{theorem}
\label{thm:lowerbound:general}
Let $g: \mathbb{N} \to \mathbb{R}^+$ be an arbitrary function.
There is no randomized $(0.5+\epsilon)$-approximate algorithm for dynamic submodular maximization under cardinality constraint $k$ 
with an expected amortized query time of $g(k)$, 
even if the optimal value is known after every insertion/deletion.
\end{theorem}

\begin{proof}
Assume for the sake of contradiction, 
there exists a constant $\epsilon$, a function $g: \mathbb{N} \to \mathbb{R}^+$, and a $(0.5+\epsilon)$-approximation algorithm for dynamic submodular maximization with at most $g(k)$ amortized query per insertion/deletion.

According to Lemma~\ref{lem:chen&peng:klogk}, there is a constant $C_{\epsilon}>0$ such that for all $k>C_{\epsilon}$ and $n \geq k^{2/\epsilon}$, any $(0.5+\epsilon)$-approximation algorithm requires at least $n^{\alpha_{\epsilon}}/k^3$ amortized query. 
Let $k$ be an arbitrary natural number such that $k>C_{\epsilon}$, 
we define $n_{0} := \max((g(k)\cdot k^3)^{-\alpha_{\epsilon}}, k^{{2}/{\epsilon}})$. 
By the definition of $n_0$, we have $n_0 \ge (g(k)\cdot k^3)^{-\alpha_{\epsilon}}$, therefore $n_0^{\alpha_{\epsilon}} \ge g(k)\cdot k^3$, and then $n_{0}^{\alpha_{\epsilon}}/k^3 \ge g(k)$.
In conclusion, for any $n>n_{0}$, as $k^2 \le n^{\epsilon}$ constraint holds, Lemma~\ref{lem:chen&peng:klogk} implies that
the amortized query complexity is at least
$n^{\alpha_{\epsilon}}/k^3 > n_{0}^{\alpha_{\epsilon}}/k^3 \ge g(k)$, even if we know the optimal value.
Thus, the algorithm requires more than $g(k)$ amortized query complexity which is in contradiction with the assumption.

\end{proof}

As a result of Theorem \ref{thm:lowerbound:general}, for any $\epsilon>0$, even if $k$ is a constant, the required amortized query complexity to find a $0.5+\epsilon$-approximate solution increases by increasing $n$.
Therefore, even if we know the optimal value, it is not possible to find a parameterized algorithm with an approximation factor better than $0.5$ while query complexity is a function of only $k$; even if the query complexity is a double exponential of $k$. For example, if we are looking for an algorithm with amortized query complexity of $2^{2^k}$, when $n$ goes up enough, it is not possible to get a $(0.5+\epsilon)$-approximate solution for any $\epsilon>0$. 
Hence, the best approximation ratio of an algorithm that is parameterized by only $k$ is $0.5$, even with the assumption of knowing the optimal value. 
Surprisingly, the parameterized dynamic $0.5$-approximation algorithm that we presented above has expected amortized query complexity of $\mO(k \log k)$ if we know the optimal value. 


\end{document}